\documentclass[10pt,reqno]{amsart}

\usepackage{amsfonts,amsmath,amssymb,amsthm}
\usepackage{tikz}
\usepackage[margin=1in]{geometry}
\usepackage{enumerate, enumitem, multirow, array}
\usepackage{makecell} 
\usepackage{changepage}
\usepackage{threeparttable}

\newtheorem{thm}{Theorem}[section]
\newtheorem{prop}[thm]{Proposition}
\newtheorem{lemma}[thm]{Lemma}
\newtheorem{claim}[thm]{Claim}
\newtheorem{cor}[thm]{Corollary}
\newtheorem{obs}[thm]{Observation}
\newtheorem{rem}[thm]{Remark}
\newtheorem{defn}[thm]{Definition}
\newtheorem{fact}[thm]{Fact}
\newtheorem{example}[thm]{Example}

\newcommand{\set}[2]{\{#1_1,#1_2, \ldots, #1_{#2}\}} 
\newcommand{\ov}[1]{\overline{#1}}
\newcommand{\B}{\mathcal{B}}
\newcommand{\M}{\mathcal{M}}
\newcommand{\h}{\mathcal{H}}

\newtheorem*{SankAlg}{Sankoff's Algorithm}
\newtheorem*{FitAlg}{Fitch's Algorithm}

\usepackage[foot]{amsaddr}
\usepackage[backend=bibtex, style=authoryear]{biblatex}

\title[Computational complexity of calculating partition functions of medians]{The computational complexity of calculating partition \\functions of
optimal medians with Hamming distance}

\author{Istv\'an Mikl\'os$^{*,\dagger}$}

\author{Heather Smith$^\ddagger$}

\thanks{Email:  {\tt miklos.istvan@renyi.mta.hu} (Istv\'an Mikl\'os) {\tt heather.smith@math.gatech.edu} (Heather Smith)}

\address{$^*$ MTA R\'enyi Institute, Re\'altanoda u. 13-15, 1053 Budapest, Hungary}
\address{$^\dagger$ MTA SZTAKI, L\'agym\'anyosi u. 11, 1111 Budapest, Hungary}
\address{$^\ddagger$ School of Mathematics, Georgia Institute of Mathematics, Atlanta, GA 30332, USA}

\begin{document}

\begin{abstract}
In this paper, we show that calculating the partition function of optimal medians of binary strings with Hamming distance is \#P-complete for several weight functions. The case when the weight function is the factorial function has application in bioinformatics. In that case, the partition function counts the most parsimonious evolutionary scenarios on a star tree under several models in bioinformatics. The results are extended to binary trees and we show that it is also \#P-complete to calculate the most parsimonious evolutionary scenarios on an arbitrary binary tree under the substitution model of biological sequences and under the Single Cut-or-Join model for genome rearrangements.
\end{abstract}

\maketitle

 \section{The Partition Function}
 \label{sec:intro}
 
 For $n,m\in \mathbb{Z}^+$, fix a multiset of binary strings $B=\{\nu_1, \nu_2, \ldots, \nu_m\}$ where $\nu_i\in \{0,1\}^n$ for each $i\in [m]$. An optimal median $\mu$ is a binary string also in $\{0,1\}^n$ which minimizes $\sum_{i=1}^m H(\nu_i,\mu)$ where $H(\nu_i,\mu)$ is the Hamming distance between $\nu_i$ and $\mu$. Define $\mathcal{M}(B)$ to be the set of all optimal medians for the multiset $B$. 
 
Take $f$ to be a non-negative, real-valued function. In this paper, we examine the partition function 
\[Z(B,f(x)) = \sum_{\mu\in\mathcal{M}(B)} \prod_{i\in[m]} f(H(\nu_i, \mu)).\]
 
 The case when $f(x) := x!$ has application to phylogenetic trees and genome rearrangement.  Under the 
\textit{Single Cut-or-Join} \parencite{Feijao}  model for genome rearrangement, 
genomes are represented as edge labelled directed graphs forming paths and cycles, the direction of the edges along any path and cycle might vary. Such a graph can be encoded in binary strings, where each possible pair of edge endpoints (adjacency) is represented with one bit. The bit is 1 if the adjacency is presented in the genome and 0 otherwise. 
Note that although any genome can be represented with such a binary string, not all binary strings represent a genome since  two adjacencies might be in conflict if they share a common edge endpoint, and thus, their bits cannot be both 1.
A mutation is a bit flip: a flip from 0 to 1 represents a join, a flip from 1 to 0 represents a cut. Any flip from 1 to 0 is possible, however, flipping 0 to 1 is possible only if it does not cause a conflict.
Still, it can be proved that if two binary strings $\mu$ and $\nu$ represent two genomes $G_1$ and $G_2$ under the Single Cut-or-Join model, the fewest number of mutations to transform $G_1$ into $G_2$ is $H(\mu,\nu)$ \parencite{Feijao}. A \emph{scenario} is an ordering of the mutations necessary such that each intermediate string obtained from performing the cuts and joins one at a time represents a valid genome.  An upper bound on the number of scenarios is $H(\mu,\nu)!$. This upper bound is precisely the number of scenarios if there is no conflict in the presented adjacencies in genomes $G_1$ and $G_2$; that is, there is no constraint that some of the adjacencies first must be cut before some other adjacencies are created with joins. 

Next fix a multiset, $B$, of $m$ binary strings from $\{0,1\}^n$. If these strings label the leaves of a star tree $K_{1,m}$, the center of the star, or common ancestor, should be labeled with a median from $\mathcal{M}(B)$ which minimizes the number of mutations required, summed over all edges of the star tree. A \textit{most parsimonious scenario} for $K_{1,m}$ with $B$ labeling the leaves consists of a median $\mu$ from $\mathcal{M}(B)$ and a scenario transforming $\mu$ to $\nu_i$ for each $i\in[m]$ to label the edges of the star tree. The partition function $Z(B, x!)$ counts the number of most parsimonious scenarios if there is no conflict in the presented adjacencies. In this paper, we show that counting the most parsimonious scenarios is computationally hard already for these special cases.

Other bioinformatics models also use strings and changes of characters in them, for example when the strings represent biological sequences (DNA, RNA or protein) and changes of characters represent substitutions. We will refer to these models as \textit{substitution models} of biological sequences \parencite{Felsensteinbook}. The negative results represented in this paper also holds for these models, too.
 
 In Section~\ref{sec:complexity}, we establish some basics about computational complexity classes. Sections ~\ref{sec:set-up factorial} and ~\ref{sec:factorial} are devoted to computing the value of $Z(B, x!)$. In Section~\ref{sec:factorial approx}, we explore the possibility of stochastic approximations for $Z(B, x!)$. Then we turn our attention to the more general $Z(B,f(x))$ in Section~\ref{sec:general f(x)}.
When $\log f(x)$ is strictly concave up or strictly concave down, we obtain some further computational complexity results under mild restrictions. Section~\ref{sec:general approx} is devoted to stochastic approximations for $Z(B,f(x))$ when $\log f(x)$ is strictly concave down. In Sections ~\ref{sec:binary} and ~\ref{sec: bin result}, we extend our exploration of $Z(B,x!)$ from star trees to binary trees and define a similar partition function.

  
 \section{Computational complexity}
 \label{sec:complexity}
While P and NP are complexity classes for decision problems, the following classes are for counting problems. 
The classes \#P, \#P-hard, and \#P-complete were first defined by \textcite{valiant}. The definition for \#P that we give here, while not the original, is an equivalent definition. 
\begin{defn}[\cite{welsh}]
The class \emph{\#P} contains those functions $f: \Sigma^* \rightarrow \mathbb{N}$, for some alphabet $\Sigma$, such that both of the following hold:
\begin{itemize}
\item There is a polynomial $p$, a relation $R$, and a polynomial time algorithm which, for each input $w\in \Sigma^*$ and each $y\in \Sigma^*$ with $|y| \leq p(|w|)$, determines if $R(w,y)$. 
\item For any input $w$, $f(w) = |\{y: |y| \leq p(|w|) \text{ and } R(w,y) \}|.$
\end{itemize}
\label{sharpP}
\end{defn}

\begin{defn}[\cite{valiant}]
A counting problem is in \emph{\#P-hard} if there is a polynomial time reduction to it from every problem in \#P. A counting problem is in \emph{\#P-complete} if it is in \#P and is in \#P-hard.
\end{defn}
 
 Next we give a few known computational complexity results. To state these result, we establish some terminology. 

Conjunctive normal form (CNF) is a standard format in which to express Boolean formulas. A \emph{3CNF} is a Boolean formula $\Gamma$ which is the conjunction of clauses and each clause is the disjunction of 3 literals. 
A 3CNF, $\Gamma$, with $n$ variables $\set{v}{n}$ and $k$ clauses takes the form $\Gamma = c_1 \wedge c_2 \wedge \ldots \wedge c_k$ where each $c_i$ is a clause which is the disjunction of three literals and the literals are from $\{v_i\}_{i=1}^n \cup \{\overline{v_i}\}_{i=1}^n$. 
Because $\Gamma$ was said to have $n$ variables, we may assume that, for each $i\in [n]$, $v_i$ or $\overline{v_i}$ appears in some clause of $\Gamma$.
 Each $v_i$ is a \emph{positive literal} while each $\overline{v_i}$ is a \emph{negative literal}. The negative literal $\overline{v_i}$ is the negation of $v_i$. We identify $\ov{\ov{v_i}}$ with the literal $v_i$ and we refer to $\{v_i\}_{i=1}^m$ as the \emph{variables} of $\Gamma$ and always assume that the set of variables has an ordering.
 
 A \emph{truth assignment} for $\Gamma$ is a function $f: \{v_i\}_{i=1}^n \rightarrow \{T, F\}^n$ which assigns a value of true or false to each variable.  If a truth assignment makes $\Gamma$  true, we say it satisfies $\Gamma$. Otherwise, a truth assignment does not satisfy $\Gamma$ in which case there is at least one clause which is not satisfied.

\begin{defn}[3SAT]
Given an arbitrary $\Gamma$ in 3CNF with $n$ variables and $k$ clauses, decide if there is a truth assignment for $\Gamma$ which satisfies $\Gamma$.  
\end{defn}

\begin{defn}[\#3SAT]
Given an arbitrary Boolean formula $\Gamma$ in 3CNF with $n$ variables and $k$ clauses, count the number of truth assignments which satisfy $\Gamma$.
\end{defn}

\begin{thm} [\cite{cook}]
3SAT $\in$ NP-complete.
\end{thm}

\begin{thm} [\cite{valiant}]
\#3SAT $\in$ \#P-complete.
\end{thm}
  
Define \emph{D3CNF} to be the subset of 3CNF containing only those $\Gamma = \bigwedge_{i\in[k]} c_i$ such that for each $i\in [k]$, 
\begin{itemize}
\item $c_i$ contains three distinct literals, and
\item $c_i$ does not contain both $v_j$ and $\overline{v_j}$ for any $j\in[n]$.
\end{itemize} 
 This defines the following two problems. 

\begin{defn}[D3SAT]
For an arbitrary $\Gamma$ in D3CNF with $n$ variables and $k$ clauses, decide if there is a truth assignment which satisfies $\Gamma$. 
\end{defn}

\begin{defn}[\#D3SAT]
For an arbitrary $\Gamma$ in D3CNF with $n$ variables and $k$ clauses, count the number of truth assignments which satisfy $\Gamma$.
\label{def: sharpD3SAT}
\end{defn}

The following two results are proven through reductions from \#3SAT and 3SAT.

\begin{lemma}
\#D3SAT $\in$ \#P-complete.
\label{Lemma: Distinct}
\end{lemma}

\begin{proof}
This is a reduction from \#3SAT. Let $\Gamma$ be a 3CNF with $n$ variables and $k$ clauses, $n\geq 3$. 
Let $v_\alpha, v_\beta, v_\gamma$ be literals in $\Gamma$ with $\alpha\neq\beta\neq\gamma \neq \alpha$. 
Observe that each of the following pairs have the same satisfying truth assignments. 
\begin{align*}
 (v_\alpha \vee v_\beta\vee v_\beta) &\text{ and } (v_\alpha\vee v_\beta \vee v_\gamma) \wedge ( v_\alpha \vee v_\beta \vee \overline{ v_\gamma}). \\
 (v_\alpha \vee v_\alpha \vee v_\alpha) & \text{ and } (v_\alpha \vee v_\beta \vee v_\gamma) \wedge (v_\alpha \vee \overline{ v_\beta} \vee v_\gamma) \wedge (v_\alpha \vee v_\beta \vee \overline{ v_\gamma}) \wedge (v_\alpha \vee \overline{ v_\beta} \vee \overline{ v_\gamma}). 
 \end{align*}
Further, a clause of the form $(v_\alpha \vee \overline{v_\alpha} \vee v_\beta)$ is alway true, so it can be removed.  
 
 Making these replacements in $\Gamma$ will result in a D3CNF $\Gamma'$ with $n'$ variables ($n' \leq n$) and at most $4k$ clauses. Because some clauses like $(v_\alpha \vee \overline{v_\alpha} \vee v_\beta)$ are in $\Gamma$ but not in $\Gamma'$, it is possible that $n' <n$. 
 
Given a satisfying truth assignment for $\Gamma'$, we may extend it to a satisfying truth assignment for $\Gamma$ in $2^{n-n'}$ ways. This is because the variables in $\Gamma$ which are not in $\Gamma'$ do not affect the ability of a truth assignment to satisfy $\Gamma$. On the other hand, each satisfying truth assignment for $\Gamma$, restricted to the variables of $\Gamma'$, will be a satisfying truth assignment for $\Gamma'$. 
\end{proof}

\begin{lemma}
D3SAT $\in$ NP-complete.
\end{lemma}

\begin{proof}
As described in the last proof, for any 3CNF $\Gamma$, there is a corresponding D3CNF $\Gamma'$ which is computable in polynomial time such that $\Gamma'$ has at least one satisfying truth assignment exactly when $\Gamma$ has at least one satisfying truth assignment. 
\end{proof}

Next, we return our attention to the partition functions which were introduced in Section~\ref{sec:intro}.
The complexity results in this paper address subquestions and analogues of the following problems. 

\begin{defn}[\#SPS] Given a tree $T$ and a labeling $\varphi$ of the leaves of $T$ with binary strings, \#SPS (\emph{Small Parsimony Substitution}) asks for the exact number of most parsimonious scenarios where the scenario on an edge is an ordering of the substitutions that must take place to transform the median sequence into the sequence at the leaf.
\end{defn}

\begin{defn}[\#SPSCJ]
Given a tree $T$ and a labeling $\varphi$ of the leaves of $T$ with binary strings representing genomes under the SCJ model, \#SPSCJ (\emph{Small Parsimony Single Cut-or-Join}) asks for the exact number of most parsimonious scenarios where the scenario on an edge is an ordering of the cuts and joins that must take place to transform the median sequence into the sequence at the leaf such that the sequence produced after each cut or join represents a valid genome.
\label{def: SPSCJ}
\end{defn}

Clearly, \#SPS is a special case of \#SPSCJ, the case when there is no conflict in the  adjacencies present in the genomes assigned to the leaves of the evolutionary tree. While $Z(B,x!)$ is only an upper bound for an instance of \#SPSCJ, it is the exact answer to each instance of \#SPS. 

\begin{lemma}
The problem of calculating $Z(B,x!)$ is in \#P.
\label{lemma: sharpP}
\end{lemma}
\begin{proof}
The input includes a multiset  $B = \{\nu_i\}_{i=1}^m$ where each $\nu_i$ is from $\{0,1\}^n$. Viewing this in the sense of phylogenetic trees, a witness consists of a median $\mu$ to label the center of the star and a scenario to label each edge of the tree. The size of the input is $O(m n)$.

Recall that a median $\mu'$ minimizes the quantity $\sum_{i=1}^m H(\nu_i, \mu')$. We can find a single median $\mu'$ in $O(mn)$ time by examining the $k^{th}$ coordinate of each string in $B$ and making the $k^{th}$ coordinate of $\mu'$ the value that appears in a majority of the strings in $B$, breaking ties arbitrarily. To verify that the given binary string $\mu$ is indeed a median, we need only compare $\sum_{i=1}^m H(\nu_i, \mu)$ and $\sum_{i=1}^m H(\nu_i, \mu')$. If they are the same, then $\mu$ is a median. For each edge, we can verify that the given permutation is a scenario for that edge in $O(m)$ time by comparing the bits of $\mu$ and $\nu_i$.  
By Definition~\ref{sharpP},  \#SPS is in \#P.
\end{proof}

Most of the complexity results are reductions from \#D3SAT. In other words, given a D3CNF $\Gamma$ with $n$ variables and $k$ clauses, we create a multiset of $m$ binary strings of length $2n+t$ (where $t$ and $m$ are polynomials of $n$ and $k$) to label the leaves of the tree. These strings will be chosen so that the number of most parsimonious substitution scenarios is related to the number of satisfying truth assignments for $\Gamma$. 

\section{Set-up for results on star trees}
\label{sec:set-up factorial}

The discussion in this section and the next is specific to $Z(B,x!)$. The first section details some tools and constructions that will be needed for the proof of our main result in Section~\ref{sec:factorial}. This main result, Theorem~\ref{thm: number SCJ scenarios}, states that computing $Z(B,x!)$ is \#P-complete.

The proof of Theorem~\ref{thm: number SCJ scenarios} will define a polynomial reduction from \#D3SAT (Definition~\ref{def: sharpD3SAT}) to compute $Z(B,x!)$. Fix an arbitrary D3CNF, $\Gamma$, with $n$ variables and $k$ clauses. Fix a prime $p\leq 5\max\{300,n+5\}$ which will be utilized later. 

Our task is to define a multiset of binary strings $\mathcal{D}(p)$ to encode $\Gamma$. The multiset $\mathcal{D}(p)$ will be chosen so that the set of medians $\mathcal{M}(\mathcal{D}(p))$ will have a list of desired characteristics. First, each string in $\mathcal{D}(p)$ and each median in $\mathcal{M}(\mathcal{D}(p))$ will have length $2n+t$ with coordinates 
\[(x_1, y_1, x_2, y_2, \ldots, x_n, y_n, e_1, e_2,\ldots, e_t)\]
where $n$ is the number of variables in $\Gamma$ and the $t$ is a polynomial of $n$ and $k$ which will be defined later.
Second, $\mathcal{M}(\mathcal{D}(p))$ will be the set of all binary strings $\mu$ of length $2n+t$ that have $\mu[e_i]=0$ for each $i\in[t]$. In other words, $\mathcal{D}(p)$ will be defined so that $\mathcal{M}(\mathcal{D}(p))$ equals $\{0,1\}^{2n} \times \{0\}^t$. 
Let $\mathcal{M}'(\mathcal{D}(p))$ denote the subset of $\mathcal{M}(\mathcal{D}(p))$ with the additional property that $\mu[x_i] \neq \mu[y_i]$ for all $i\in [n]$. Once we have established that $\mathcal{M}(\mathcal{D}(p)) =\{0,1\}^{2n} \times \{0\}^t$, we can conclude $\mathcal{M}'(\mathcal{D}(p))=\{01,10\}^n \times \{0\}^t$. This allows for a connection with truth assignments for $\Gamma$. 

\begin{defn}
Let $n\in \mathbb{Z}^+$. For arbitrary $\Gamma$ in D3CNF with $n$ variables, let $S$ be a multiset of binary strings on the coordinates $(x_1, y_1,  \ldots, x_n, y_n, e_1, \ldots, e_t).$
There is an injective function $f$ which assigns to each median $\mu\in\mathcal{M}'(S)$ a truth assignment for $\Gamma$. In particular, $f(\mu)$ will assign a value of true to the $i^{th}$ variable of $\Gamma$ if $\mu[x_i]=1$ and false if $\mu[x_i]=0$. 
\label{defi: bijection}
\end{defn}

\begin{rem}
If multiset $S$ is chosen so that $\mathcal{M}'(S) = \{01,10\}^n \times \{0\}^t$, then Definition~\ref{defi: bijection} provides a bijection between $\mathcal{M}'(S)$ and the truth assignments for $\Gamma$. 
\label{defi: bijection2}
\end{rem}

\begin{defn}
Let $n\in \mathbb{Z}^+$. Given an arbitrary D3CNF, $\Gamma$, with $n$ variables, let $S$ be an arbitrary multiset of binary strings  on the coordinates $(x_1, y_1,  \ldots, x_n, y_n, e_1, \ldots, e_t)$ for some $t\in \mathbb{Z}^+$. 
Define $\mathcal{M}'_{\Gamma}(S)$ to be a subset of $\mathcal{M}'(S)$, containing only those medians which, through the bijection in Definition~\ref{defi: bijection}, correspond to a satisfying truth assignment for $\Gamma$. Since a single clause $c$ in $\Gamma$ is also a D3CNF, this defines $\mathcal{M}'_{c}(S)$ as well.
\label{defi: MGamma}
\end{defn}

To calculate
\[Z(\mathcal{D}(p), x!)=\sum_{\mu\in \mathcal{M}(\mathcal{D}(p))} \prod_{i\in [m]} H(\mu, \nu_i)!,\]
 we first calculate $\prod_{i\in [m]} H(\mu, \nu_i)!$ for each median $\mu\in\mathcal{M}(\mathcal{D}(p)).$ 
The multiset $\mathcal{D}(p)$ will be constructed so that there is a constant $K(p)$ (specified in Claim~\ref{all props}) which is not a multiple of $p$ and such that for any $\mu\in\mathcal{M}'_{\Gamma}(\mathcal{D}(p))$ 
$\prod_{i\in [m]} H(\mu, \nu_i)! = K(p)$. Each median $\mu \in \mathcal{M}(\mathcal{D}(p)) \setminus \mathcal{M}'_{\Gamma}(\mathcal{D}(p))$ will have $\prod_{i\in [m]} H(\mu, \nu_i)! \equiv 0 \mod p$. As a result
\[\sum_{\mu\in \mathcal{M}(\mathcal{D}(p))} \prod_{i\in [m]} H(\mu, \nu_i)! \equiv 
|\mathcal{M}'_{\Gamma}(\mathcal{D}(p))| K(p) \!\! \mod p.\]
Repeating this construction for sufficiently many primes $p\leq 5 \max\{ 300, n+5\}$, we obtain enough congruences, which together with the knowledge that there are at most $2^n$ satisfying truth assignment for $\Gamma$, uniquely determine the size of $\mathcal{M}'_{\Gamma}(\mathcal{D}(p))$ which is equal to the number of satisfying truth assignments for $\Gamma$.

Later we will see that the main work goes into developing a multiset $\mathcal{D}(p)$ with the property that for any $\mu\in \mathcal{M}'_{\Gamma}(\mathcal{D}(p))$ and any $\mu' \in \mathcal{M}'(\mathcal{D}(p))\setminus \mathcal{M}'_{\Gamma}(\mathcal{D}(p))$ have \[\prod_{i\in [m]} H(\mu, \nu_i)! \neq \prod_{i\in [m]} H(\mu', \nu_i)!.\] In Section~\ref{section: define tables}, we define the strings $\mathcal{D}(p)$ which are used in the proofs to distinguish medians in $\mathcal{M}'_{\Gamma}(\mathcal{D}(p))$ from medians in $\mathcal{M}'(\mathcal{D}(p))\setminus \mathcal{M}'_{\Gamma}(\mathcal{D}(p))$. 

\subsection{Encoding Boolean clauses in binary strings.}\label{section: define tables}
A truth assignment satisfies $\Gamma$ if and only if it satisfies every clause in $\Gamma$. Hence, we will encode each clause $c_i$ of $\Gamma$ in a set of 50 strings 
\[\mathcal{C}_i:=\{\nu_1^i, \nu_2^i, \ldots, \nu_{50}^i\}\]
 which will be defined through Table~\ref{table: n+5}. These 50 strings are designed to distinguish those medians in $\mathcal{M}'_{c_i}(\mathcal{C}_i)$ from those in $\mathcal{M}'(\mathcal{C}_i)\setminus \mathcal{M}'_{c_i}(\mathcal{C}_i)$. Confirmation of this will come in Section~\ref{distinguish}. Because every truth assignment that does not satisfy $\Gamma$ has at least one clause in $\Gamma$ that is does not satisfy, we will see that the disjoint union $\biguplus_{i\in[k]} \mathcal{C}_i$ distinguishes between  $\mathcal{M}'_{\Gamma}\left(\biguplus_{i\in[k]} \mathcal{C}_i\right)$ and $\mathcal{M}'\left(\biguplus_{i\in[k]} \mathcal{C}_i\right)\setminus \mathcal{M}'_{\Gamma}\left(\biguplus_{i\in[k]} \mathcal{C}_i\right)$. 

The following definition gives a guide for defining a multiset of binary strings. 
\begin{defn}[Defining strings]
For arbitrary $m,n\in \mathbb{Z}^+$ and $t\in \mathbb{Z}^{+} \cup \{0\}$, to define a multiset of binary strings $\{\eta_1, \eta_2, \ldots, \eta_m\}$ on coordinates $(x_1, y_1, \ldots, x_n, y_n, e_1, \ldots, e_t)$, it suffices to
\begin{itemize}
\item define $\eta_j[x_\ell]$ and $\eta_j[y_\ell]$ for each $j\in[m]$ and $\ell \in [n]$, and 
\item define a function $e:[m] \rightarrow \mathbb{Z}^{+} \cup \{0\}$.
\end{itemize}
 We say $\eta_j$ has $e(j)$ additional ones. In order to infer the values $\eta_j[e_\ell]$ for each $j\in[m]$ and $\ell\in[t]$, follow this procedure: 

\begin{adjustwidth}{.5in}{}
Partition $[t]$ into subsets $E, E_1, E_2, \ldots, E_{m}$ so that the size of $E_{j}$ is precisely $e(j)$,  and ${E}=[t]\setminus \bigcup_{j\in[m]} E_j$. For each $j\in[m]$ and each $\ell\in [t]$, set $\eta_j[e_\ell]=1$ if and only if $\ell\in E_j$ and $\eta_{j'}[e^\ell]=0$ otherwise.
\end{adjustwidth}
\label{defi: ones}
\end{defn}

\begin{rem}
Let $m\in \mathbb{Z}^+$. For an arbitrary multiset $\{\eta_j\}_{j=1}^m$ of binary strings built using Definition~\ref{defi: ones}, 
 for each $\ell\in [t]$, there is a unique $j\in [m]$ such that $\eta_j[e_\ell]=1$. Consequently, each $\mu\in \mathcal{M}(\{\eta_j\}_{j=1}^m)$ will have $\mu[e_\ell]=0$ for all $\ell\in[t]$ because $\mu$ must minimize $\sum_{j\in[m]} H(\eta_j, \mu)$. 
 \label{fact: ones}
\end{rem}

\begin{defn}
Let $n\in \mathbb{Z}^+$ and $t\in \mathbb{Z}^{+} \cup \{0\}$ be arbitrary. Two binary strings $\eta$ and $\ov{\eta}$ with  coordinates $(x_1, y_1, \ldots, x_n, y_n, e_1, \ldots, e_{t})$, are said to be \emph{complementary on the first $2n$ coordinates} if $\eta[x_i]=1-\overline{\eta}[x_i]$ and $\eta[y_i]=1-\overline{\eta}[y_i]$ for each $i\in [n]$. 
\label{complements}
\end{defn}

The following fact will be useful. 
\begin{fact}
Let $\eta$  and $\ov{\eta}$ be binary strings on coordinates \[(x_1, y_1, \ldots, x_n, y_n, e_1, \ldots, e_t).\] Set $e(\eta):=\sum_{i\in[t]} \eta[e_i]$, the number of additional ones in $\eta$. Define $e(\ov{\eta})$ similarly. If $\eta$ and $\ov{\eta}$ are complementary on the first $2n$ coordinates, then for any $\mu\in \{0,1\}^{2n} \times \{0\}^t$, 
\[H(\mu,\eta) + H(\mu, \overline\eta) = 2n + e(\eta) + e(\overline\eta).\]
\vspace{-.3in}
\label{fact: Hcomp}
\end{fact}
\begin{proof}
For each $i\in [n]$, either $\mu[x_i] = \eta[x_i]$ or $\mu[x_i]=\overline\eta[x_i]$, but not both. This is also true for each $y_i$. This accounts for the $2n$ in the sum. Because $\mu[e_i]=0$ for all $i\in [t]$, each $i\in [t]$ with $\eta[e_i]=1$ will contribute one to the sum. Also each $i\in [t]$ with $\overline\eta[e_i]$ will contribute one to the sum. This completes the proof.
\end{proof}

\begin{defn}
Given an arbitrary multiset of binary strings $S$, we say that a coordinate $s$ is \emph{ambiguous} if there are exactly $\frac{1}{2}|S|$ binary strings $\eta\in S$, counted with multiplicity, such that $\eta[s]=0$.
Consequently, if you change the value of a median at an ambiguous coordinate, you obtain another median. Note that if $|S|$ is odd, then there are no ambiguous coordinates and there is exactly one median.
\label{defn: ambig}
\end{defn}

\begin{fact}
Let $S$ be a multiset of binary strings which are defined on the  coordinates \[(x_1, y_1, \ldots, x_n, y_n, e_1, \ldots, e_t).\] If $S$ can be partitioned into pairs of strings where the two strings in a pair are complementary on the first $2n$ coordinates, then each $x_i$ and each $y_i$ is an ambiguous coordinate.
\label{fact: comp ambig}
\end{fact}

Fix an arbitrary D3CNF, $\Gamma$, with $n$ variables and $k$ clauses. Fix a clause $c_i$ in $\Gamma$. For this clause, we are now ready to define a set of 50 strings \[\mathcal{C}_i=\{\nu_1^i, \nu_2^i, \ldots, \nu_{50}^i\}.\] First assume that $c_i=v_\alpha \vee v_\beta \vee v_\gamma$, a disjunction of three positive literals. Because $\Gamma$ is a D3CNF, we may assume $\alpha<\beta<\gamma$.

 For each $j\in [50]$, we will supply the following three pieces of information for $\nu_j^i$: 
\begin{enumerate}
\item[(a)] The values for $\nu_j^i[x_\alpha], \nu_j^i[y_\alpha], \nu_j^i[x_\beta], \nu_j^i[y_\beta], \nu_j^i[x_\gamma], \nu_j^i[y_\gamma]$ will be explicitly defined. 
\item[(b)]  A constant $\kappa_{ij}\in \{0,1\}$ will be given so that $\nu_j^i[x_\ell]=\nu_j^i[y_{\ell'}]=\kappa_{ij}$ for all $\ell, \ell'\in [n]\setminus\{\alpha, \beta, \gamma\}$.
\item[(c)] The string will be assigned some number of additional ones. 
\end{enumerate}
By Definition~\ref{defi: ones}, this is sufficient to explicitly define $\nu_j^i$. 

In Table~\ref{table: n+5}, there is a row for each string in $\mathcal{C}_i$. The three defining pieces of information are found in Columns A, B, and C respectively. The remainder of the table will be explained in Subsection~\ref{table explanation}. 

For each $j\in [50]$, row $j$ of Table~\ref{table: n+5} supplies the three ingredients needed to define $\nu_j^i$. By matching the 6-bit string in Column A of row $j$ with
\[(\nu_j^i[x_\alpha], \nu_j^i[y_\alpha], \nu_j^i[x_\beta], \nu_j^i[y_\beta], \nu_j^i[x_\gamma], \nu_j^i[y_\gamma])\] we obtain the 6 values for (a).  
The constant $\kappa_{ij}$ for (b) is found in Column B of row $j$. For (c), the number of additional ones in $\nu_j^i$ is found in Column C of row $j$.

\begin{table}[ht!]
\centering
\caption{The 50 strings in $\mathcal{C}_i$ for a single clause $c_i$ along with their Hamming distance from medians in $\mathcal{M}'$.}
\begin{tabular}{c}
\resizebox{5.5in}{!}{
$
\begin{array}{c|c|c|c|c|||c|c|c|c|c|c|c|c|c|}
\cline{3-13}
\multicolumn{2}{c|}{} & \textbf{A} & \textbf{B} & \textbf{C} & \textbf{M1} &\textbf{M2}&\textbf{M3}&\textbf{M4}&\textbf{M5}&\textbf{M6}&\textbf{M7}&\textbf{M8}\\ 
\hline
&&\text{Values of}& \nu^i_j[x_\ell],&& 10\,10\,10& 10\,10\,01&10\,01\,10&01\,10\,10&10\, 01\, 01&01\,10\,01&01\,01\,10&01\, 01\, 01 \\
&\text{Row}&\nu^i_j \text{ on its} & \nu^i_j[y_\ell]&\text{Add'l} &&&&&&&&\\
&\text{\#}&\text{support set}& (v_\ell\not\in c_i)& \text{Ones} &&&&&&&&\\
\hline \hline \hline
\multirow{3}{*}{\rotatebox[origin=c]{90}{$\mathcal{N}_1^{(+3)}$}}
&1&01 \,00 \,00 &0&+3&n+4&n+4&n+4&n+2&n+4&n+2&n+2&n+2\\
&2&00 \,01 \,00 &0&+3&n+4&n+4&n+2&n+4&n+2&n+4&n+2&n+2\\
&3&00 \,00 \,01 &0&+3&n+4&n+2&n+4&n+4&n+2&n+2&n+4&n+2\\

\hline
\hline
\multirow{3}{*}{\rotatebox[origin=c]{90}{$\overline{\mathcal{N}_1}^{(+0)}$}}
&4& 10\,11\,11 &1&+0&n-1&n-1&n-1&n+1&n-1&n+1&n+1&n+1\\
&5&11\,10\,11 &1&+0&n-1&n-1&n+1&n-1&n+1&n-1&n+1&n+1\\
&6&11\,11\,10 &1&+0&n-1&n+1&n-1&n-1&n+1&n+1&n-1&n+1\\

\hline
\hline

\multirow{21}{*}{\rotatebox[origin=c]{90}{${\mathcal{I}_2}^{(+2)} \setminus {\mathcal{N}_2}^{(+2)}$}}
&7&10\,10\,00 &0&+2&n&n&n+2&n+2&n+2&n+2&n+4&n+4\\
&8&10\,00\,10 &0&+2&n&n+2&n&n+2&n+2&n+4&n+2&n+4\\
&9&00\,10\,10 &0&+2&n&n+2&n+2&n&n+4&n+2&n+2&n+4\\
\cline{2-13} 
&10&10\,10\,00 &0&+2&n&n&n+2&n+2&n+2&n+2&n+4&n+4\\
&11&10\,00\,01 &0&+2&n+2&n&n+2&n+4&n&n+2&n+4&n+2\\
&12&00\,10\,01 &0&+2&n+2&n&n+4&n+2&n+2&n&n+4&n+2\\
\cline{2-13}
&13&10\,01\,00 &0&+2&n+2&n+2&n&n+4&n&n+4&n+2&n+2\\
&14&10\,00\,10 &0&+2&n&n+2&n&n+2&n+2&n+4&n+2&n+4\\
&15&00\,01\,10 &0&+2& 5&n+4&n&n+2&n+2&n+4&n&n+2\\
\cline{2-13}
&16&01\,10\,00 &0&+2&n+2&n+2&n+4&n&n+4&n&n+2&n+2\\
&17&01\,00\,10 &0&+2&n+2&n+4&n+2&n&n+4&n+2&n&n+2\\
&18&00\,10\,10 &0&+2&n&n+2&n+2&n&n+4&n+2&n+2&n+4\\
\cline{2-13}
&19&10\,01\,00 &0&+2&n+2&n+2&n&n+4&n&n+4&n+2&n+2\\
&20&10\,00\,01 &0&+2&n+2&n&n+2&n+4&n&n+2&n+4&n+2\\
&21&00\,01\,01 &0&+2&n+4&n+4&n+2&n+4&n&n+2&n+2&n\\
\cline{2-13}
&22&01\,10\,00 &0&+2&n+2&n+2&n+4&n&n+4&n&n+2&n+2\\
&23&01\,00\,01 &0&+2&n+4&n+2&n+4&n+2&n+2&n&n+2&n\\
&24&00\,10\,01 &0&+2&n+2&n&n+4&n+2&n+2&n&n+4&n+2\\
\cline{2-13}
&25&01\,01\,00 &0&+2&n+4&n+4&n+2&n+2&n+2&n+2&n&n\\
&26&01\,00\,10 &0&+2&n+2&n+4&n+2&n&n+4&n+2&n&n+2\\
&27&00\,01\,10 &0&+2&n+2&n+4&n&n+2&n+2&n+4&n&n+2\\ 

\hline
\hline

\multirow{21}{*}{\rotatebox[origin=c]{90}{$\overline{\mathcal{I}_2}^{(+1)} \setminus \overline{\mathcal{N}_2}^{(+1)}$ }} 
&28&10\,10\,11 &1&+1&n-1&n-1&n+1&n+1&n+1&n+1&n+3&n+3\\
&29&10\,11\,01 &1&+1&n+1&n-1&n+1&n+3&n-1&n+1&n+3&n+1\\
&30&11\,10\,01 &1&+1&n+1&n-1&n+3&n+1&n+1&n-1&n+3&n+1\\
\cline{2-13}
&31&10\,01\,11 &1&+1&n+1&n+1&n-1&n+3&n-1&n+3&n+1&n+1\\
&32&10\,11\,10 &1&+1&n-1&n+1&n-1&n+1&n+1&n+3&n+1&n+3\\
&33&11\,01\,10 &1&+1&n+1&n+3&n-1&n+1&n+1&n+3&n-1&n+1\\
\cline{2-13}
&34&01\,10\,11 &1&+1&n+1&n+1&n+3&n-1&n+3&n-1&n+1&n+1\\
&35&01\,11\,10 &1&+1&n+1&n+3&n+1&n-1&n+3&n+1&n-1&n+1\\
&36&11\,10\,10 &1&+1&n-1&n+1&n+1&n-1&n+3&n+1&n+1&n+3\\
\cline{2-13}
&37&10\,01\,11 &1&+1&n+1&n+1&n-1&n+3&n-1&n+3&n+1&n+1\\
&38&10\,11\,01 &1&+1&n+1&n-1&n+1&n+3&n-1&n+1&n+3&n+1\\
&39&11\,01\,01 &1&+1&n+3&n+1&n+1&n+3&n-1&n+1&n+1&n-1\\
\cline{2-13}
&40&01\,10\,11 &1&+1&n+1&n+1&n+3&n-1&n+3&n-1&n+1&n+1\\
&41&01\,11\,01 &1&+1&n+3&n+1&n+3&n+1&n+1&n-1&n+1&n-1\\
&42&11\,10\,01 &1&+1&n+1&n-1&n+3&n+1&n+1&n-1&n+3&n+1\\
\cline{2-13}
&43&01\,01\,11 &1&+1&n+3&n+3&n+1&n+1&n+1&n+1&n-1&n-1\\
&44&01\,11\,10 &1&+1&n+1&n+3&n+1&n-1&n+3&n+1&n-1&n+1\\
&45&11\,01\,10 &1&+1&n+1&n+3&n-1&n+1&n+1&n+3&n-1&n+1\\
\cline{2-13}
&46&01\,01\,11 &1&+1&n+3&n+3&n+1&n+1&n+1&n+1&n-1&n-1\\
&47&01\,11\,01 &1&+1&n+3&n+1&n+3&n+1&n+1&n-1&n+1&n-1\\
&48&11\,01\,01 &1&+1&n+3&n+1&n+1&n+3&n-1&n+1&n+1&n-1\\

\hline						
\hline

\rotatebox[origin=c]{90}{$\mathcal{N}_3^{(+1)}$}
&49&01\,01\,01 &0&+1&n+3&n+2&n+2&n+4&n&n&n&n-2\\	
\hline \hline

\rotatebox[origin=c]{90}{$\overline{\mathcal{ N}_3}^{(+2)}$}
&50&10\,10\,10 &1&+2&n-1&n+1&n+1&n+1&n+3&n+3&n+3&n-3\\	
\hline

\end{array}$
}
\end{tabular}

\begin{tablenotes}[flushleft]
\item[] \tiny{For a clause $c_i$, the left three columns define the 50 strings in $\mathcal{C}_i$. In row $j$, the 6-bit string gives the values of $\nu_j^i$ on the support set $S_i$ as described by Table~\ref{table: S_i}. The second column gives the constant value to be assigned to all $x_\ell$ and $y_\ell$ which are not in $S_i$. The third column specifies the number of extra ones in $\nu_j^i$. The collection $\{01,10\}^3$ is listed along the top row. The entry in row $j$ and column $\ell$ is the number of additional ones in $\nu_j^i$ added to the Hamming distance between the 6-bit string in row $j$ and the 6-bit string at the top of column $\ell$.}
\end{tablenotes}
\label{table: n+5}
\end{table}

With a slight modification in the reading of Column A, the 50 rows of Table~\ref{table: n+5} will also supply the 50 strings for a clause which contains negative literals. Fix an arbitrary clause $c_i$ in $\Gamma$ which now may have negative literals. For each $j\in[50]$, the definition of string $\nu_j^i$ will again be based on Columns  A, B, C of row $j$ in Table~\ref{table: n+5} where the same information will be gleaned from Columns B and C. The only difference is with Column A which will be explained next. 

If $c_i$ contains the variables $v_\alpha, v_\beta, v_\gamma$ where some of these may be present as negative literals, set $S_i:=\{x_\alpha,y_\alpha,x_\beta,y_\beta,x_\gamma,y_\gamma\}$. We call $S_i$ the \emph{support set} of $c_i$.
Clause $c_i$ must be one of the 8 clauses listed in Column A of Table~\ref{table: S_i}. For $j\in [50]$, $\nu_j^i$ is defined on the coordinates $S_i$ by matching the entry in the right column of the $c_i$ row of Table~\ref{table: S_i} with the 6-bit string in Column A of the $j^{th}$ row of Table~\ref{table: n+5}. 

\begin{example}
For an example, when $c_i=v_\alpha \vee \overline{v_\beta} \vee \overline{v_\gamma}$, the last row of Table~\ref{table: n+5} says that the string $\nu_{50}^i$ must have 
\[(\nu_{50}^i[x_\alpha], \nu_{50}^i[y_\alpha], \nu_{50}^i[y_\beta], \nu_{50}^i[x_\beta], \nu_{50}^i[y_\gamma], \nu_{50}^i[x_\gamma]) = (101010).\]
Therefore, $\nu_{50}^i[x_\alpha]=1$, $\nu_{50}^i[y_\alpha]=0$, $\nu_{50}^i[x_\beta]=0$, $\nu_{50}^i[y_\beta]=1$, $\nu_{50}^i[x_\gamma]=0$, and $\nu_{50}^i[y_\gamma]=1$. Further, Column B implies $\nu_{50}^i(x_\ell)=\nu_{50}^i(y_\ell)=1$ for all $\ell\in [n]\setminus\{\alpha, \beta, \gamma\}$ and, from Column C, $\nu_{50}^i$ will have 2 additional ones. 
\end{example}

Now that we have defined $\mathcal{C}_i$ for any clause $c_i$, let us analyze $\mathcal{M}(\mathcal{C}_i)$.
By Fact~\ref{fact: ones}, for every $\mu\in \mathcal{M}(\mathcal{C}_i)$ and $\ell \in [t]$, $\mu[e_\ell]=0$.

 In Column B of Table~\ref{table: n+5}, it is evident that for any $\ell \in [n]\setminus \{\alpha, \beta, \gamma\}$, the number of strings $\nu_j^i$ with $\nu_j^i[x_\ell]=0$ is $25=\frac{1}{2}|\mathcal{C}_i|$. Therefore, by Definition~\ref{defn: ambig}, the coordinates $x_\ell$ and $y_\ell$ are {ambiguous}. Through careful inspection of the strings in Column A of Table~\ref{table: n+5}, we see that coordinates $x_{\ell'}$ and $y_{\ell'}$ are also ambiguous for each $\ell'\in\{\alpha, \beta, \gamma\}$. Therefore we have proven the following fact, which was one of our goals: 

 \begin{fact}
 For an arbitrary clause $c_i$ with three distinct variables,
 \[\mathcal{M}(\mathcal{C}_i)= \{0,1\}^{2n} \times \{0\}^t.\]
 \label{fact: medians}
\end{fact}
\vspace{-0.3in}

\begin{rem}
By visual inspection of Table~\ref{table: n+5}, the binary strings $\mathcal{C}_i$ can be partitioned into pairs where the two strings in a pair are complementary on the first $2n$ coordinates.
\label{table pairs} 
\end{rem}

\begin{table}[h!]
\centering
\begin{threeparttable}
\caption{A key for interpreting Column A of Table~\ref{table: n+5}.}
\label{table: S_i}
\begin{tabular}{c}
$\begin{array}{c|c}
\text{Clause} & \text{Key to interpret Column A of Table~\ref{table: n+5}} \\
\hline
v_\alpha \vee v_\beta \vee v_\gamma & (\nu_j^i[x_\alpha], \nu_j^i[y_\alpha], \nu_j^i[x_\beta], \nu_j^i[y_\beta], \nu_j^i[x_\gamma], \nu_j^i[y_\gamma])\\
\overline{v_\alpha} \vee v_\beta \vee v_\gamma & (\nu_j^i[y_\alpha], \nu_j^i[x_\alpha], \nu_j^i[x_\beta], \nu_j^i[y_\beta], \nu_j^i[x_\gamma], \nu_j^i[y_\gamma])\\
v_\alpha \vee \overline{v_\beta} \vee v_\gamma & (\nu_j^i[x_\alpha], \nu_j^i[y_\alpha], \nu_j^i[y_\beta], \nu_j^i[x_\beta], \nu_j^i[x_\gamma], \nu_j^i[y_\gamma])\\
v_\alpha \vee v_\beta \vee \overline{v_\gamma} & (\nu_j^i[x_\alpha], \nu_j^i[y_\alpha], \nu_j^i[x_\beta], \nu_j^i[y_\beta], \nu_j^i[y_\gamma], \nu_j^i[x_\gamma])\\
\overline{v_\alpha} \vee \overline{v_\beta} \vee v_\gamma & (\nu_j^i[y_\alpha], \nu_j^i[x_\alpha], \nu_j^i[y_\beta], \nu_j^i[x_\beta], \nu_j^i[x_\gamma], \nu_j^i[y_\gamma])\\
\overline{v_\alpha} \vee v_\beta \vee \overline{v_\gamma} & (\nu_j^i[y_\alpha], \nu_j^i[x_\alpha], \nu_j^i[x_\beta], \nu_j^i[y_\beta], \nu_j^i[y_\gamma], \nu_j^i[x_\gamma])\\
v_\alpha \vee \overline{v_\beta} \vee \overline{v_\gamma} & (\nu_j^i[x_\alpha], \nu_j^i[y_\alpha], \nu_j^i[y_\beta], \nu_j^i[x_\beta], \nu_j^i[y_\gamma], \nu_j^i[x_\gamma])\\
\overline{v_\alpha} \vee \overline{v_\beta} \vee \overline{v_\gamma} & (\nu_j^i[y_\alpha], \nu_j^i[x_\alpha], \nu_j^i[y_\beta], \nu_j^i[x_\beta], \nu_j^i[y_\gamma], \nu_j^i[x_\gamma])\\
\end{array}$

\end{tabular}
\begin{tablenotes}[flushleft]
\item[]
For any clause in the left column, the corresponding entry in the right column above will be matched with the 6-bit string in Column A of row $j$ of Table~\ref{table: n+5} to determine the value of $\nu_j^i$ at each bit in the support set $S_i$.
\end{tablenotes}
\end{threeparttable}
\end{table}

\subsection{Hamming distances between $C_i$ and possible medians} 
\label{table explanation}
Here we explain the remainder of Table~\ref{table: n+5}.
Fix a clause $c_i$ in $\Gamma$ which will be used throughout this subsection. Suppose $c_i$ has variables $v_\alpha$, $v_\beta$, and $v_\gamma$. 
By Fact~\ref{fact: medians}, $\mathcal{M}(\mathcal{C}_i)=\{0,1\}^{2n}\times \{0\}^t.$ Therefore, $\mathcal{M}'(\mathcal{C}_i)$  must be equal to $\{01,10\}^n \times \{0\}^t$. For this subsection, define 
\begin{align*}
\mathcal{M} :=\mathcal{M}(\mathcal{C}_i),\qquad \qquad
\mathcal{M}' := \mathcal{M}'(\mathcal{C}_i).
\end{align*}

Define an equivalence relation $\sim_i$ on $\mathcal{M}'$ such that two medians are equivalent  if they agree on the coordinates in the support set $S_i$ of $c_i$. The result will be 8 equivalence classes
because $\mu[x_\ell] \neq \mu[y_\ell]$ for each $\ell\in \{\alpha, \beta, \gamma\}$ for each $\mu\in \mathcal{M}'$. 

Here we define a one-to-one correspondence between the equivalence classes of $\mathcal{M}'$ under $\sim_i$ and the 6-bit strings heading Columns M1 through M8 in Table~\ref{table: n+5}. 

\begin{defn}
Fix a clause $c_i$ and an integer $\ell\in[8]$.
Consider the 6-bit string $\delta$ which heads column M$\ell$. In Table~\ref{table: S_i}, locate the tuple in the right column corresponding to our fixed clause $c_i$. After replacing each $\nu_j^i$ with $\mu$ in the tuple, match this tuple with $\delta$. This gives six values that a median $\mu\in \mathcal{M}'$ must have if it is in the equivalence class represented by the column heading $\delta$. 
\label{median bijection}
\end{defn}

In Definition~\ref{defi: bijection}, we defined a correspondence between $\mathcal{M}'$ and truth assignments for $\Gamma$. In Definition~\ref{defi: MGamma}, we introduced the notation $\mathcal{M}'_\Gamma(\mathcal{C}_i)$ for the collection of medians in $\mathcal{M}'$ which correspond to satisfying truth assignments for $\Gamma$. Similarly, we defined $\mathcal{M}'_{c_i}(\mathcal{C}_i)$ for each clause $c_i$ in $\Gamma$. For the remainder of this subsection, set 
\begin{align*}
\mathcal{M}'_\Gamma :=\mathcal{M}'_\Gamma(\mathcal{C}_i), \qquad \qquad
 \mathcal{M}'_{c_i} := \mathcal{M}'_{c_i}(\mathcal{C}_i).
 \end{align*}
 The following claim uses the correspondence in Definition~\ref{median bijection} to connect $\mathcal{M}' \setminus \mathcal{M}'_{c_i}$ with a particular equivalence class.

\begin{claim}
Let $c_i$ be a clause in $\Gamma$. For any $\mu\in {\mathcal M}'$, $\mu$ is in the equivalence class
represented by Column M8 of Table~\ref{table: n+5} if and only if $\mu \in  {\mathcal M}' \setminus  {\mathcal M}'_{c_i}$.
\label{bad class}
\end{claim}

\begin{proof}
Fix a clause $c_i$ with variables $v_\alpha, v_\beta, v_\gamma$. This clause may have some negative literals. We focus our attention on $v_\alpha$. The arguments for $v_\beta$ and $v_\gamma$ are exactly the same. 

There are two cases depending on whether $v_\alpha$ appears as a positive literal or a negative literal in $c_i$. 

In the case where $v_\alpha$ appears in $c_i$ as a positive literal, the truth assignment which makes $c_i$ false assigns a value of false to $v_\alpha$. A corresponding median $\mu\in \mathcal{M}'$ has $\mu[x_\alpha]=0$ and $\mu[y_\alpha]=1$. Because $v_\alpha$ appears as a positive literal in $c_i$, the entry in the second column of Table~\ref{table: S_i} has $\mu[x_\alpha]$ followed by $\mu[y_\alpha]$. So, in this case, the 6-bit string which heads the column for medians in $\mathcal{M}' \setminus \mathcal{M}'_{c_i}$ has $01$ in the first two entries. 

In the case where $v_\alpha$ appears as a negative literal in $c_i$, the non-satisfying truth assignments for $c_i$ must have $v_\alpha$ true. The corresponding medians $\mu\in \mathcal{M}'$ will have $\mu[x_\alpha]=1$ and $\mu[y_\alpha]=0$. For the clauses with variable $v_\alpha$ appearing as a negative literal in $c_i$, a quick glance at Table~\ref{table: S_i} reveals that $\mu[y_\alpha]$ immediately precedes $\mu[x_\alpha]$ in the 6-bit column headings in Table~\ref{table: n+5}. As a result, the column representing medians in $\mathcal{M}' \setminus \mathcal{M}'_{c_i}$ has $01$  in the first two entries. 

Repeating this argument for $v_\beta$ and $v_\gamma$, we see that medians in $\mathcal{M}'\setminus \mathcal{M}'_{c_i}$ are represented by the column with heading $010101$. 
\end{proof}

Now that we have defined the rows and columns of Table~\ref{table: n+5}, we conclude this subsection by defining the entries within Table~\ref{table: n+5} for fixed clause $c_i$.

Let $\mu\in \mathcal{M}'$ be an arbitrary median that falls into the equivalence class represented by Column M$\ell$ for some $\ell\in[8]$. The entry $a_{j\ell}$ in Row $j$ and Column M$\ell$ of Table~\ref{table: n+5} is $H(\mu, \nu_j^i)$. This value can be calculated as follows: 
\begin{itemize}
\item First, take the Hamming distance between the 6-bit string in Column $A$ of Row $j$ and the 6-bit string in the header of Column M$\ell$. This is equal to the Hamming distance between the restrictions of $\mu$ and $\nu_j^i$ to the support set $S_i$ for $c_i$. 
\item For any $s \not\in \{\alpha, \beta, \gamma\}$, $\mu[x_s] \neq \mu[y_s]$ and $\nu_j^i[x_s] = \nu_j^i[y_s]$. Therefore the Hamming distance between $(\mu[x_s] , \mu[y_s])$ and $(\nu_j^i[x_s] , \nu_j^i[y_s])$ is 1 for each $s\in [n] \setminus \{\alpha, \beta, \gamma\}$. 
\item Finally, because $\mu[e_s]=0$ for all $s\in[t]$, the Hamming distance between the restrictions of $\mu$ and $\nu_j^i$ to the coordinates $(e_1, e_2, \ldots, e_t)$ is the number of additional ones in $\nu_j^i$ which is found in Column C of Row $j$.
\end{itemize}
Adding these three values together gives the entry $a_{j\ell}$. 

\subsection{Distinguishing the satisfying truth assignments}
\label{distinguish}
Fix a clause $c_i$ in arbitrary D3CNF $\Gamma$. For this subsection, we again set $\mathcal{M}':= \mathcal{M}'(\mathcal{C}_i)$, $\mathcal{M}'_\Gamma:= \mathcal{M}'_\Gamma(\mathcal{C}_i)$, and $\mathcal{M}'_{c_i}:= \mathcal{M}'_{c_i}(\mathcal{C}_i)$.  
For each $\mu \in \mathcal{M}'\setminus \mathcal{M}'_{c_i}$, $\mu$ is in the equivalence class represented by Column M8 according to Claim~\ref{bad class}. Then reading the entries in Column M8 of Table~\ref{table: n+5}, we find the multiset (where parenthetical subscripts give the multiplicity of that value in the multiset):
\begin{align}
\{H(\mu,\nu_j^i): j\in[50]\} = \{&(n-2)_{(1)}, (n-1)_{(6)}, n_{(3)}, (n+1)_{(15)},  \nonumber \\ & (n+2)_{(15)}, (n+3)_{(3)}, (n+4)_{(6)}, (n+5)_{(1)} \}.
\label{dist for nonsat}
\end{align}
Otherwise, for each median $\mu\in \mathcal{M}'_{c_i}$, $\mu$ is in one of 7 equivalence classes represented in Columns M1 through M7. The entries in each of these columns yields	
\begin{align}
 \{H(\mu,\nu_j^i): j\in[50]\} = \{(n-1)_{(7)}, n_{(6)}, (n+1)_{(12)}, (n+2)_{(12)}, (n+3)_{(6)}, (n+4)_{(7)}\}.
 \label{dist for sat}
 \end{align}

Therefore,we can use $\mathcal{C}_i$ to distinguish between the medians in $\mathcal{M}'_{c_i}$ and the medians in $\mathcal{M}'\setminus \mathcal{M}'_{c_i}$. For example, given $\mu\in \mathcal{M}' = \{01,10\}^n\times \{0\}^t$, if we determine that  $(n+5) \in \{H(\mu, \nu_j^i): j\in [50]\}$, then we can conclude $\mu\in \mathcal{M'}\setminus  \mathcal{M}'_{c_i}.$

Now we wish to consider all of the $C_i$ multisets together. It is clear that each $x_i$ and each $y_i$ coordinates will remain ambiguous in the multiset $\biguplus_{i\in[k]} \mathcal{C}_i$. For the additional ones, we will take $t$ large enough to maintain the property that, for each $i\in [t]$, there is at most one binary string $\eta$ in $\biguplus_{i\in[k]} \mathcal{C}_i$ with $\eta[e_i]=1$. As a result, 
\[\mathcal{M}\left(\biguplus_{i\in[k]} \mathcal{C}_i\right) = \{0,1\}^n\times \{0\}^t.\]
Further, 
\begin{align*}
\mathcal{M}'_{c_i} := \mathcal{M}'_{c_i}(\mathcal{C}_i) = \mathcal{M}'_{c_i} \left(\biguplus_{i\in[k]} \mathcal{C}_i \right), \\
\mathcal{M}'_{\Gamma}:= \mathcal{M}'_{\Gamma}(\mathcal{C}_i) = \mathcal{M}'_{\Gamma} \left(\biguplus_{i\in[k]} \mathcal{C}_i \right) .
\end{align*}
By definition of the sets $\mathcal{M}'_{c_i}$ and $\mathcal{M}'_{\Gamma}$,
\begin{align}
\mathcal{M}'_{\Gamma} &= \bigcap_{i\in[k]}\mathcal{M}'_{c_i},\\
\mathcal{M}' \setminus \mathcal{M}'_{\Gamma} & = \mathcal{M}' \setminus \bigcap_{i\in[k]}\mathcal{M}'_{c_i} = \bigcup_{i\in[k]} \left(\mathcal{M'}\setminus  \mathcal{M}'_{c_i}\right).
\end{align}
Therefore the multiset $\biguplus_{i\in[k]}\mathcal{C}_i$ will serve as a tool to distinguish $\mathcal{M}'_{\Gamma}$ from $\mathcal{M}' \setminus \mathcal{M}'_{\Gamma}$. 

\section{Complexity of computing $Z(B,x!)$}
\label{sec:factorial}
Before stating Theorem~\ref{thm: number SCJ scenarios}, we need a result which is equivalent to the Prime Number Theorem. Define \[\theta(x) := \sum_{\substack{p \leq x \\  p\,\, prime}} \log p.\]

\begin{thm}
  $\theta(x) \sim x.$ 
\end{thm}

As a result, the next lemma and corollary hold.  
\begin{lemma}[\cite{rosser}]
For $2 \leq x$, \[\left( 1- \frac{2.85}{\log x} \right)x \leq \theta(x) \leq \left( 1+ \frac{2.85}{\log x} \right)x.\]
\end{lemma}

\begin{cor}
For any $n\geq 300$,
	\[e^{n/2} \leq \prod_{\substack{p\leq n \\ p \,\, prime}} p \leq e^{3n/2}.\]
	\label{Cor: prime}
\end{cor}

Now we can prove the main result for star trees.

\begin{thm}
Calculating $Z(B,x!)$ is \#P-complete.
\label{thm: number SCJ scenarios}
\end{thm}

\begin{proof}
In Lemma~\ref{lemma: sharpP}, we verified that calculating $Z(B,x!)$ is in \#P .
To show \#P-complete, we give a polynomial time reduction from \#D3SAT.
Fix an arbitrary D3CNF $\Gamma=c_1 \wedge c_2 \wedge \ldots \wedge c_k$ where each $c_i$ is a clause and $\Gamma$ has $n$ variables. 

Using the bound in Corollary~\ref{Cor: prime}, let $n' =\max\{300, n+5\}$.
Fix a prime number $p$ which is greater than $n'$ and at most $5n'$. Let \[q:=p-(n+5).\]

We will explicitly define a multiset \[\mathcal{D}(p) = \mathcal{A}(p)\cup \bigcup_{i\in [n]} \mathcal{B}_i(p) \cup \bigcup_{i\in[k]} \mathcal{C}_i(p)\] consisting of $2+2n+50k$ binary strings with coordinates 
\[(x_1, y_1, x_2, y_2, \ldots, x_n, y_n, e_1, \ldots, e_{t(p)})\]
where 
\begin{eqnarray}t(p):=2(q+4) + 2n(q+3) + k(75 + 50q).\label{t(p)} \end{eqnarray}
The coordinates $e_1, e_2, \ldots, e_{t(p)}$ are for the \emph{additional ones}. In order to define each $\eta \in \mathcal{D}(p)$, we will give exact values for $\eta[x_j]$ and $\eta[y_j]$ for each $j\in [n]$ and specify the number of additional ones that $\eta$ will have. Definition~\ref{defi: ones} tells how to obtain the values of $\eta[e_j]$ for each $j\in [t(p)]$ from this information.  

All strings in $\mathcal{D}(p)$ will come in pairs which are complementary on the first $2n$ entries (Definition~\ref{defn: ambig}).  As a result, we can use Fact~\ref{fact: comp ambig} to see that each of the first $2n$ coordinates are ambiguous in $\mathcal{D}(p)$. 

Now we begin defining the strings in multiset that together create $\mathcal{D}(p)$. The set $\mathcal{A}(p)$ consists of two strings, $\alpha$ and $\ov\alpha$. Define  $\alpha$ to have $\alpha[x_i]=\alpha[y_i]=1$ for all $i\in [n]$ and $q+4$ additional ones. Define $\ov\alpha$ to be complementary to $\alpha$ on the first $2n$ entries and have $q+4$ additional ones. 

For each $j\in [n]$, the set $\mathcal{B}_j(p)$ will consist of two strings, $\beta_j$ and $\ov{\beta_j}$. Define $\beta_j$ to be the string with $\beta_j[x_j]=\beta_j[y_j]=1$ and for all $j'\in [n]$ with $j'\neq j$, $\beta_{j}[x_{j'}] = \beta_j[y_{j'}] =0$ and $q+3$ additional ones. Define $\ov{\beta_j}$ to be complementary to $\beta_j$ on the first $2n$ entries and have $q+3$ additional ones. 

For each $i\in [k]$, the set $\mathcal{C}_i(p)$ will have 50 strings. These are obtained by adding $q$ more additional ones to the 50 strings in $\mathcal{C}_i$ which were defined through Table~\ref{table: n+5} (see Section~\ref{section: define tables}). In other words, increase each entry in Column C of Table~\ref{table: n+5} by $q$ to obtain $\mathcal{C}_i(p)$.

In summary, we have constructed the strings
\[\mathcal{D}(p):=\mathcal{A}(p)\cup \bigcup_{i\in [n]} \mathcal{B}_i(p) \cup \bigcup_{i\in[k]} \mathcal{C}_i(p).\] 
As described in Definition~\ref{defi: MGamma} and for each clause $c_i$ in $\Gamma$, set
\begin{align*}
\mathcal{M}(p)& :=\mathcal{M}(\mathcal{D}(p)) ,\qquad &
\mathcal{M}'(p) &:=\mathcal{M}'(\mathcal{D}(p)),  \\
\mathcal{M}'_{c_i}(p)&:=\mathcal{M}'_{c_i}(\mathcal{D}(p)),  \qquad &
\mathcal{M}'_\Gamma(p)&:=\mathcal{M}'_\Gamma(\mathcal{D}(p)) .
\end{align*}
As stated in Fact~\ref{fact: ones}, each $\mu\in \mathcal{M}(p)$ has $\mu[e_j]=0$ for all $j\in [t(p)]$. Additionally, because all of the strings in $\mathcal{D}(p)$ come in complementary pairs, the coordinates $x_j$ and $y_j$ are ambiguous for each $j\in[n]$ (Fact~\ref{fact: comp ambig}). Thus there are $2^{2n}$ medians $\mu$. More precisely, 
\begin{align}\mathcal{M}(p) &= \{0,1\}^{2n} \times \{0\}^{t(p)} \label{M(p)} \text{ and} \\
\mathcal{M}'(p)& = \{01,10\}^n \times \{0\}^{t(p)}.\nonumber \end{align} 
Define 
\[ \mathcal{H}(\mu, \mathcal{A}(p)) := \prod_{a\in \mathcal{A}(p)} H(\mu, a)! \]
 and likewise define $\mathcal{H}(\mu, \mathcal{B}_j(p))$ and $\mathcal{H}(\mu,\mathcal{C}_i(p))$ for each $j\in [n]$ and $i\in [k]$. 
Therefore the number of scenarios admitted by median $\mu$ can be expressed by 
\[ \mathcal{H}(\mu) :=  \mathcal{H}(\mu,\mathcal{A}(p)) \cdot \prod_{i\in[n]} \mathcal{H}(\mu,\mathcal{B}_i(p)) \cdot \prod_{i\in[k]} \mathcal{H}(\mu,\mathcal{C}_i(p)).\]

At this point, we wish to calculate $\mathcal{H}(\mu) \! \mod p$ for each median $\mu\in \mathcal{M}(p)$. 
To analyze $\mathcal{H}(\mu)$ for each $\mu\in \mathcal{M}$, we define the following 3 properties that a median $\mu\in \mathcal{M}(p)$ may have. 
\begin{itemize}[leftmargin=1in]
\item [\emph{Property 1}.] $\sum_{i\in [n]} (\mu[x_i] + \mu[y_i]) = n$.
\item [\emph{Property 2}.] $\mu\in\mathcal{M}'(p)$.
\item[\emph{Property 3}.] $\mu \in \mathcal{M}'_\Gamma(p)$. 
\end{itemize}

First notice that these properties are nested. Any $\mu \in \mathcal{M}(p)$ with Property 2 must also have Property 1. Likewise, if $\mu$ has Property 3, it will also have Property 2. The next 4 claims divide $\mathcal{M}(p)$ into 4 classes and examine $\mathcal{H}(\mu)$ for medians in each class. 

\begin{claim} 
For arbitrary $\mu \in \mathcal{M}(p)$, if $\mu$ does not have Property 1, and consequently does not have Property 2 or 3, then $\mathcal{H}(\mu) \equiv 0 \mod p$. 
\end{claim}
\begin{proof}
Let $\mu$ be an arbitrary median in $\mathcal{M}(p)$. For $\alpha\in \mathcal{A}(p)$, Fact~\ref{fact: Hcomp} gives \[H(\mu,\alpha) + H(\mu, \overline\alpha) = 2n + (q+4) + (q+4) = 2p-2.\] Hence, there is an integer $r$ such that $q+4 \leq r \leq 2n+q+4$ and $H(\mu,\alpha)=r$ with
\[\mathcal{H}(\mu,\mathcal{A}(p)) = r!(2p-2-r)!.\]  
Since $\mu$ does not have Property 1, we can conclude that exactly one of the following holds: 
\begin{align*}
H(\mu, \alpha)& \geq (n+1) + (q+4) = p, \text{ or} \\
H(\mu,\overline\alpha)&\geq (n+1) + (q+4) =p.
\end{align*}
Therefore, either $r\geq p$ or $(2p-2-r)\geq p$. In the first case, $r!$ is divisible by $p$ and, in the second, $(2p-2-r)!$ is divisible by $p$. Therefore $\mathcal{H}(\mu,\mathcal{A}(p)) \equiv 0\! \mod p$ and consequently $\mathcal{H}(\mu) \equiv 0\! \mod p$. 
\end{proof}

\begin{claim}
For an arbitrary $\mu \in \mathcal{M}(p)$, if $\mu$ has Property 1, but does not have Property 2, then $\mathcal{H}(\mu)\equiv 0 \!\mod p$.
\end{claim}
\begin{proof}
Suppose $\mu \in \mathcal{M}(p) \setminus \mathcal{M}'(p)$ but $\mu$ has Property 1. 
Because $\mu\not \in \mathcal{M}'(p)$, there is an integer $j_0\in [n]$ such that $\mu[x_{j_0}]=\mu[y_{j_0}]$. In the case when $\mu[x_{j_0}]=0$, we have $H(\mu, \beta_{j_0})=(n+2) + (q+3)=p$. Otherwise $\mu[x_i]=1$ which implies $H(\mu, \overline{\beta_{j_0}})=(n+2)+(q+3) =p$. 
In either case,
\[\mathcal{H}(\mu, \mathcal{B}_{j_0}(p)) = p!(p-4)!\]  
and consequently $\mathcal{H}(\mu) \equiv 0 \mod p.$
\end{proof}

\begin{claim}
For an arbitrary $\mu\in \mathcal{M}(p)$, if $\mu$ has Properties 1 and 2, but does not have Property 3, then $\mathcal{H}(\mu) \equiv 0 \mod p$. 
\end{claim}
\begin{proof}
Let $\mu$ be in $\mathcal{M}'(p) \setminus\mathcal{M}'_{\Gamma} (p)$. Since $\mu$ corresponds to a truth assignment which does not satisfy $\Gamma$, there is a clause $c_{i_0}$ in $\Gamma$ which is not satisfied by this truth assignment. Therefore $\mu\in \mathcal{M}' (p)\setminus \mathcal{M}'_{c_{i_0}} (p)$. By \eqref{dist for nonsat}, before adding the $q$ additional ones to each string from $\mathcal{C}_{i_0}$, we have 
\begin{align}
 \{H(\mu, \nu^{i_0}_j): \nu^{i_0}_j \in \mathcal{C}_{i_0}\} =  
 \{&(n-2)_{(1)}, (n-1)_{(6)}, n_{(3)}, (n+1)_{(15)},  \nonumber \\ & (n+2)_{(15)}, (n+3)_{(3)}, (n+4)_{(6)}, (n+5)_{(1)} \}.
\end{align}
To create $\mathcal{C}_{i_0}(p)$, we added $q$ additional ones to each string in $\mathcal{C}_{i_0}$ which increased each Hamming distance by $q$. Therefore 
\begin{align*}
 \{H(\mu, \nu^{i_0}_j): \nu^{i_0}_j \in \mathcal{C}_{i_0}(p)\} =  
 \{&(p-7)_{(1)}, (p-6)_{(6)}, (p-5)_{(3)}, (p-4)_{(15)},  \nonumber \\ & (p-3)_{(15)}, (p-2)_{(3)}, (p-1)_{(6)}, p_{(1)} \}.
 \end{align*}
As a result,
\[\mathcal{H}(\mu, \mathcal{C}_{i_0}(p)) = (p-7)! (p-6)!^6 (p-5)!^3 (p-4)!^{15} (p-3)!^{15} (p-2)!^3 (p-1)!^6 p!\] which is divisible by $p$. Therefore  $\mathcal{H}(\mu) \equiv 0 \mod p.$
\end{proof}

\begin{claim}
For an arbitrary $\mu\in \mathcal{M}(p)$ having Properties 1, 2, and 3, the value 
\[\mathcal{H}(\mu) = (p-6)!^{7k} (p-5)!^{6k} (p-4)!^{12k} (p-3)!^{12k} (p-2)!^{6k+2n} (p-1)!^{7k+2},\] which is not congruent to 0 modulo $p$. 
\label{all props}
\end{claim}
\begin{proof}
Let $\mu\in\mathcal{M}'_\Gamma(p)$. Because it has Property 1, 
\[\mathcal{H}(\mu, \mathcal{A}(p)) = (n + (q+4))!^2 = (p-1)!^2.\]
Since $\mu$ has Property 2, for any $i\in [n]$,
\[\mathcal{H}(\mu, \mathcal{B}_i(p))  = (n + (q+3))!^2 = (p-2)!^2 .\]
Finally, $\mu$ satisfies Property 3 which means $\mu\in \mathcal{M}'_{c_i}(p)$ for all clauses $c_i$ in $\Gamma$. 

Recall that each string $\eta \in \mathcal{C}_i(p)$ is created from a string $\eta'\in\mathcal{C}_i$ by adding $q$ more additional ones. Therefore $H(\mu, \eta) = H(\mu, \eta')+q$.
So, the multiset $\mathcal{H}(\mu, \mathcal{C}_i(p))$ can be obtained from $\mathcal{H}(\mu, \mathcal{C}_i)$ found in \eqref{dist for sat} by adding $q$ to each element. As a result,
 \[\mathcal{H}(\mu, \mathcal{C}_i(p))  = (p-6)!^7 (p-5)!^6 (p-4)!^{12} (p-3)!^{12} (p-2)!^6 (p-1)!^7 .\]

Therefore 
\begin{align}
\mathcal{H}(\mu) = (p-6)!^{7k} (p-5)!^{6k} (p-4)!^{12k} (p-3)!^{12k} (p-2)!^{6k+2n} (p-1)!^{7k+2}.
\label{H_good}
\end{align}
Because $p$ is prime, $\mathcal{H}(\mu) \not\equiv 0 \mod p.$ 
\end{proof} 

Set \[T(p) := \sum_{\mu\in \mathcal{M}(p)} \mathcal{H}(\mu).\] Set $K(p)$ equal to the function of $p$ displayed in \eqref{H_good}. Thus $K(p)$ is precisely the value of the number of SCJ scenarios admitted by an arbitrary $\mu\in\mathcal{M}'_{\Gamma}(p)$. 
If we calculate $T(p)\!\! \mod p$, the four claims show that 
\begin{align}
T(p) \equiv \sum_{\mu\in \mathcal{M}'_\Gamma(p)} \mathcal{H}(\mu)
 \,\equiv \, |\mathcal{M}'_\Gamma(p)| \cdot K(p) \!\! \mod p.
 \end{align}
If $\gamma$ is the number of satisfying truth assignments for $\Gamma$, then $\gamma=|\mathcal{M}'_\Gamma(p)| $ by Definition~\ref{defi: MGamma}. Therefore 
\[ \gamma \cdot K(p) \equiv T(p)\!\! \mod p.\]
Since $p$ does not divide $K(p)$ (Claim~\ref{all props}), there exists an integer $K'(p)$ such that $K(p)\cdot K'(p) \equiv 1 \!\!\mod p$. Thus 
\[ \gamma \equiv K'(p) \cdot T(p) \!\!\mod p.\]
While this alone is not sufficient to determine the value of $\gamma$, we can repeat this construction for many different prime values to obtain more congruences. 

Recall $p$ was fixed to be a prime greater than $n'$ and at most $5n'$. Repeat the above construction for each prime $p_1, p_2, \ldots, p_m$  in this range. 
The result is a list of congruences:
\begin{align*}
\gamma &\equiv K'(p_1) \cdot T(p_1)\!\! \mod p_1, \\
\gamma &\equiv K'(p_2) \cdot T(p_2) \!\! \mod p_2, \\
& \vdots \\
\gamma &\equiv K'(p_m) \cdot T(p_m) \!\! \mod p_m. 
\end{align*}

Because $p_1, p_2, \ldots, p_m$ are all prime, the Chinese Remainder Theorem guarantees a solution for $\gamma$ which is unique modulo $\prod_{i\in [m]} p_i$. 
By the Corollary~\ref{Cor: prime}, 
\[\prod_{i\in [m]} p_i =\frac{\displaystyle\prod_{\substack{p\leq 5n' \\ p \,\, prime}} p}{ \displaystyle\prod_{\substack{p\leq n' \\ p \,\, prime}} p}
\geq \frac{e^{5n'/2}}{e^{3n'/2}} = e^{n'} \geq e^{n}.\]
Since $\gamma$ is the number of satisfying truth assignments for $\Gamma$, and there are only $n$ literals which can realize one of two values, $\gamma\leq 2^n$. Since $\prod_{i\in [m]} p_i \geq e^n> 2^n\geq \gamma$, the Chinese Remainder Theorem gives the exact value of $\gamma$. 

In summary, for D3CNF $\Gamma$ with $n$ variables and $k$ clauses, we use the Sieve of Eratosthenes to identify the primes between $n'$ and $5n'$. This runs in $O(n^2)$ time. 
Then for each prime $p$ in this interval (which is at most $\max\{2n, 600\}$ primes), we create $50k + 2n +2$ binary strings of length $2n + t(p)$ where $t(p)$ is a polynomial in $n$ and $p$ with $p\in O(n)$. Finally, the Chinese Remainder Theorem will solve the system of congruences in $O(\log^2 (p_1p_2\ldots p_m))$ time \parencite{bach}. For us, this is $ O(n^2\log^2 n)$ because each prime is at most $5n$ and $m\leq 2n$.

Therefore, if we had algorithm to determine the value of $Z(B,x!)$ which ran in time polynomial in the size of $B$ and the length of the strings in $B$, then we have created here a polynomial time algorithm to determine the number of satisfying truth assignments for a D3CNF, a problem which is known to be \#P-complete. This finishes the proof.
\end{proof}

\section{Stochastic Approximations for $Z(B,x!)$}
\label{sec:factorial approx}
In the previous section, we proved calculating $Z(B,x!)$ is a \#P-complete problem. The natural next question is whether or not this value can be approximated. Viewing this problem as counting most parsimonious scenarios for the star phylogenetic tree with leaves labeled by the strings in $B$, we are also interested in a near uniform sampler of these labelings. 

\begin{defn}
A counting problem \#A in \#P has an FPAUS (fully polynomial almost uniform sampler) if there is a randomized algorithm such that, for any instance of \#A and any $\epsilon > 0$, the algorithm outputs an element $x\in X$, the solution space for \#A, with probability $p(x)$ where
\[\frac{1}{2} \sum_{x\in X} |p(x) - U(x)| \leq \epsilon\]
where $U$ is the uniform distribution on X and the algorithm runs in time polynomial in the size of the instance of \#A and $-\log \epsilon$. 
\label{defn: FPAUS}
\end{defn}

\begin{defn} A counting problem \#A in \#P has an FPRAS (fully polynomial randomized approximation scheme) if there is a randomized algorithm such that, for any instance of \#A and any $\epsilon, \delta >0$, the algorithm outputs an approximation $\hat{f}$ for the true answer $f$ of the counting problem satisfying the following inequality
\begin{equation}
P\left(\frac{f}{1+\epsilon} \le \hat{f} \le f(1+\epsilon)\right) \ge 1- \delta
\end{equation}
Furthermore, the algorithm runs in time polynomial in the size of the instance \#A, $\epsilon^{-1}$,  and $-\log(\delta)$.
\label{defn:FPRAS}
\end{defn}

The modulo prime number computation technique which was used to prove that calculating $Z(B,x!)$ is in \#P-complete has 
been used to show that other problems are \#P-complete. For example, \textcite{brightwell} used this technique to prove that counting the number of linear extensions of a partially order set is \#P-complete. 
For this same problem, \textcite{karzanov} found a rapidly mixing Markov chain to sample the linear extensions. Since counting the linear extensions of a partially ordered set is a self-reducible counting problem, this means that it also has an FPRAS \parencite{jvv86}.
This may suggest that our problem of counting most parismonious scenarios also has an FPAUS and FPRAS. 
However, here we give a straightforward Markov chain to sample the most parsimonious scenarios that turns out to be torpidly mixing, suggesting that our problem may not have an FPAUS. 
With evidence for both the positive answer and the negative answer, the question of whether or not there is an FPAUS and FPRAS for $Z(B,x!)$ remains open.

Recall that a median $\mu$ for $B=\{\nu_i\}_{i=1}^m$ minimizes $\sum_{i\in [m]} H(\nu_i, \mu)$. Therefore, in the $k^{th}$ bit, the value of $\mu$ must agree with a majority of the strings in $B$. If exactly half of the strings in $B$ have a 1 in the $k^{th}$ bit, then $\mu$ may take either a 0 or a 1 in the $k^{th}$ bit. We call such a bit an \emph{ambiguous bit}. Therefore a median for $B$ is determine by the value it takes in the ambiguous bits. 
As a result, if $B$ has an odd number of strings, then there is exactly one median. Here we assume that the size of $B$ is even. 

Define a primer Markov chain, $P$, to transition between the medians. 
As mentioned, it suffices to define our Markov Chain on the state space of all possible values that a median could take on the ambiguous bits. From any median, make a transition with the following probabilities: 
\begin{itemize}
\item With probability 1/2, remain in the current state.
\item With probability 1/2, randomly and uniformly select an ambiguous bit and change its value. 
\end{itemize}
Because we remain at the current state with  probability $\frac{1}{2}$, by definition this is a lazy Markov chain.

\begin{obs}
The primer Markov chain $P$ is irreducible and aperiodic. 
\end{obs}

For a fixed multiset of strings $B = \{\nu_i\}_{i=1}^m$ and median $\mu\in \mathcal{M}(B)$, define 
\[f(\mu):=\prod_{i=1}^m H(\mu,\nu_i).\] 
Now we employ the Metropolis-Hastings algorithm \parencite{Metropolis} to obtain a secondary Markov chain $C$ with a desired limit distribution as follows. The states remain the same, but the transition probabilities are changed in the following way. From state $\mu$, we propose a next state $\mu'$ which differs from $\mu$ in at most one bit. If $\mu'$ is different from $\mu$, accept this transition with probability \[\min\left\{1,\frac{f(\mu')}{f(\mu)}\right\}.\]
In other words, if $\mu'$ was reached from $\mu$ with probability $P(\mu'|\mu)$, then in the secondary Markov chain $C$ the transition from $\mu$ to $\mu'$ will be made with probability
\[C(\mu'|\mu) = P(\mu'|\mu)\cdot\min\left\{1,\frac{f(\mu')}{f(\mu)}\right\}. \]

For a given collection of strings, the function $f$ defines a probability distribution $\theta$ on the medians where $\theta(\mu)$ is directly proportional to $f(\mu)$. In other words, 
\begin{eqnarray}\theta(\mu) \propto f(\mu)\label{theta}\end{eqnarray} 
or $\theta(\mu) = k f(\mu)$ for some constant $k$ and any median $\mu$. 

\begin{obs}
Markov chain $C$ is reversible and converges to the limit distribution $\theta$. 
\end{obs}

Therefore, we have a Markov chain on the state space of medians which, in the limit, will sample each median $\mu$ with distribution proportional to $\prod_{i=1}^m H(\mu,\nu_i)$. Once we have a median, it is easy to uniformly sample from the scenarios that it admits. 

Now we will show that the Markov chain $C$ is torpidly mixing (not rapidly mixing). To prove this result, we will need the following definitions.

For any nonempty subset $S$ of the set of medians $\mathcal{M}(B)$, the \emph{capacity of $S$} is 
\[\theta(S):= \sum_{\mu\in S} \theta(\mu)\]
 and the \emph{ergodic flow out of $S$} is 
 \[F(S):= \sum_{\substack{\mu\in S \\ \nu\in \mathcal{M}(B)\setminus S}} \theta(\mu) C(\mu|\nu).\]
 The \emph{conductance} is 
 \[\Phi:=\min\left\{ \frac{F(S)}{\theta(S)} : S\subseteq M, 0< \theta(S)\leq \frac{1}{2} \right\}.\]

\begin{thm}[\cite{bremaud}]
A Markov chain is rapidly mixing if and only if $\Phi \geq \frac{1}{p(n)}$ for some polynomial $p(n)$ which is not identically zero.
\label{rapid_mixing}
\end{thm}

Consider the following instance of $Z(B,x!)$. 
Define $\nu$ and $\overline{\nu}$ to be the strings in $\{0,1\}^n$ where $\nu$ is the string of all 0s and $\overline{\nu}$ is the string of all 1s. Let $B=\{\nu_i\}_{i=1}^{2t}$ be the multiset containing $t$ copies of $\nu$ and $t$ copies of $\overline{\nu}$. The set of medians $\mathcal{M}(B)$ is equal to $\{0,1\}^n$. Further, if $\mu$ has exactly $k$ ones, then 
\[\prod_{i=1}^{2t} H(\mu,\nu_i)!= \left(k!(n-k)!\right)^t.\] 
Consequently \[Z(B,x!) = \sum_{\mu\in \mathcal{M}(B)} \prod_{i=1}^{2t} H(\mu,\nu_i)  = \sum_{k=0}^n \binom{n}{k} \left(k! (n-k)!\right)^t:=T.\] 
Therefore $\theta(\mu) = \frac{1}{T}\left(k!(n-k)!\right)^t.$

Suppose $n$ is odd. Consider the subset $S$ which contains all medians with at most $\left\lfloor \frac{n}{2} \right\rfloor$ ones. For this subset,  the capacity is $\theta(S) =  \frac{1}{2}.$

Let $S'$ be the set of medians $\mu$ in $S$ with exactly $\lfloor \frac{n}{2} \rfloor$ ones. Let $\hat{S}$ be the set of medians $\nu$ in $\mathcal{M} \setminus S$ with exactly $\lceil \frac{n}{2} \rceil$ ones. Then $|S'|=\binom{n}{\left\lfloor \frac{n}{2} \right\rfloor} = \binom{n}{\left\lceil \frac{n}{2} \right\rceil} = |\hat{S}|$. 
For each $\mu\in S \setminus S'$ and $\nu\in \mathcal{M}\setminus S$, $C(\mu|\nu)= 0$.
Further, for each $\mu\in S'$, there are only $\lceil \frac{n}{2} \rceil$ medians $\nu$ in $\mathcal{M} \setminus S$ such that $C(\mu|\nu) \neq 0$ and for these medians $\nu\in \hat{S}$ and $C(\mu|\nu)=\frac{1}{2}\cdot\frac{1}{n}$. 

For the ergodic flow out of $S$, we have
\begin{align*}
F(S) = & \sum_{\substack{\mu\in S \\ \nu\in \mathcal{M}\setminus S}} \theta(\mu) C(\mu|\nu) \\
	   = &	\sum_{\substack{\mu\in S' \\ \nu\in \mathcal{M}\setminus S}} \frac{1}{T} \left(\left\lfloor \frac{n}{2} \right\rfloor! \left\lceil \frac{n}{2} \right\rceil ! \right)^t  C(\mu|\nu)\\
      = &	\sum_{\substack{{\mu\in S'}}} \frac{1}{T} \left(\left\lfloor \frac{n}{2} \right\rfloor! \left\lceil \frac{n}{2} \right\rceil ! \right)^t  \frac{1}{2}\frac{1}{n} \left\lceil \frac{n}{2} \right\rceil \\
            = &	\binom{n}{\left\lfloor \frac{n}{2} \right\rfloor} \frac{1}{T} \left(\left\lfloor \frac{n}{2} \right\rfloor! \left\lceil \frac{n}{2} \right\rceil ! \right)^t  \frac{1}{2}\frac{1}{n} \left\lceil \frac{n}{2} \right\rceil \\
             = &	 \frac{1}{2n}  \left\lceil \frac{n}{2} \right\rceil \frac{n!}{T}
 \left(\left\lfloor \frac{n}{2} \right\rfloor! \left\lceil \frac{n}{2} \right\rceil ! \right)^{t-1} \\
   = &	 \frac{1}{2n}  \frac{n+1}{2}  \frac{n!}{  n! \sum_{k=0}^n \left(k! (n-k)!\right)^{t-1}  }    \left(\left\lfloor \frac{n}{2} \right\rfloor! \left\lceil \frac{n}{2} \right\rceil ! \right)^{t-1} \\
  \leq&  \frac{1}{2}\frac{1}{   \left(0! n!\right)^{t-1}  }  \left(\left\lfloor \frac{n}{2} \right\rfloor! \left\lceil \frac{n}{2} \right\rceil ! \right)^{t-1} \\
     \leq & \frac{1}{2}\frac{1}{  \binom{n}{\left\lfloor \frac{n}{2} \right\rfloor} ^{t-1}  } .
\end{align*}
	   
This implies
\begin{align*}
\Phi \leq & \frac{F(S)}{\pi(S)} 
\leq \frac{1}{  \binom{n}{\left\lfloor \frac{n}{2} \right\rfloor} ^{t-1}  } 
	 \leq  \left(\frac{n+1}{2^n}\right)^{t-1} \leq \left(\frac{2^{n/2}}{2^n}\right)^{t-1} = \frac{1}{2^{n(t-1)/2}}.
\end{align*}

Therefore, if $t> 1$, then as $n$ grows, we see that $\Phi$ cannot be lower-bounded by a function of the form $\frac{1}{p(n)}$ where $p$ is a polynomial in $n$. Therefore the Markov chain $C$ is torpidly mixing by Theorem~\ref{rapid_mixing}.

\section{Complexity of computing $Z(B,f(x))$}
\label{sec:general f(x)}
In this section, we consider the generalized $Z(B,f(x))$. 
First, fix a continuous function $f: \mathbb{R} \rightarrow \mathbb{R}$. Then define the following problem: 
\begin{defn}
Given an arbitrary $m\in \mathbb{Z}^{+}$, let $B=\{\nu_i\}_{i=1}^m$ be an arbitrary multiset of binary strings.  Determine the value of 
\[\sum_{\mu\in\mathcal{M}(S)} \prod_{i\in[m]} f(H(\nu_i, \mu)).\]
\end{defn}

In the previous section, we showed that computing $Z(B,x!)$ is \#P-complete. Here we work toward determining the computational complexity of $Z(B,f(x))$ for various functions $f(x)$. First, we formalize a definition and develop a couple of tools.

\begin{defn}
A function $g:\mathbb{R} \rightarrow \mathbb{R}$ is \textit{strictly concave up} if for any $x,y,z\in \mathbb{R}$, $x<y<z$, 
	\[ \frac{g(z)-g(x)}{z-x} > g(y).\]
\label{concave up}
\end{defn}

\begin{lemma}
If $\log f(x)$ is a strictly concave up function, then for any $x<y$ and $a>0$,
\[\frac{ f(x)f(y) }{ f(x-a)f(y+a) }< 1.\]
\label{concavity2}
\end{lemma}

\begin{proof}
By the intermediate value theorem, there are real values $c,d$ with $c\in (x-a,x)$ and $d\in (y,y+a)$ such that 
\begin{align*}
(\log f)' (c) &= \frac{1}{a}(\log f(x)-\log f(x-a)) \text{ and}\\
(\log f)' (d) &= \frac{1}{a}(\log f(y+a)-\log f(y)).
\end{align*}
 Because $(\log f)'(x)$ is strictly increasing, $g'(c) < g'(d)$. Therefore,
	\begin{align*}
 \frac{1}{a}(\log f(x)-\log f(x-a)) & < \frac{1}{a}(\log f(y+a)-\log f(y)) \\
  \log f(x)-\log f(x-a) & < \log f(y+a)-\log f(y) \\
   \log \frac{f(x)}{f(x-a)} & < \log \frac{f(y+a)}{f(y)} \\
   \frac{f(x)}{f(x-a)} & < \frac{f(y+a)}{f(y)} \\
   \frac{f(x)f(y)}{f(x-a)f(y+a)} & < 1 .
   \end{align*}
\end{proof}

\begin{fact} Fix $k\in \mathbb{Z}^{+} \cup \{0\}$. Let $f(x)$ be a function such that $\log f(x)$ is strictly concave up. Then \[\min_{\substack{\alpha,\beta \in \mathbb{Z}^{+} \cup \{0\} \\ a+b=k}}f(a)f(b) = f \left(\left\lfloor \frac{k}{2}\right \rfloor\right) f\left(\left\lceil \frac{k}{2} \right\rceil\right).\]
\label{fact: min product}
\end{fact}

\begin{proof}
Let $x= \left\lfloor \frac{k}{2}\right \rfloor $ and $y= \left\lceil \frac{k}{2} \right\rceil$. By Lemma~\ref{concavity2}, $f(x-a)f(y+a) < f(x)f(y)$ which gives the desired result. 
\end{proof}

\begin{thm}
 Fix a function $f(x): \mathbb{Z}^{+} \cup \{0\} \rightarrow [0,\infty)$  which satisfies the following properties:
\begin{itemize}
\item $\log f(x)$ is strictly concave up,
\item the function values of $f$ can be computed in polynomial time, and 
\item  for all but finitely many $n\in \mathbb{Z}$, $n\geq 2$, 
\[ \frac{f(n-2)[f(n+1)]^3[f(n+2)]^3 f(n+5)}{f(n-1)[f(n)]^3[f(n+3)]^3f(n+4)} > 1.\]
\end{itemize}
For arbitrary $m,s\in \mathbb{Z}^+$ and $D\in \mathbb{R}$, let $S:=\{\nu_1, \nu_2, \ldots, \nu_m\}$ be a multiset of binary strings, each of length $s$.
Then it is \#P-complete to determine how many medians $\mu$ for $S$ have 
\begin{eqnarray}\prod_{i\in[m]} f\left(H(\nu_i,\mu)\right) \leq D . \label{min_hard}\end{eqnarray}
\label{NP concave up}
\end{thm}

 \begin{proof}
Fix a function $f(x)$ with the properties listed in the theorem. 

If is straightforward to see that computing $Z(B,f(x))$ is in \#P. Fix an instance consisting of integers $m$ and $s$, real number $D$, and a multiset $B$ of binary strings of length $\ell$. Let $\mu$ be a binary string of the same length as each $\nu_i$. We can verify that $\mu$ is a median in time $O(m\ell)$. Each $H(\nu_i, \mu)$ can be computed in time $O(\ell)$. Because $H(\nu_i, \mu) \leq \ell$, we can compute $f(H(\nu_i, \mu))$ in time polynomial in the size of the input by the conditions on $f$. Finally, checking if the product is at most $D$ is also a polynomial time calculation. Therefore  computing $Z(B,f(x))$ is in \#P.

To prove \#P-hardness, we will provide a reduction from \#D3SAT. 
Fix $\Gamma$, a D3CNF with $n$ variables and $k$ clauses, set 
\begin{align*}
\kappa
 = &  
  [f(n)]^{9k} [f(n+1)]^{22k+2kn} [f(n+2)]^{48k + 12kn} \\
  & \cdot [f(n+3)]^{48k + 12kn} [f(n+4)]^{22k+2kn} [f(n+5)]^{9k} 
\end{align*}
The idea is to define a multiset, $\mathcal{D}$, of binary strings with the following properties:
\begin{itemize}
\item Each median $\mu$ which corresponds to a satisfying truth assignment for $\Gamma$ will have 
\[\prod_{\eta \in \mathcal{D}} f(H(\eta, \mu)) = \kappa\]
\item Each other median $\mu'$ will have \[\prod_{\eta \in \mathcal{D}} f(H(\eta, \mu')) > \kappa.\]
\end{itemize}

Create a total of $158k + 28kn$ strings of length $2n+260k + 35kn$ with coordinates
\[(x_1,y_1,x_2,y_2, \ldots, x_n,y_n, e_1, e_2, \ldots, e_t)\]
where $t=260k+35kn$. 
This multiset of binary strings will be defined as the union of three multisets: 
\[\mathcal{D} = \mathcal{A} \uplus \biguplus_{i\in [n]}\mathcal{B}_i \uplus \biguplus_{i\in[k]}\mathcal{C}'_i.\]
As in Definition~\ref{defi: ones}, we will define each string $\eta \in \mathcal{D}$ by explicitly giving the values of $\eta[x_i]$ and $\eta[y_i]$ for each $i\in [n]$ and telling the number of additional ones. 

The collection $\mathcal{A}$ contains $108k$ strings. For $a\in [t]$, let $\alpha^{(+a)}$ be the string with $\alpha[x_i]=\alpha[y_i]=1$ for all $1\leq i\leq n$ and $a$ additional ones. Define $\overline{\alpha}^{(+a)}$ to be the binary string which is complementary to $\alpha^{(+0)}$ on the first $2n$ coordinates and has $a$ additional ones. The multiset $\mathcal{A}$ will consist of the following strings: 
\begin{itemize}
\item k copies each of $\alpha^{(+0)}$ and $\overline{\alpha}^{(+0)}$,
\item 8k copies each of $\alpha^{(+1)}$ and $\overline{\alpha}^{(+1)}$,
\item 18k copies each of $\alpha^{(+2)}$ and $\overline{\alpha}^{(+2)}$,
\item 18k copies each of $\alpha^{(+3)}$ and $\overline{\alpha}^{(+3)}$,
\item 8k copies each of $\alpha^{(+4)}$ and $\overline{\alpha}^{(+4)}$,
\item k copies each of $\alpha^{(+5)}$ and $\overline{\alpha}^{(+5)}$.
\end{itemize}

The collection $\mathcal{B}=\biguplus_{i\in[n]} \mathcal{B}_i$ contains $28kn$ strings. For each $i\in [n]$, $a\in [t]$, let $\beta_i^{(+a)}$ be the string with $\beta_i[x_i]=\beta_i[y_i]=1$, $\beta_i[x_j]=\beta_i[y_j]=0$ for $j\neq i$, and with $a$ additional ones. Define the binary string $\overline{\beta_i}^{(+a)}$ to be complementary to $\beta_i^{(+0)}$ on the first $2n$ coordinates and have $a$ additional ones. The collection $\mathcal{B}_i$ consists of the following $28k$ strings:
\begin{itemize}
\item $k$ copies each of $\beta_i^{(+1)}$ and $\overline{\beta_i}^{(+1)}$,
\item $6k$ copies each of $\beta_i^{(+2)}$ and $\overline{\beta_i}^{(+2)}$,
\item $6k$ copies each of $\beta_i^{(+3)}$ and $\overline{\beta_i}^{(+3)}$,
\item $k$ copies each of $\beta_i^{(+4)}$ and $\overline{\beta_i}^{(+4)}$.
\end{itemize}

The collection $\mathcal{C}=\biguplus_{i\in [k]}\mathcal{C}'_i$ contains $50k$ strings. Each set $\mathcal{C}'_i$, which is associated with clause $c_i$, consists of 50 strings. In Section~\ref{section: define tables}, we defined set  $\mathcal{C}_i$ through Table~\ref{table: n+5}. For each $\nu_j^i\in \mathcal{C}_i$, create $\hat\nu_j^i$ by increasing the number of additional ones in $\nu_j^i$ by one. Then
\[\mathcal{C}'_i := \{\hat\nu_j^i: \nu_j^i\in \mathcal{C}_i\}.\]

Let $\M$ be the set of all medians for $\mathcal{D}$. 
From Definition \ref{defi: MGamma} and for each clause $c_i$ in $\Gamma$, set
\begin{align*}
\mathcal{M}& :=\mathcal{M}(\mathcal{D}), \qquad &
\mathcal{M}' &:=\mathcal{M}'(\mathcal{D}),  \\
\mathcal{M}'_{c_i}&:=\mathcal{M}'_{c_i}(\mathcal{D}),  \qquad &
\mathcal{M}'_\Gamma&:=\mathcal{M}'_\Gamma(\mathcal{D}).
\end{align*}

According to Remark~\ref{defi: ones}, all medians $\mu$ must have $\mu[e_i]=0$ for all $i\in [t]$. In $\mathcal{A}$, $\mathcal{B}$, and $\mathcal{C}$, the strings come in pairs where one is complementary to the other on the first $2n$ coordinates. By Fact~\ref{fact: comp ambig}, each of the $x_i$ and $y_i$ coordinates are ambiguous. Therefore 
\begin{align*}
\mathcal{M} &= \{0,1\}^{2n}\times \{0\}^t,\\
\mathcal{M}' &= \{01, 10\}^{n}\times \{0\}^t.
\end{align*}

Define 
\[\mathcal{H}(\mu,\mathcal{A}) := \prod_{a\in \mathcal{A}} f(H(\mu, a)).\]
Similarly define $\mathcal{H}(\mu, \mathcal{B}_i)$ and $\mathcal{H}(\mu, \mathcal{C}'_j)$. Set 
\[\mathcal{H}(\mu) := \mathcal{H}(\mu,\mathcal{A}) \cdot \prod_{i\in [n]} \mathcal{H}(\mu,\mathcal{B}_i) \cdot \prod_{j\in [k]} \mathcal{H}(\mu,\mathcal{C}'_j) .\]

For each $\mu\in \mathcal{M}$, we obtain a lower bound for $\h(\mu)$ and for each $\mu\in \M'_{\Gamma}$ we describe an exact value for $\h(\mu)$. Divide $\mathcal{M}$ into 4 classes using the following three properties which a median $\mu\in\M$ may have. 

\begin{itemize}[leftmargin=1in]
\item [\textit{Property 1}.] $\sum_{i\in [n]} (\mu[x_i] + \mu[y_i]) = n$.
\item [\textit{Property 2}.] $\mu\in\mathcal{M}'$.
\item[\textit{Property 3}.] $\mu \in \mathcal{M}'_\Gamma$. 
\end{itemize}

Notice that these properties are nested. Any median $\mu\in\M$ with Property 2, must also have Property 1. Further, any $\mu\in \M$ with Property 3 must also have Property 2. The following claims provide lower bounds for medians according to their properties. 

\begin{claim}
If $\mu\in \M$ has Property 1, then 
\begin{align}
 \h(\mu, \mathcal{A}) = & [f(n)]^{2k} [f(n+1)]^{16k} [f(n+2)]^{36k} \nonumber \\
 		& \cdot [f(n+3)]^{36k} [f(n+4)]^{16k} [f(n+5)]^{2k}  \label{alpha good} \\
 =:& \alpha_{good} . \nonumber 
 \end{align}
 Otherwise, 
 \begin{align}
 \h(\mu, \mathcal{A}) \geq 
 &[f(n-1)f(n+1)]^{k} [f(n)f(n+2)]^{8k} [f(n+1)f(n+3)]^{18k} \nonumber \\
 &\cdot [f(n+2)f(n+4)]^{18k} [f(n+3)f(n+5)]^{8k} [f(n+4)f(n+6)]^{k}  \label{alpha bad}\\
 =: & \alpha_{bad}. \nonumber 
 \end{align}
 \label{alpha}
\end{claim}

\begin{proof}
If $\mu\in \mathcal{M}$ has Property 1, then $H(\mu,\alpha^{(+0)})=H(\mu, \overline\alpha^{(+0)})=n$ because $\alpha^{(+0)}[x_i] = \alpha^{(+0)}[y_i] =1$ for all $i\in [n]$ while $\mu$ only has $n$ ones in the first 2n entries. Because $\mu[e_i]=0$ for all $i\in [t]$, by Definition~\ref{defi: ones}, 
\[H(\mu,\alpha^{(+a)})=H(\mu, \overline\alpha^{(+a)})=n+a.\]
Recalling the exact strings that appear in $\mathcal{A}$, we quickly obtain \eqref{alpha good}. 

If $\mu$ does not have Property 1, then either $\mu$ has more than $n$ ones in the first $2n$ entries, implying $H(\mu,\alpha^{(+0)})>n$, or $\mu$ has less than $n$ ones in the first $2n$ entries, implying $H(\mu,\overline\alpha^{(+0)})>n$. 
By Fact~\ref{fact: Hcomp}, $H(\mu,\alpha^{(+0)}) +  H(\mu,\overline\alpha^{(+0)}) = 2n$. By Fact~\ref{fact: min product} and the above observations, 
\begin{align*}
f\left(H\left(\mu,\alpha^{(+0)}\right)\right)\cdot f\left(H(\mu, \overline\alpha^{(+0)})\right)&\geq f(n-1)f(n+1),\\
f\left(H\left(\mu,\alpha^{(+a)}\right)\right) \cdot f \left(H\left(\mu, \overline\alpha^{(+a)}\right)\right)&\geq f(n-1 + a)f(n+1+a).
\end{align*}
Recalling the exact strings in $\mathcal{B}$, we obtain the lower bound in \eqref{alpha bad}. 
\end{proof}

\begin{claim}
For each $\mu\in \M$ and each $i\in[n]$,
\begin{align}
\h(\mu, \mathcal{B}_i) \geq 
& [f(n+1)]^{2k} [f(n+2)]^{12k} [f(n+3)]^{12k} [f(n+4)]^{2k} \label{beta good}
=:  \beta_{good} .
\end{align}
If $\mu$ has Property 2, then for every $i\in [n]$,
\begin{align*}
\h(\mu, \mathcal{B}_i) =: 
 \beta_{good}. \nonumber
\end{align*}
If $\mu$ satisfies Property 1, but not Property 2, then there exists $i_0\in [n]$ such that
\begin{align}
\h(\mu, \mathcal{B}_{i_0}) 
= &
[ f(n-1)f(n+3) ]^{k} [ f(n)f(n+4) ]^{6k} \nonumber\\
& \cdot [ f(n+1)f(n+5) ]^{6k} [ f(n+2)f(n+6)]^{k} \label{beta bad}\\
=: &\beta_{bad} . \nonumber
\end{align}
\label{beta}
\end{claim}
\begin{proof}
For any $\mu\in \M$, by Fact~\ref{fact: Hcomp}, 
\[H(\mu, \beta^{(+0)}) + H(\mu, \overline\beta^{(+0)})  =  2n.\] 
By Fact~\ref{fact: min product}, for each $a\in \mathbb{Z}^{+} \cup \{0\}$,
\begin{align*}
f\left(H\left(\mu, \beta^{(+0)}\right)\right) &\cdot f\left( H\left(\mu, \overline\beta^{(+0)}\right)\right) \geq \left[f(n)\right]^2, \text{ and}\\
f \left(H\left(\mu, \beta^{(+a)}\right) \right) &\cdot f\left(H\left(\mu, \overline\beta^{(+a)}\right)\right) \geq \left[f(n+a)\right]^2.
\end{align*}
Therefore, for any $\mu\in \M$, 
\[\h(\mu, \mathcal{B}_i)  \geq \beta_{good}.\]

If $\mu\in \M'$, then for each $i\in [n]$, $\mu[x_i] \neq \mu[y_i]$. On the other hand, for each $i\in [n]$, $j\in [n]$, $\beta_i^{(+a)}[x_j] = \beta_i^{(+a)}[y_j]$. Therefore for any $i,j\in [n]$, 
\[H((\mu[x_j],\mu[y_j]), (\beta_i^{(+a)}[x_j] , \beta_i^{(+a)}[y_j]) ) =1.\] The same holds if $\beta_i^{(+a)}$ is replaced with $\overline\beta_i^{(+a)}$. Therefore,
\begin{align*}
H\left(\mu, \beta_i^{(+0)}\right) &=H\left(\mu, \overline\beta_i^{(+0)}\right) = n,  \\
H\left(\mu, \beta_i^{(+a)}\right) &= H\left(\mu, \overline\beta_i^{(+a)}\right) =n+a.
\end{align*}
As a result $\h(\mu) = \beta_{good}$. 

If $\mu$ satisfies Property 1 but not Property 2, then we can define a tighter lower bound on $\h(\mu, \mathcal{B}_i)$. In particular, because $\mu\not\in \M'$, there exists $i_0\in [n]$ such that $\mu[x_{i_0}] = \mu[y_{i_0}]$. Recall $\beta_{i_0}^{(+a)}[x_{i_0}] = \beta_{i_0}^{(+a)}[y_{i_0}]=1$ and $\overline\beta_{i_0}^{(+a)}[x_{i_0}] = \overline\beta_{i_0}^{(+a)}[y_{i_0}]=0$. Therefore, 
\begin{align*}
\mu[x_{i_0}] =1 \Rightarrow & H((\mu[x_{i_0}], \mu[y_{i_0}]), (\beta_{i_0}^{(+a)}[x_{i_0}], \beta_{i_0}^{(+a)}[y_{i_0}]))=0, \\
&H((\mu[x_{i_0}], \mu[y_{i_0}]), (\overline\beta_{i_0}^{(+a)}[x_{i_0}], \overline\beta_{i_0}^{(+a)}[y_{i_0}]))=2,\text{ and}\\
\mu[x_{i_0}] =0 \Rightarrow & H((\mu[x_{i_0}], \mu[y_{i_0}]), (\overline\beta_{i_0}^{(+a)}[x_{i_0}], \overline\beta_{i_0}^{(+a)}[y_{i_0}]))=0, \\
&H((\mu[x_{i_0}], \mu[y_{i_0}]), (\beta_{i_0}^{(+a)}[x_{i_0}], \beta_{i_0}^{(+a)}[y_{i_0}]))=2.
\end{align*}

Because $\mu$ satisfies Property 1, there are exactly $n$ ones among the first $2n$ coordinates. Without loss of generality, $\mu[x_{i_0}] = \mu[y_{i_0}]=1$. Set \[S:=\{x_j, y_j : j\in [n], j\neq i_0\}.\]  Then $\mu$ has $n-2$ ones and $n$ zeros among the coordinates in $S$. However, $\beta_{i_0}^{(+a)}$ takes the value 0 on each of the coordinates of $S$ and $\overline\beta_{i_0}^{(+a)}$ takes the value 1 on the coordinates of $S$. Therefore,  
\begin{align*}
\mu[x_{i_0}] =1 \Rightarrow & H(\mu, \beta_{i_0}^{(+0)})=0 + (n-2), \\
&H(\mu, \overline\beta_{i_0}^{(+0)})=2 + n,\text{ and}\\
\mu[x_{i_0}] =0 \Rightarrow & H(\mu, \overline\beta_{i_0}^{(+0)})=0 + (n-2), \\
&H(\mu, \beta_{i_0}^{(+0)})=2 + n.
\end{align*}
As a result, 
\begin{align*}
H(\mu, \beta_{i_0}^{(+0)})H(\mu, \overline\beta_{i_0}^{(+0)}) & = (n-2)(n+2),\\
H(\mu, \beta_{i_0}^{(+a)})H(\mu, \overline\beta_{i_0}^{(+a)}) & = (n-2+a)(n+2+a).
\end{align*}
Taking into account all binary strings in $\B_{i_0}$, we conclude 
$\h(\mu, \mathcal{B}_{i_0}) =\beta_{bad}$ in \eqref{beta bad}.
\end{proof}

\begin{fact}
For the quantities defined in Claim~\ref{beta}, $\beta_{good}<\beta_{bad}$. 
Consequently, if $\mu\in \M \setminus \M'$ and satisfies Property 1, then $\prod_{i\in[n]} \h(\mu, \mathcal{B}_i) \geq \beta_{bad} \beta^{k-1}_{good}.$ If $\mu\in \M \setminus \M'$ and does not satisfy Property 1, then $\prod_{i\in[n]} \h(\mu, \mathcal{B}_i) \geq\beta^{k}_{good}$.
\label{beta comp}
\end{fact}

\begin{proof}
Observe 
\begin{align*}
\frac{\beta_{bad}}{\beta_{good}}
&= \left[ \frac{f(n-1)f(n)^6 f(n+1)^4 f(n+4)^4 f(n+5)^6 f(n+6)}{f(n+2)^{11} f(n+3)^{11}}\right]^k \\
&=\left[ \frac{f(n-1)f(n+6)}{f(n+2) f(n+3)}\right]^k  \left[ \frac{f(n)  f(n+5)}{f(n+2) f(n+3)}\right]^{6k}
\left[ \frac{ f(n+1)f(n+4)}{f(n+2) f(n+3)}\right]^{4k} \\
& >1
\end{align*}
where the last inequality follows from Lemma~\ref{concavity2}.
\end{proof}

\begin{claim}
For any $\mu\in \M$ and for each $j\in [k]$, 
\[ \h(\mu, \mathcal{C}'_i) \geq [f(n+2)]^{25} [f(n+3)]^{25}=: \gamma_{min}.\]
If $\mu\in \M'_\Gamma$, then  for each $j\in [k]$,
\begin{align}
\h(\mu, \mathcal{C}'_i) = 
& [f(n)]^7 [f(n+1)]^6 [f(n+2)]^{12}  [f(n+3)]^{12} [f(n+4)]^6 [f(n+5)]^7
=:\gamma_{good}. \label{gamma good}
\end{align}
If $\mu\in\M' \setminus \M'_\Gamma$, then there exists $i_0\in [k]$ such that
\begin{align}
\h(\mu, \mathcal{C}'_{i_0}) = 
& f(n-1) [f(n)]^6 [f(n+1)]^3 [f(n+2)]^{15} \nonumber \\
 & \cdot [f(n+3)]^{15} [f(n+4)]^3 [f(n+5)]^6 f(n+6)  \label{gamma bad}\\
 =&:  \gamma_{bad}. \nonumber
\end{align}
\label{gamma}
\end{claim}
\begin{proof}
Let $\mu$ be an arbitrary median in $\M$. By Remark~\ref{table pairs}, the binary strings in $\mathcal{C}_i$ come in pairs that are complementary on the first $2n$ entries. With a careful examination of Table~\ref{table: n+5}, if $\eta, \eta'\in \mathcal{C}_i$ are complementary on the first $2n$ coordinates, then $e(\eta) + e(\eta') = 3$ where $e$ is the function specifying the number of additional ones. By the definition of $\mathcal{C}'_i$, the strings still come in complementary pairs, $(\hat\eta, \hat\eta')$, but here $e(\hat\eta) + e( \hat\eta') = 5$ because the number of additional ones in $\hat\eta$ and $\hat\eta'$ is precisely one more than the number in $\eta$ and $\eta'$. By Fact~\ref{fact: Hcomp}, for each of the  25 pairs in $\mathcal{C}'_i$,
\[H(\mu, \hat\eta) + H(\mu, \hat\eta') = 2n+5.\]
Then by Fact~\ref{fact: min product}, 
\[f(H(\mu, \hat\eta)) f(H(\mu, \hat\eta')) \geq f(n+2) f(n+3)\]
which gives the general bound $\gamma_{min}$. 

Now suppose $\mu\in \M'_\Gamma$. This implies $\mu\in \M'_{c_i}$ for all clauses $c_i$ in $\Gamma$. By the definition of $\mathcal{C}'_i$, for each $\hat\nu_j^i\in \mathcal{C}'_i$, 
$H(\mu, \hat\nu_j^i) = H(\mu, \nu_j^i)+1$ where $\nu_j^i\in \mathcal{C}_i$. From \eqref{dist for sat}, we see 
\[ \{H(\mu,\hat\nu_j^i): j\in[50]\} = \{n_{(7)}, (n+1)_{(6)}, (n+2)_{(12)}, (n+3)_{(12)}, (n+4)_{(6)}, (n+5)_{(7)}\} . \]
This immediately implies $\h(\mu, \mathcal{C}_i') = \gamma_{good}$ in \eqref{gamma good}.

Finally, suppose $\mu\in \M' \setminus \M'_{\Gamma}$. Using the bijection in Definition~\ref{defi: bijection}, $\mu$ must correspond to a truth assignment which does not satisfy $\Gamma$. So there is a clause $c_{i_0}$ in $\Gamma$ which is not satisfied. Therefore $\mu\in \M' \setminus \M'_{c_{i_0}}$. From \eqref{dist for nonsat}, adding 1 to each $H(\mu, \nu_j^{i_0})$ to obtain $ H(\mu, \hat\nu_j^{i_0}) $, we obtain
\begin{align}
\{H(\mu,\nu_j^{i_0}): j\in[50]\} = \{&(n-1)_{(1)}, n_{(6)}, (n+1)_{(3)}, (n+2)_{(15)},  \nonumber \\ & (n+3)_{(15)}, (n+4)_{(3)}, (n+5)_{(6)}, (n+6)_{(1)} \}.
\end{align}
This directly implies $\h(\mu, \mathcal{C}_{i_0}) = \gamma_{bad}$ in \eqref{gamma bad}.
\end{proof}

\begin{fact}
For the quantities defined in Claim~\ref{gamma}, $\gamma_{good}< \gamma_{bad}$. As a result, when $\mu\in \M'\setminus \M'_{\Gamma}$, 
 \[\h(\mu, \mathcal{C}) \geq \gamma_{bad}\gamma_{good}^{k-1}.\]
\label{gamma comp}
\end{fact}
\begin{proof}
Indeed, this was our initial assumption: 
\begin{align*}
\frac{\gamma_{bad}}{\gamma_{good}} 
&= \frac{ f(n-1)[f(n+2)]^3[f(n+3)]^3 f(n+6) }{ f(n)[f(n+1)]^3 [f(n+4)]^3[f(n+5)] }>1. 
\end{align*}
The bound for $\h(\mu, \mathcal{C})$ results from the fact that $\mu\in \M'$ either corresponds to a satisfying truth assignment for $c_i$ or a non-satisfying truth assignment for each clause $c_i$. 
\end{proof}

In summary, Claims~\ref{alpha}, ~\ref{beta}, and ~\ref{gamma} along with Facts~\ref{beta comp} and ~\ref{gamma comp}, we give the following bounds. 
If $\mu\in \M'_\Gamma$, 
\begin{align*}
 \h(\mu) = \alpha_{good} \beta_{good}^n \gamma_{good}^k =: h_3.
\end{align*}
If $\mu\in \M' \setminus \M'_\Gamma$, 
\begin{align*}
 \h(\mu) \geq \alpha_{good} \beta_{good}^n \gamma_{bad}\gamma_{good}^{k-1} =: h_2.
\end{align*}
If $\mu\in \M \setminus \M'$ and has Property 1, 
\begin{align*}
\h(\mu) \geq \alpha_{good} \beta_{bad}\beta_{good}^{n-1} \gamma_{min}^k =: h_1.
\end{align*}
If $\mu\in \M$ but does not have Property 1, 
\begin{align*}
\h(\mu) \geq \alpha_{bad} \beta_{good}^{n} \gamma_{min}^k =: h_0.
\end{align*}

In order to complete, the proof, we only need to show $h_3 < h_i$ for $i\in \{0,1,2\}$.
By one of our assumptions about $f(x)$, we have already verified in Fact~\ref{gamma comp} that
\begin{align*}
\frac{h_2}{h_3} 
= \frac{\gamma_{bad}}{\gamma_{good}} > 1. 
\end{align*}

Next observe 
\begin{align*}
\frac{h_1}{h_2}
= & \frac{\beta_{bad}}{\beta_{good}}\cdot \frac{\gamma_{min}^k}{\gamma_{bad} \gamma_{good}^{k-1}}\\
> &\frac{\beta_{bad}}{\beta_{good}}\cdot \frac{\gamma_{min}^k}{\gamma_{bad} ^k}\\ 
 =& \left[ \frac{f(n-1)f(n+3)}{ [f(n+1)]^2}   \left[\frac{f(n)f(n+4)}{ [f(n+2)]^2}\right]^6   
    \left[\frac{f(n+1)f(n+5)}{ [f(n+3)]^2}\right]^6   \frac{f(n+2)f(n+6)}{ [f(n+4)]^2}   \right] ^k\\
  &
    \cdot\left[\frac{[f(n+2)]^{10} [f(n+3)]^{10} }{f(n-1) [f(n)]^6 [f(n+1)]^3 [f(n+4)]^3 [f(n+5)]^6 f(n+6)}\right]^k\\
 =& \left[\frac{f(n+1) f(n+4)}{ f(n+2) f(n+3) }\right]^k\\
 >& 1 
\end{align*}
where the last inequality follows from Lemma~\ref{concavity2}.

Finally we prove that $h_0 > h_2$. 
\begin{align*}
\frac{h_0}{h_2}
= & \frac{\alpha_{bad}}{\alpha_{good}}\frac{\gamma_{min}^k}{\gamma_{bad} \gamma_{good}^{k-1}}\\
> & \frac{\alpha_{bad}}{\alpha_{good}}\cdot \frac{\gamma_{min}^k}{\gamma_{bad} ^k}\\
= & \left[ \frac{f(n-1)f(n+1)}{[f(n)]^2} \left[\frac{f(n)f(n+2)}{[f(n+1)]^2}\right]^8
   \left[\frac{f(n+1)f(n+3)}{[f(n+2)]^2}\right]^{18} \right. \\
&   \cdot \left.\left[\frac{f(n+2)f(n+4)}{[f(n+3)]^2}\right]^{18} 
  \left[\frac{f(n+3)f(n+5)}{[f(n+4)]^2}\right]^{8} \frac{f(n+4)f(n+6)}{[f(n+5)]^2}\right]^k\\
  & \cdot \left[\frac{[f(n+2)]^{10} [f(n+3)]^{10} }{f(n-1) [f(n)]^6 [f(n+1)]^3 [f(n+4)]^3 [f(n+5)]^6 f(n+6)}\right]^k\\[.5em]
 = & \frac{[f(n-1)]^{k} [f(n)]^{8k} [f(n+1)]^{19k} [f(n+2)]^{26k}}
 		{ [f(n)]^{2k} [f(n+1)]^{16k} [f(n+2)]^{36k}}\\[.5em]
	& \cdot \frac{ [f(n+3)]^{26k} [f(n+4)]^{19k} [f(n+5)]^{8k} [f(n+6)]^{k}}
 		{[f(n+3)]^{36k} [f(n+4)]^{16k} [f(n+5)]^{2k}}\\[.5em]
		& \cdot \frac{[f(n+2)]^{10k} [f(n+3)]^{10k} }{[f(n-1)]^{k} [f(n)]^{6k} [f(n+1)]^{3k} [f(n+4)]^{3k} [f(n+5)]^{6k} f(n+6)^{k}}\\
 = & 1.
\end{align*}

Therefore for any $\hat\mu \in \M'_\Gamma$ and $\mu\in \M \setminus\M'_\Gamma $, then $\h(\hat\mu)< \h(\mu)$. Thus, if we could determine, in polynomial time, how many medians $\mu\in \mathcal{M}$ have  $h(\mu)\leq h_3$, then we could determine how many satisfying truth assignments exist for $\Gamma$ in polynomial time.   
\end{proof}

\begin{cor}
Fix a function $f(x): \mathbb{Z}^{+} \cup \{0\} \rightarrow [0,\infty)$  which satisfies the following properties:
\begin{itemize}
\item $\log f(x)$ is strictly concave up, 
\item the function values of $f$ can be computed in polynomial time, and 
\item  for all but finitely many $n\in \mathbb{Z}$, $n\geq 2$, 
\[ \frac{f(n-2)[f(n+1)]^3[f(n+2)]^3 f(n+5)}{f(n-1)[f(n)]^3[f(n+3)]^3f(n+4)} > 1.\]
\end{itemize}
For arbitrary $m,s\in \mathbb{Z}^{+}$ and $D\in \mathbb{R}$, let $S:=\{\nu_1, \nu_2, \ldots, \nu_m\}$ be a multiset of binary strings, each of length $s$ and let $\mathcal{M}$ be the set of medians for $S$.
Then it is NP-complete to determine if   
\begin{eqnarray}\min_{\mu\in \Gamma}\prod_{i\in[m]} f\left(H(\nu_i,\mu)\right) \leq D . \end{eqnarray}
\label{prop: concave up}
\end{cor}

This next theorem gives the same result as Theorem~\ref{NP concave up} with one change in the conditions on $f$. While Theorem~\ref{NP concave up} required that 
\[ \frac{f(n-2)[f(n+1)]^3[f(n+2)]^3 f(n+5)}{f(n-1)[f(n)]^3[f(n+3)]^3f(n+4)} > 1,\] Theorem~\ref{NP concave up2} switches the inequality to consider functions in which the ratio is less than 1. 

\begin{thm}
 Fix a function $f(x): \mathbb{Z}^{+} \cup \{0\} \rightarrow [0,\infty)$  which satisfies the following properties:
\begin{itemize}
\item $\log f(x)$ is strictly concave up,
\item the function values of $f$ can be computed in polynomial time, and 
\item  for all but finitely many $n\in \mathbb{Z}$, $n\geq 2$, 
\[ \frac{f(n-2)[f(n+1)]^3[f(n+2)]^3 f(n+5)}{f(n-1)[f(n)]^3[f(n+3)]^3f(n+4)} < 1.\]
\end{itemize}
For arbitrary $m,s\in \mathbb{N}$ and $D\in \mathbb{R}$, let $S:=\{\nu_1, \nu_2, \ldots, \nu_m\}$ be a multiset of binary strings, each of length $s$.
Then it is \#P-complete to determine how many medians $\mu$ for $S$ have  
\[\prod_{i\in[m]} f\left(H(\nu_i,\mu)\right) \leq D .\]
\label{NP concave up2}
\end{thm}

\begin{proof}
This proof closely mirrors the proof of Theorem~\ref{NP concave up}. Here we will note the changes that need to be made. 

This time, we define $98k+24kn$ binary strings, each of length $2n+245k+60kn$ with coordinates 
\[(x_1, y_1, \ldots, x_n, y_n, e_1, \ldots, e_t).\] 

Let $\alpha^{(+a)}$ and $\overline\alpha^{(+a)}$ be defined as before. The collection ${\mathcal{A}}$ will now consist of the following $72k$ strings: 
\begin{itemize}
\item $4k$ copies each of $\alpha^{(+1)}$ and $\overline\alpha^{(+1)}$,
\item $14k$ copies each of $\alpha^{(+2)}$ and $\overline\alpha^{(+2)}$,
\item $14k$ copies each of $\alpha^{(+3)}$ and $\overline\alpha^{(+3)}$,
\item $4k$ copies each of $\alpha^{(+4)}$ and $\overline\alpha^{(+4)}$.
\end{itemize}

Define $\beta_i^{(+a)}$ and $\overline\beta_i^{(+a)}$ as before. The collection ${\mathcal{B}_i}$ now consists of the following $24k$ strings: 
\begin{itemize}
\item $6k$ copies each of $\beta_i^{(+2)}$ and $\overline\beta_i^{(+2)}$,
\item $6k$ copies each of $\beta_i^{(+3)}$ and $\overline\beta_i^{(+3)}$.
\end{itemize}

Following the explanation found in Section~\ref{section: define tables}, Table \ref{table: n+5 complement} defines 26 binary strings $\overline{\mathcal{C}_i}$ for a clause. As in the proof of Theorem~\ref{NP concave up}, we will add 1 additional one to each of the 26 strings in $\overline{\mathcal{C}_i}$ to create $\mathcal{C}'_i$. 

Using the same Properties 1, 2, and 3 as before, we obtain the following values which are analogous to the bounds in Claims~\ref{alpha}, ~\ref{beta}, and ~\ref{gamma}: 
\begin{align*}
\alpha_{good} :=& [f(n+1)]^{8k} [f(n+2)]^{28k} [f(n+3)]^{28k} [f(n+4)]^{8k},  \\
\alpha_{bad} :=&   [f(n)f(n+2)]^{4k} [f(n+1)f(n+3)]^{14k}\\
			& \cdot [f(n+2)f(n+4)]^{14k} [f(n+3)f(n+5)]^{4k} ,
			\end{align*}
			\begin{align*}
\beta_{good} :=&  [f(n+2)]^{12k} [f(n+3)]^{12k}, \\
\beta_{min} :=&    [f(n+2)]^{12k} [f(n+3)]^{12k},\\
\beta_{bad} :=&   [f(n)f(n+4)]^{6k} [f(n+1)f(n+5)]^{6k}, \\
\gamma_{good} :=& f(n-1) [f(n)]^{3} [f(n+1)]^{3} [f(n+2)]^{6} \\
			& \cdot [f(n+3)]^{6} [f(n+4)]^{3} [f(n+5)]^{3} f(n+6),  \\
\gamma_{bad} :=&  [f(n)]^{4}[f(n+1)]^6 [f(n+2)]^{3} [f(n+3)]^{3} [f(n+4)]^{6} [f(n+5)]^{4}, \\
\gamma_{min} :=&  [f(n+2)]^{13} [f(n+3)]^{13}. \\
\end{align*}

By our assumption about $f(x)$,
\begin{align*}
\frac{\gamma_{bad}}{\gamma_{good}}
&= \frac{ f(n) [f(n+1)]^3 [f(n+4)]^3 f(n+5)}{ f(n-1) [f(n+2)]^3 [f(n+3)]^3 f(n+6)}>1
\end{align*}
which implies $\gamma_{bad}>\gamma_{good}$. 

Next we determine the values of $h_0, h_1, h_2, h_3$ in this setting. As before, if $\mu$ has Property $i$, but not property $i+1$, then $\h(\mu)\geq h_i$. Further, if $\mu$ has Property 3, $\h(\mu) = h_3$.  If $\mu$ does not have Property 1, then $\h(\mu)\geq h_0$. 
\begin{align*}
h_3 &:=\alpha_{good} \beta_{good}^n \gamma_{good}^k. \\
h_2 &:= \alpha_{good} \beta_{good}^n \gamma_{bad}\gamma_{good}^{k-1}. \\
h_1 &:=\alpha_{good} \beta_{bad}\beta_{min}^{n-1} \gamma_{min}^k. \\
h_0 &:= \alpha_{bad} \beta_{min}^{n} \gamma_{min}^k.
\end{align*}
As in the proof of Theorem~\ref{NP concave up}, we will show $h_3 <h_0, h_1, h_2$. 

By our assumption about $f(x)$, 
\begin{align*}
\frac{ h_2}{ h_3} 
= \frac{ \gamma_{bad}}{ \gamma_{good}} > 1.
\end{align*}
Also
\begin{align*}
\frac{ h_1}{ h_2}
= & \frac{ \beta_{bad}}{\beta_{good}}\cdot \frac{\gamma_{min}^k}{\gamma_{bad} \gamma_{good}^{k-1}}\\
> &\frac{\beta_{bad}}{\beta_{good}}\cdot \frac{\gamma_{min}^k}{\gamma_{bad} ^k}\\
=& \left[\left[\frac{f(n)f(n+4)}{[f(n+2)]^2}\right]^{6}\left[\frac{f(n+1)f(n+5)}{[f(n+3)]^2}\right]^{6}\right]^k
	\\&\cdot\left[\frac{ [f(n+2)]^{10} [f(n+3)]^{10} }{ [f(n)]^4 [f(n+1)]^6 [f(n+4)]^6 [f(n+5)]^4}\right]^k\\
=& \left[ \frac{ f(n) f(n+5) }{ f(n+2) f(n+3) } \right]^{2k}\\
> & 1 \tag*{by Lemma~\ref{concavity2}.} 
\end{align*}

Finally we show $h_0 >h_2$.
\begin{align*}
\frac{ h_0}{ h_2}
= & \frac{ \alpha_{bad}}{\alpha_{good}}\frac{\gamma_{min}^k}{\gamma_{bad} \gamma_{good}^{k-1}}\\
> & \frac{\alpha_{bad}}{\alpha_{good}}\cdot \frac{\gamma_{min}^k}{\gamma_{bad} ^k}\\
=& \left[ \left[\frac{f(n)f(n+2)}{[f(n+1)]^2}\right]^4  \left[\frac{f(n+1)f(n+3)}{[f(n+2)]^{2} }\right]^{14} \right.\\& \cdot \left. \left[\frac{f(n+2)f(n+4)}{[f(n+3)]^{2}}\right]^{14} \left[\frac{f(n+3)f(n+5)}{[f(n+4)]^2}\right]^4	\right]^k\\
	&\cdot \left[\frac{ [f(n+2)]^{10} [f(n+3)]^{10} }{ [f(n)]^4 [f(n+1)]^6 [f(n+4)]^6 [f(n+5)]^4}\right]^k
	\\[.5em]
=& \frac{[f(n)]^{4k} [f(n+1)]^{14k} [f(n+2)]^{18k} [f(n+3)]^{18k} [f(n+4)]^{14k} [f(n+5)]^{4k}}
	{[f(n+1)]^{8k} [f(n+2)]^{28k} [f(n+3)]^{28k} [f(n+4)]^{8k}}\\[.5em]
	&\cdot \frac{ [f(n+2)]^{10k} [f(n+3)]^{10k} }{ [f(n)]^{4k} [f(n+1)]^{6k} [f(n+4)]^{6k} [f(n+5)]^{4k}} \\
=& 1.	
\end{align*}
Making each of these changes in the proof of Theorem~\ref{NP concave up}, we complete the proof of Theorem~\ref{NP concave up2}. 
\end{proof} 

\begin{table}[ht]
\centering
\begin{threeparttable}
\caption{The 26 strings to complement the collection in Table \ref{table: n+5} along with their Hamming distance from medians in $\mathcal{M}'$. }
\begin{tabular}{c}
\resizebox{5.5in}{!}{
$
\begin{array}{|c|c|c|c|||c|c|c|c|c|c|c|c|}
\hline

\cline{2-12}
 & \textbf{A} & \textbf{B} & \textbf{C} & \textbf{M1} &\textbf{M2}&\textbf{M3}&\textbf{M4}&\textbf{M5}&\textbf{M6}&\textbf{M7}&\textbf{M8}\\ 
\hline
&\text{Values of}& \nu^i_j[x_\ell],&& 
10\,10\,10& 10\,10\,01&10\,01\,10&   01\,10\,10&10\, 01\, 01&01\,10\,01&01\,01\,10&01\, 01\, 01 \\
\text{Row}&\nu^i_j \text{ on its} & \nu^i_j[y_\ell]&\text{Add'l} &&&&&&&&\\
\text{\#}&\text{support set}& (v_\ell\not\in c_i)& \text{Ones} &&&&&&&&\\
\hline\hline \hline
1&01\,00\,00&0&+0		&n+1&n+1&n+1&n-1&n+1&n-1&n-1&n-1\\
2&00\,01\,00&0&+0		&n+1&n+1&n-1&n+1&n-1&n+1&n-1&n-1\\
3&00\,00\,01&0&+0		&n+1&n-1&n+1&n+1&n-1&n-1&n+1&n-1\\
\hline
4&10\,11\,11&1&+3		&n+2&n+2&n+2&n+4&n+2&n+4&n+4&n+4\\
5&11\,10\,11&1&+3		&n+2&n+2&n+4&n+2&n+4&n+2&n+4&n+4\\
6&11\,11\,10&1&+3		&n+2&n+4&n+2&n+2&n+4&n+4&n+2&n+4\\
\hline\hline
7&01\,01\,00&0&+2		&n+4&n+4&n+2&n+2&n+2&n+2&n&n\\
8&01\,00\,01&0&+2		&n+4&n+2&n+4&n+2&n+2&n&n+2&n\\
9&00\,01\,01&0&+2		&n+4&n+2&n+2&n+4&n&n+2&n+2&n\\
\hline
10&10\,10\,11&1&+1		&n-1&n-1&n+1&n+1&n+1&n+1&n+3&n+3\\
11&10\,11\,10&1&+1		&n-1&n+1&n-1&n+1&n+1&n+3&n+1&n+3\\
12&11\,10\,10&1&+1		&n-1&n+1&n+1&n-1&n+3&n+1&n+1&n+3\\
\hline\hline
13&10\,01\,01&0&+1		&n+2&n&n&n+4&n-2&n+2&n+2&n\\
14&01\,10\,01&0&+1		&n+2&n&n+4&n&n+2&n-2&n+2&n\\
15&01\,01\,10&0&+1		&n+2&n+4&n&n&n+2&n+2&n-2&n\\
16&10\,10\,01&0&+1		&n&n-2&n+2&n+2&n&n&n+4&n+2\\
17&10\,01\,10&0&+1		&n&n+2&n-2&n+2&n&n+4&n&n+2\\
18&01\,10\,10&0&+1		&n&n+2&n+2&n-2&n+4&n&n&n+2\\
19&10\,10\,10&0&+1		&n-2&n&n&n&n+2&n+2&n+2&n+4\\
\hline
20&01\,01\,01&1&+2		&n+5&n+3&n+3&n+3&n+1&n+1&n+1&n-1\\
21&10\,01\,01&1&+2		&n+3&n+1&n+1&n+5&n-1&n+3&n+3&n+1\\
22&01\,10\,01&1&+2		&n+3&n+1&n+5&n+1&n+3&n-1&n+3&n+1\\
23&01\,01\,10&1&+2		&n+3&n+5&n+1&n+1&n+3&n+3&n-1&n+1\\
24&10\,10\,01&1&+2		&n+1&n-1&n+3&n+3&n+1&n+1&n+5&n+3\\
25&10\,01\,10&1&+2		&n+1&n+3&n-1&n+3&n+1&n+5&n+1&n+3\\
26&01\,10\,10&1&+2		&n+1&n+3&n+3&n-1&n+5&n+1&n+1&n+3\\
\hline
\end{array}
$
}
\end{tabular}
\begin{tablenotes}[flushleft]
\item[] The information in this table is to be read in the same way as the information in Table~\ref{table: n+5}. This is detailed in Section~\ref{section: define tables}.
\end{tablenotes}
\end{threeparttable}
\label{table: n+5 complement}
\end{table}
\begin{cor}
Fix a function $f(x): \mathbb{Z}^{+} \cup \{0\} \rightarrow [0,\infty)$  which satisfies the following properties:
\begin{itemize}
\item $\log f(x)$ is strictly concave up, 
\item the function values of $f$ can be computed in polynomial time, and 
\item  for all but finitely many $n\in \mathbb{Z}$, $n\geq 2$, 
\[ \frac{f(n-2)[f(n+1)]^3[f(n+2)]^3 f(n+5)}{f(n-1)[f(n)]^3[f(n+3)]^3f(n+4)} < 1.\]
\end{itemize}
For arbitrary $m,s\in \mathbb{Z}^{+}$ and $D\in \mathbb{R}$, let $S:=\{\nu_1, \nu_2, \ldots, \nu_m\}$ be a multiset of binary strings, each of length $s$ and let $\mathcal{M}$ be the set of medians for $S$.
Then it is NP-complete to determine if 
\begin{eqnarray}\min_{\mu\in \mathcal{M}}\prod_{i\in[m]} f\left(H(\nu_i,\mu)\right) \leq D . \end{eqnarray}
\label{prop: concave up2}
\end{cor}

We can also state the following corollaries for functions which are strictly concave down. 

\begin{cor}
 Fix a function $f(x): \mathbb{Z}^{+} \cup \{0\} \rightarrow [0,\infty)$  which satisfies the following properties:
\begin{itemize}
\item $\log f(x)$ is strictly concave down, 
\item the function values can be computed in polynomial time, and 
\item  for all but finitely many $n\in \mathbb{Z}$, $n\geq 2$, 
\[ \frac{f(n-2)[f(n+1)]^3[f(n+2)]^3 f(n+5)}{f(n-1)[f(n)]^3[f(n+3)]^3f(n+4)} \neq 1.\]
\end{itemize}
For arbitrary $m,s\in \mathbb{Z}^+$ and $D\in \mathbb{R}$, let $S:=\{\nu_1, \nu_2, \ldots, \nu_m\}$ be a multiset of binary strings, each of length $s$.
Then it is \#P-complete to determine if how many medians $\mu$ for $S$ have  
\begin{eqnarray}\prod_{i\in[m]} f\left(H(\nu_i,\mu)\right) \geq D .\end{eqnarray}
\label{NP concave down}
\end{cor}

\begin{proof}
If the function $f(x)$ has the property that $\log f(x)$ is strictly concave down, then $\log \frac{1}{f(x)}$ is strictly concave up.  Therefore by Theorems~\ref{NP concave up} and ~\ref{NP concave up2} for the function $\frac{1}{f(x)}$, it is \#P-hard to determine the number of medians $\mu$ which satisfy $\prod_{i\in [m]}\frac{1}{f(H(\nu_i, \mu))} \leq \frac{1}{D}$. This is equivalent to asking for the number of medians $\mu$ have $\prod_{i\in [m]}{f(H(\nu_i, \mu))} \geq {D}$. 
\end{proof}

\begin{cor}
 Fix a function $f(x): \mathbb{Z}^{+} \cup \{0\} \rightarrow [0,\infty)$  which satisfies the following properties:
\begin{itemize}
\item $\log f(x)$ is strictly concave down, 
\item the function values of $f$ can be computed in polynomial time, and 
\item  for all but finitely many $n\in \mathbb{Z}$, $n\geq 2$, 
\[ \frac{f(n-2)[f(n+1)]^3[f(n+2)]^3 f(n+5)}{f(n-1)[f(n)]^3[f(n+3)]^3f(n+4)} \neq 1.\]
\end{itemize}
For arbitrary $m,s\in \mathbb{Z}^+$ and $D\in \mathbb{R}$, let $S:=\{\nu_1, \nu_2, \ldots, \nu_m\}$ be a multiset of binary strings, each of length $s$ where $\mathcal{M}$ is the set of medians for $S$.
Then it is NP-complete to determine if 
\begin{eqnarray}\min_{\mu\in \mathcal{M}} \prod_{i\in[m]} f\left(H(\nu_i,\mu)\right) \geq D . \end{eqnarray}
\end{cor}

\section{Stochastic Approximations for $Z(B,f(x))$}
\label{sec:general approx}
We have seen several proofs showing that it is hard to calculate many of these quantities. As in Section~\ref{sec:factorial approx}, we may further ask if any of these quantities can be approximated. We will again focus on approximations via an FPRAS (Definition~\ref{defn:FPRAS}).

Before stating our results, we define a couple more complexity classes for decision problems: 

\begin{defn}[\cite{gill}]
A decision problem, A, is in the class RP (randomized polynomial time) if there is a probabilistic Turing machine that runs in polynomial time in the size of the input, returns ``true'' with probability at least $\frac{1}{2}$ when the answer for A is true, and returns ``false'' with probability 1 when the answer for A is false. 
\label{defn: RP}
\end{defn}

\begin{defn}[\cite{gill}]
A decision problem, A, is in the class BPP (bounded-error probabilistic polynomial time) if there is a probabilistic Turing machine that runs in polynomial time in the size of the input, returns ``true'' with probability at least $\frac{2}{3}$ when the answer for A is true, and returns ``false'' with probability $\frac{2}{3}$ when the answer for A is false. 
\label{defn: BPP}
\end{defn}

One result connecting these classes is the following: 
\begin{thm}[\cite{papadi}]
If the intersection of NP and BPP is non-empty, then RP=NP. 
\label{thm: BPP}
\end{thm}

Note that each result below holds for functions $f(x)$ with $\log f(x)$ strictly concave down. The analogous results for the functions whose logarithm is concave up are still open. Our first result can be interpreted as sampling medians for \#SPSCJ with a probability distribution analogous to the number of scenarios, but dependent on $f(x)$.

\begin{thm}
Fix a function $f(x): \mathbb{Z}^{+} \cup \{0\} \rightarrow [0,\infty)$  which satisfies the following properties:
\begin{itemize}
\item $\log f(x)$ is strictly concave down,
\item the function values of $f$ can be computed in polynomial time, and 
\item  there exists $\epsilon>0$ such that for all but finitely many $n\in \mathbb{Z}$, $n\geq 2$, 
\[ \frac{f(n-2)[f(n+1)]^3[f(n+2)]^3 f(n+5)}{f(n-1)[f(n)]^3[f(n+3)]^3f(n+4)} < 1-\epsilon.\]
\end{itemize}
For arbitrary $m,s\in \mathbb{Z}^+$, let $S:=\{\nu_1, \nu_2, \ldots, \nu_m\}$ be a multiset of binary strings, each of length $s$. 
If there is a rapidly mixing Markov chain with stationary distribution proportional to $\prod_{i\in[m]} f\left(H(\nu_i,\mu)\right)$,  
then RP=NP. 
\label{thm: FPAUS1}
\end{thm}

\begin{proof}
Fix a function $f$ as described in the theorem. Because $\log f(x)$ is strictly concave down, $\log (f(x))^{-1}$ is strictly concave up. Set $g(x):=(f(x))^{-1}$.

Now recall the proof of Theorem~\ref{NP concave up} for strictly concave up functions. Take a D3CNF $\Gamma$ with $n$ variables and create a multiset of binary strings, $\mathcal{D}$. The set of medians for $\mathcal{D}$ is $\mathcal{M}=\{0,1\}^{2n} \times \{0\}^t$. There is a one-to-one correspondence between the medians in the subset $\mathcal{M}' = \{01,10\}^n \times \{0\}^t$ and the truth assignments for $\Gamma$. Those medians which correspond to satisfying truth assignments for $\Gamma$ form the set $\mathcal{M}'_{\Gamma}$. The multiset $\mathcal{D}$ is constructed so that each $\mu\in \mathcal{M}'_{\Gamma}$ has 
 \[\prod_{\eta \in \mathcal{D}} \frac{1}{f(H(\eta, \mu))} = \prod_{\eta \in \mathcal{D}} g(H(\eta, \mu)) = \alpha_{good} \beta_{good}^n \gamma_{good}^k =: h_3 \]
 and all other medians have 
 \[\prod_{\eta \in \mathcal{D}} \frac{1}{f(H(\eta, \mu))} = \prod_{\eta \in \mathcal{D}} g(H(\eta, \mu)) > \alpha_{good} \beta_{good}^n \gamma_{bad}\gamma_{good}^{k-1}=: h_2 .\]
 Equivalently, if $\mu\in \mathcal{M}'_\Gamma$, then
  \[\prod_{\eta \in \mathcal{D}} {f(H(\eta, \mu))} = \frac{1}{ h_3 }.\]
  Otherwise, 
  \[\prod_{\eta \in \mathcal{D}} {f(H(\eta, \mu))} < \frac{1}{ h_2 }.\]  
 
 Further, 
 \begin{align*}
 \frac{h_2}{h_3} 
 &= \frac{\gamma_{bad}}{\gamma_{good}} \\
 &= \frac{ g(n-1)[g(n+2)]^3[g(n+3)]^3 g(n+6) }{ g(n)[g(n+1)]^3 [g(n+4)]^3[g(n+5)] }  \\
 &=\frac{ f(n)[f(n+1)]^3 [f(n+4)]^3[f(n+5)] }{f(n-1)[f(n+2)]^3[f(n+3)]^3 f(n+6)}  \\
 &> \frac{1}{ 1-\epsilon}
 \end{align*}
  where the last inequality is a result of the assumption in the theorem statement.
 As a result 
 \[\frac{1}{h_2} \left(\frac{1}{1-\epsilon}\right) < \frac{1}{h_3}.\]
 
 Now select an integer $r$, dependent only on the values of $n$ and $\epsilon$, such that $\left(\frac{1}{1-\epsilon}\right)^r > 2^{2n+2}$.
 Create a new multiset $\mathcal{D}(r)$ of binary strings such that 
 \[\mathcal{D}(r) = \underbrace{\mathcal{D} \uplus \ldots \uplus \mathcal{D}}_{r \text{ times}}.\]
 
The set of medians for $\mathcal{D}(r)$ is the same as the set of medians for $\mathcal{D}$. However this time, for a median $\mu\in\mathcal{M}'_\Gamma$, 
\[\prod_{\eta \in \mathcal{D}(r)} {f(H(\eta, \mu))} = \left(\frac{1}{ h_3 }\right)^r.\]
Otherwise, if $\mu\in \M \setminus \M'_\Gamma$,
  \[\prod_{\eta \in \mathcal{D}(r)} {f(H(\eta, \mu))} < \left(\frac{1}{ h_2 }\right)^r.\]
 By the choice of $r$, 
 \[  \left( \frac{1}{h_2}  \right)^r 2^{2n} <\left( \frac{1}{h_2}  \right)^r 2^{2n+2} <\left( \frac{1}{h_2} \left(\frac{1}{1-\epsilon}\right)\right)^r < \left(\frac{1}{h_3}\right)^r.\]
Since $\mathcal{M}=\{0,1\}^{2n} \times \{0\}^t$, $|\M| = 2^{2n}$ and the above inequality shows that for each $\mu_0\in \mathcal{M}'_\Gamma$, 
\[\prod_{\eta \in \mathcal{D}} {f(H(\eta, \mu_0))} > \sum_{\mu\in \mathcal{M} \setminus \mathcal{M}'_{\Gamma}} \prod_{\eta \in \mathcal{D}} {f(H(\eta, \mu))}.\]
Further,
\[\prod_{\eta \in \mathcal{D}} {f(H(\eta, \mu_0))} > \frac{1}{2} \sum_{\mu\in \mathcal{M}} \prod_{\eta \in \mathcal{D}} {f(H(\eta, \mu))}.\]

 Now suppose that we had a rapidly mixing Markov chain on the medians for this instance as stated in the theorem. From the calculations above, it must sample medians which correspond to satisfying truth assignments for $\Gamma$ with probability at least $\frac{1}{2}$. This is precisely an RP for D3SAT. However, this immediately implies RP=NP because D3SAT is NP-complete. 
\end{proof}

The following theorem gives the same result as the last one for different functions $f$. In particular, it switches the inequality that $f$ is required to satisfy. 
\begin{thm}
Fix a function $f(x): \mathbb{Z}^{+} \cup \{0\} \rightarrow [0,\infty)$  which satisfies the following properties:
\begin{itemize}
\item $\log f(x)$ is strictly concave down,
\item the function values of $f$ can be computed in polynomial time, and 
\item  there exists $\epsilon>0$ such that for all but finitely many $n\in \mathbb{Z}$, $n\geq 2$, 
\[ \frac{f(n-2)[f(n+1)]^3[f(n+2)]^3 f(n+5)}{f(n-1)[f(n)]^3[f(n+3)]^3f(n+4)} > 1+\epsilon.\]
\end{itemize}
For arbitrary $m,s\in \mathbb{Z}^+$, let $S:=\{\nu_1, \nu_2, \ldots, \nu_m\}$ be a multiset of binary strings, each of length $s$.
If there is a rapidly mixing Markov chain with distribution proportional to $\prod_{i\in[m]} f\left(H(\nu_i,\mu)\right)$,  
then RP=NP. 
\end{thm}

\begin{proof}
The proof for this theorem follows the same line of reasoning as the proof for Theorem~\ref{thm: FPAUS1}. However, it makes use of details in Theorem~\ref{NP concave up2} rather than Theorem~\ref{NP concave up}.
\end{proof}

When $f$ is a function with $\log f(x)$ is concave down, we examine the possibility of an FPRAS (Definition~\ref{defn:FPRAS}) for $Z(B,f(x))$. 
\begin{thm}
Fix a function $f(x): \mathbb{Z}^{+} \cup \{0\} \rightarrow [0,\infty)$  for which:
\begin{itemize}
\item $\log f(x)$ is strictly concave down,
\item the function values of $f$ can be computed in polynomial time, and 
\item  there exists $\epsilon>0$ such that for all but finitely many $n\in \mathbb{Z}$, $n\geq 2$, 
\[ \frac{f(n-2)[f(n+1)]^3[f(n+2)]^3 f(n+5)}{f(n-1)[f(n)]^3[f(n+3)]^3f(n+4)} < 1-\epsilon.\]
\end{itemize}
For arbitrary $m,s\in \mathbb{Z}^+$, let $S:=\{\nu_1, \nu_2, \ldots, \nu_m\}$ be a multiset of binary strings, each of length $s$.
If there is an FPRAS for calculating 
\[Z(B, f(x)) = \sum_{\mu\in \mathcal{M}} \prod_{i\in[m]} f(H(\nu_i, \mu)),\]
then RP=NP. 
\label{FPRAS1}
\end{thm}

\begin{proof}
Let $r$ be an integer so that $\left(\frac{1}{1-\epsilon}\right)^{r} > 2^{2n+2}$. In the proof of Theorem~\ref{thm: FPAUS1}, we created a new multiset of strings $\mathcal{D}(r)$. The set of medians $\mathcal{M}$ for $\mathcal{D}(r)$ is precisely $\{0,1\}^{2n}\times \{0\}^t$ and each median $\mu \in \M'_\Gamma$ which corresponds to a satisfying truth assignment for $\Gamma$ has 
\[\prod_{\eta \in \mathcal{D}(r)} {f(H(\eta, \mu))} = \left(\frac{1}{ h_3 }\right)^r.\] 
All other medians have
  \[\prod_{\eta \in \mathcal{D}(r)} {f(H(\eta, \mu))} < \left(\frac{1}{ h_2 }\right)^r.\] 
  Therefore, if $\Gamma$ has no satisfying assignments, 
  \[\sum_{\mu\in \mathcal{M}} \prod_{\eta \in \mathcal{D}(r)} {f(H(\eta, \mu))} < 2^{2n} \left(\frac{1}{ h_2 }\right)^r.\]
  If there is a satisfying assignment for $\Gamma$, then 
  \[\sum_{\mu\in \mathcal{M}} \prod_{\eta \in \mathcal{D}(r)} {f(H(\eta, \mu))} \geq  \left(\frac{1}{ h_3 }\right)^r.\]
  By the choice of $r$, we have the following inequality to relate the two quantities:
  \[\left( \frac{1}{h_2}  \right)^r 2^{2n+2} <  \left(\frac{1}{h_3}\right)^r.\]
  
 Now suppose that there is an FPRAS for $T:=\sum_{\mu\in \mathcal{M}} \prod_{\eta \in \mathcal{D}(r)} {f(H(\eta, \mu))}$. In other words, for any $\epsilon, \delta >0$, there is a randomized algorithm as 
described in Definition~\ref{defn:FPRAS} which outputs a quantity $\hat T$ such that 
\[P\left( \frac{T}{1+\epsilon} \leq \hat T \leq T(1+\epsilon) \right) \geq 1-\delta.\]
Consider the case when $\delta=\frac{1}{3}$ and $\epsilon = 1$. Therefore,
\[P\left( \frac{1}{2}T \leq \hat T \leq 2T \right) \geq \frac{2}{3}.\]
Therefore, if $\Gamma$ can be satisfied, then $T \geq  \left(\frac{1}{ h_3 }\right)^r$ and the probability that $\hat T$ is at least $\frac{1}{2}T =\frac{1}{2}\left(\frac{1}{ h_3 }\right)^r > 2^{2n+1}\left(\frac{1}{ h_2 }\right)^r$ is $\frac{2}{3}$. On the other hand, if $\Gamma$ cannot be satisfied, then $T< 2^{2n} \left(\frac{1}{ h_2 }\right)^r$ and the probability that $\hat T$ is at most $2T = 2^{2n+1}  \left(\frac{1}{ h_2 }\right)^r$ is $\frac{2}{3}$. Therefore, we have a BPP algorithm (Definition~\ref{defn: BPP}) for D3SAT. 
Because D3SAT is NP-complete, Papadimitriou's Theorem~\ref{thm: BPP} implies RP=NP. 
 \end{proof}

A similar result holds for functions $f(x)$ which satisfy the opposite inequality. We do not give a proof as it follows the same reasoning in the proof of Theorem~\ref{FPRAS1}.
\begin{thm}
Fix a function $f(x): \mathbb{Z}^{+} \cup \{0\} \rightarrow [0,\infty)$  for which:
\begin{itemize}
\item $\log f(x)$ is strictly concave down,
\item the function values of $f$ can be computed in polynomial time, and 
\item  there exists $\epsilon>0$ such that for all but finitely many $n\in \mathbb{Z}$, $n\geq 2$, 
\[ \frac{f(n-2)[f(n+1)]^3[f(n+2)]^3 f(n+5)}{f(n-1)[f(n)]^3[f(n+3)]^3f(n+4)} > 1+\epsilon.\]
\end{itemize}
For arbitrary $m,s\in \mathbb{Z}^+$, let $S:=\{\nu_1, \nu_2, \ldots, \nu_m\}$ be a multiset of binary strings, each of length $s$.
If there is an FPRAS for calculating 
\[Z(B,f(x))=\sum_{\mu\in \mathcal{M}} \prod_{i\in[m]} f(H(\nu_i, \mu)),\]
then RP=NP. 
\end{thm}

\section{Set-up for results on binary trees}
\label{sec:binary}

Previously, we explored the value $Z(B,x!)$ as it related to the number of ways to label a star phylogenetic tree. In this section, we divert our exploration to binary phylogenetic trees. First we give a precise definition of a binary tree.

\begin{defn}
A tree is a binary tree if it is rooted and every non-leaf vertex has exactly two children.
\end{defn}

Fix a multiset $B$ of $m$ binary strings from $\{0,1\}^n$. Also fix a binary tree $T$ with $m$ leaves. Label the leaves of $T$ with the strings from $B$ via the surjective function $\varphi: L(T) \rightarrow B$. The equivalent of a median in this setting is a labeling $\varphi': V(T) \rightarrow \{0,1\}^n$ which agrees with $\varphi$ on $L(T)$ and minimizes $\sum_{uv\in E(T)} H(\varphi'(u), \varphi'(v))$.  Such a vertex labeling $\varphi'$ is called a \emph{most parsimonious labeling}. Let $\mathcal{M}(T,\varphi)$ be the set of most parsimonious labelings which extend $\varphi$ for the binary tree $T$. 

As with the star trees, we will label each edge of $T$ with a scenario. Given a most parsimonious labeling $\varphi'$ for $V(T)$, a scenario for the edge $uv$ is a permutation of the bits in which $\varphi'(u)$ and $\varphi'(v)$ differ. 

A \emph{most parsimonious scenario} for a binary tree $T$ with leaf labeling $\varphi: L(T) \rightarrow B$ consists of a most parsimonious labeling $\varphi'$ of the vertices of $T$ and a scenario to label each edge of $T$. We desire to count the number of most parsimonious scenarios which is 
\[Z_{T,\varphi}(B,x!) := \sum_{\varphi'\in \mathcal{M}(T,\varphi)} \prod_{uv\in E(T)} H(\varphi'(u), \varphi'(v))!.\]

Formally, the problem statement is below:
\begin{defn}[\#Binary]
Given arbitrary integer $m \geq 2$, let $T$ be a binary tree with $m$ leaves. Let $B=\{\nu_i\}_{i=1}^m$ be an arbitrary multiset of binary strings and let $\varphi: L(T) \rightarrow B$ be a surjective function. Determine the value of 
$Z_{T,\varphi}(B,x!) .$
\end{defn}

The main result of Section~\ref{sec: bin result} is the theorem which states \#Binary is in \#P-complete. 
In this section, we develop several tools and algorithms which lay the foundation for our main theorem.

Let $\Gamma= c_1 \wedge c_2 \wedge \ldots \wedge c_{k}$ be a D3CNF with variables $\{v_1, v_2, \ldots, v_{n}\}$. Select new variables $\{w_1, w_2, \ldots, w_{n}\}$ which do not occur in $\Gamma$. 
For each $i\in [n]$, interpretting subscript $n+1$ as $1$, define the following D3CNF, 
\begin{align}
\Phi_i:=(v_i \vee w_i \vee v_{i+1}) \wedge (v_i \vee w_i \vee \overline{v_{i+1}}) \wedge (\overline{v_i} \vee \overline{w_i} \vee v_{i+1}) \wedge (\overline{v_i} \vee \overline{w_i} \vee \overline{v_{i+1}}). 
\label{def Phi}
\end{align}
Observe that $\Phi_i$ is equivalent to the ``exclusive or'' $(v_i \vee w_i) \wedge (\overline{v_i} \vee \overline{w_i})$.
Define
\begin{equation} \Psi(\Gamma) := \Gamma \wedge \bigwedge_{i=1}^{n} \Phi_i. \label{def Phi2} \end{equation}
Necessarily, if $\Gamma$ is a D3CNF then so is $\Psi(\Gamma)$.
\begin{lemma}
For $\Gamma$, an arbitrary D3CNF, it is \#P-complete to determine the number of satisfying truth assignments for $\Psi(\Gamma)$.
\label{lemma: xor}
\end{lemma}

\begin{proof}
We have already shown in Lemma~\ref{Lemma: Distinct} that \#D3SAT is in \#P-complete. So to prove this result, we will show that the satisfying truth assignments for $\Gamma$ and for $\Psi(\Gamma)$ are in one-to-one correspondence.

Any truth assignment which satisfies $\Psi(\Gamma)$, when restricted to $\{v_1, v_2, \ldots, v_{n}\}$  will necessarily satisfy $\Gamma$. 

For the other direction, recall that $\Phi_i$ is equivalent to the ``exclusive or'' for $v_i$ and $w_i$. 
Therefore, given a satisfying truth assignment for $\Gamma$, we can create a unique satisfying truth assignment for $\Psi(\Gamma)$ by assigning to each $w_i$ the opposite value of $v_i$. 
\end{proof}

Next we provide two different algorithms for creating most parsimonious labelings given a rooted binary tree and a leaf-labeling $\varphi: L(T) \rightarrow \{0,1\}^n$. If we restrict $\varphi(\ell)$ to a single coordinate $c$ for every leaf $\ell$, we obtain a labeling $\varphi_c: L(T) \rightarrow \{0,1\}$. Each algorithm presented below will consider leaf labels from the set $\{0,1\}$ and output a most parsimonious labeling $\varphi'_c: V(T) \rightarrow \{0,1\}$. Obtaining a most parsimonious labeling for each coordinate in this way, we combine these labelings to create a most parsimonious labeling $\varphi': V(T) \rightarrow \{0,1\}^n$ for $T$ and the original leaf-labeling $\varphi$. 

Let $T$ be a binary tree with root $\rho$ and let $\varphi: L(T)\rightarrow \{0,1\}$ be a labeling for the leaves. Let $\varphi': V(T) \rightarrow \{0,1\}$ be a most parsimonious labeling which extends $\varphi$. Because each vertex is labeled with a single bit, $H(\varphi'(u), \varphi'(v))\in \{0,1\}$ for any edge $uv$. By definition, the most parsimonious labeling $\varphi'$ minimizes the sum $\sum_{uv\in E(T)} H(\varphi'(u), \varphi'(v))$. Consequently, $\varphi'$ must minimize the number of edges $uv$ such that $\varphi'(u) \neq \varphi'(v)$.

First, we have Fitch's algorithm to find most parsimonious labelings. 
\begin{FitAlg} [\cite{fitch}] 
Let $T$ be a binary tree with root $\rho$ and leaf-labeling $\varphi: L(T)\rightarrow \{0,1\}$. The following algorithm, completed in two parts, will find a most parsimonious labeling $\varphi': V(T) \rightarrow \{0,1\}$ which extends $\varphi$. 

\begin{enumerate}[leftmargin=*, label=Part \arabic*:]
\item Define a function $B$ on the vertices of $T$ as follows: 
For each leaf $\ell$, set $B(\ell) :=\{ \varphi(\ell)\}.$ Extend this assignment to all vertices of $T$ by the following rule: For a vertex $u$ with children $v_1, v_2$ such that $B(v_1)$ and $B(v_2)$ have been defined, set
\begin{align}
 B(u) := \begin{cases} B(v_1) \cap B(v_2) & \text{ if }B(v_1) \cap B(v_2)\neq \emptyset,\\
				B(v_1) \cup B(v_2) & \text{ otherwise. }
		\end{cases}  \label{Fitch B} \end{align}  
\item Select a single element $\alpha \in B(\rho)$.
Define a function $\varphi'$ on the vertices of $T$ as follows:  Set $\varphi'(\rho):=\alpha$. Extend $\varphi'$ to $V(T)$ by the following rule: If $v$ is a child of $u$ and $\varphi'(u)$ is defined, then 
\begin{align}
\varphi'(v) := \begin{cases} \varphi'(u) & \text{ if } \varphi'(u) \in B(v),\\
				1-\varphi'(u) & \text{ if } \varphi'(u) \not\in B(v).
			\end{cases}
			\label{Fitch phi}
			\end{align}
 \end{enumerate}
The resulting $\varphi'$ is a most parsimonious labeling extending $\varphi$ and is called a Fitch solution. 
\end{FitAlg}
While Fitch solutions are most parsimonious labelings, there are cases when Fitch's algorithm finds some of the most parsimonious labelings but not all of them. However, Sankoff's algorithm, described below, will produce all most parsimonious labelings \parencite{erdos}. 

\begin{SankAlg}[\cite{erdos, sankoff}]
Let $T$ be a binary tree with root $\rho$ and leaf labeling $\varphi: L(T)\rightarrow \{0,1\}$. 
This algorithm is completed in two steps. 

\begin{enumerate}[leftmargin=*, label=Part \arabic*:]
\item Define functions $s_0$ and $s_1$ on the vertices of $T$ as follows:
First, for each leaf $\ell$, 
		\begin{align}
		s_0(\ell):=\begin{cases}
		0 & \text{if } \varphi(\ell)=0,\\
		\infty & \text{otherwise.} 
		\end{cases}
		\label{Sankoff s0}
		\end{align}
		\[s_1(\ell):=\begin{cases}
		0 & \text{if } \varphi(\ell)=1,\\
		\infty & \text{otherwise.} 
		\end{cases}\]
 Extend these functions recursively to all vertices by the following: If $v_0$ and $v_1$ are children of $u$ and $s_i(v_j)$ has been defined for all $i,j\in \{0,1\}$, then 
 	\begin{align}s_0(u) := \min\{s_0(v_0), s_1(v_0) + 1\} + \min\{s_0(v_1), s_1(v_1)+1\}, \label{Sank s0 gen} \end{align}
 	\begin{align}s_1(u) := \min\{s_0(v_0)+1,  s_1(v_0)\} + \min\{s_0(v_1)+1, s_1(v_1)\}.\label{Sank s1 gen}\end{align}
Note: For any $v\in V(T)$,  $s_i(v)$ counts the minimum number of edges, within the subtree containing $v$ and its descendants, that will witness a change if a most parsimonious labeling assigned label $i$ to vertex $v$.
A leaf will have $s_0(\ell)=\infty$ (or $s_1(\ell)=\infty$ if it is impossible for a most parsimonious labeling to label $\ell$ with a 0 (1), because most parsimonious labelings must agree with the original leaf label.   
	
\item For each $v\in V(T)$, select $\alpha_v \in \{0,1\}$. Define the function $\varphi'$ on the vertices of $T$ as follows: For root $\rho$, define 
\[\varphi'(\rho):=\begin{cases}
		0 & \text{if } s_0(\rho) < s_1(\rho),\\
		\alpha_\rho & \text{if } s_0(\rho) = s_1(\rho),\\
		1 & \text{if } s_0(\rho) > s_1(\rho).
		\end{cases}\]
Extend $\varphi'$ to $V(T)$ by the following rule: 
If $v$ is a child of $u$ and $\varphi'(u)$ is defined, then define $\varphi'(v)$ as follows:
		 If $\varphi'(u)=0$, then 
		 \begin{align}
		 \varphi'(v):=\begin{cases}
		0 & \text{if }  s_0(v) < s_1(v)+1,\\
		\alpha_v & \text{if } s_0(v) = s_1(v)+1,\\
		1 & \text{if }  s_0(v) > s_1(v)+1.
		\end{cases}
		\label{Sank phi0}
		\end{align}
		If $\varphi'(u)=1$, then 
		 \begin{align}
		 \varphi'(v):=\begin{cases}
		1 & \text{if }  s_1(v) < s_0(v)+1,\\
		\alpha_v & \text{if } s_1(v) = s_0(v)+1,\\
		0 & \text{if }  s_1(v) > s_0(v)+1.
		\end{cases}
		\label{Sank phi1}
		\end{align}
 \end{enumerate}
 The resulting $\varphi'$ is a most parsimonious labeling for $T$ extending $\varphi$ and is called a Sankoff solution.
\end{SankAlg}

The following lemma draws a connection between the solutions found from each algorithm. 

\begin{lemma}
Let $T$ be a binary tree with leaf-labeling $\varphi: L(T) \rightarrow \{0,1\}$. Suppose that, for each $u,v\in V(T)$ with $v$ a child of $u$, the function $B$ in Fitch's algorithm satisfies
\begin{align}
 B(v) &=\{0,1\} \Rightarrow B(u)=\{0,1\}.
 \label{Prop for Fitch}
\end{align}
Then for $T$ and $\varphi$, all Sankoff solutions are Fitch solutions. In other words, Fitch's algorithm finds all most parsimonious labelings.
\label{Lemma: Fitch}
\end{lemma}

In order to prove Lemma~\ref{Lemma: Fitch}, we first establish a series of claims (\ref{cl: B0} through \ref{equal for 01}) under the assumptions of Lemma~\ref{Lemma: Fitch}. 

\begin{claim}
For any non-leaf vertex $v$, if $B(v)=\{x\}$ for some $x\in \{0,1\}$ in Fitch's algorithm, then $s_0(v)=0$ and $s_1(v)=2$ in Sankoff's algorithm. 
\label{cl: B0}
\end{claim}
\begin{proof}
The proof proceeds by reverse induction on the distance from the root. 
For the base case, we consider those vertices whose children are both leaves. Let $v$ be such a vertex with children $v_\ell$ and $v_r$. By symmetry of the argument, assume $B(v) = \{0\}$. Then $B(v_\ell)=B(v_r) = \{0\}$ which only happens for leaves if $\varphi(v_\ell) = \varphi(v_r)=0$. By \eqref{Sankoff s0}, $s_0(v_\ell) =s_0(v_r)=0$ and $s_1(v_\ell)=s_1(v_r)=\infty$. As desired, \eqref{Sank s0 gen} implies $s_0(v) = 0$ and \eqref{Sank s1 gen} implies $s_1(v) =2$. 

For the inductive hypothesis, assume that each vertex $v$ of distance at least $d\geq 1$ from the root has  either $s_0(v)=0$ and $s_1(v)=2$ or  $s_0(v)=2$ and $s_1(v)=0$. 
Let $u$ be a vertex of distance $d-1$ from the root. Again, we assume $B(u)=\{0\}$ as the argument for the case when $B(u)=\{1\}$ is very similar. This vertex has two children, $u_\ell$ and $u_r$. There are three cases to consider. 
\begin{enumerate}
\item[(1)] If $u_\ell$ and  $u_r$ are leaves, then the argument in the base case gives $s_0(u) =0$ and $s_1(u)=2$ as desired. 
\item[(2)] If $u_\ell$ is a leaf and $u_r$ is not a leaf, then $s_0(u_r) =0$ and $s_1(u_r)=\infty$ and, by the inductive hypothesis $s_0(u_\ell) =0$ and $s_1(u_\ell)=2$. Therefore \eqref{Sank s0 gen} implies $s_0(u) =0 $ and \eqref{Sank s1 gen} implies $s_1(u) =2$. 
\item[(3)] If $u_\ell$ and $u_r$ are not leaves, then by the inductive hypothesis, $s_0(u_\ell) =s_0(u_r)=0$ and $s_1(u_\ell)=s_1(u_r)=2$. Again,  \eqref{Sank s0 gen} implies $s_0(u) =0 $ and \eqref{Sank s1 gen} implies $s_1(u) =2$. 
\end{enumerate}
This complete the proof of the claim. 
\end{proof}

\begin{claim}
For any vertex $v$ with $B(v)=\{0,1\}$ from Fitch's algorithm, we will have $s_0(v) = s_1(v)$ in Sankoff's algorithm.
\label{cl: B01} 
\end{claim}
\begin{proof}
This claim is also proven by induction on distance from the root where the base case examines those vertices with greatest distance from the root. 

For the base case, let $v$ be a vertex with $B(v)=\{0,1\}$ and none of its descendants $u$ have $B(u)=\{0,1\}$. For children $v_\ell$ and $v_r$ of $v$ we may assume $B(v_\ell)=\{0\}$ and $B(v_r)=\{1\}$ by \eqref{Fitch B}. By Claim~\ref{cl: B0}, 
\[s_0(v_\ell)=s_1(v_r)=0 \text{ and } s_0(v_r)=s_1(v_\ell)=2.\]
Therefore $s_0(v) = s_1(v)=1$. 

For the inductive hypothesis, suppose all vertices $u$ with $B(u)=\{0,1\}$ of distance at least $d\geq 1$ from the root have $s_0(u_\ell) = s_1(u_r)$. Let $v$ be a vertex at distance $d-1$ from the root with $B(v) = \{0,1\}$. There are three cases to consider: 
\begin{enumerate}
\item[(1)] If $v$ has a child $v_\ell$ with $B(v_\ell)=\{0\}$, then by \eqref{Fitch B} the other child $v_r$ must have $B(v_r)=\{1\}$  and we can use the argument in the base case to see $s_0(v_\ell) = s_1(v_r)$. 
\item[(2)] If $v$ has a child $v_\ell$ with $B(v_\ell)=\{1\}$, then by \eqref{Fitch B}, $v$ must have another child $v_r$ with $B(v_r)=\{0\}$. This puts us back in case 1. 
\item[(3)] If $v$ has a child $v_\ell$ with $B(v_\ell)=\{0,1\}$, then by  \eqref{Fitch B}, $v$ must have another child $v_r$ with $B(v_r)=\{0,1\}$. By the inductive hypothesis, $s_0(v_\ell) = s_1(v_\ell)$ and $s_0(v_r) = s_1(v_r)$. By \eqref{Sank s0 gen} and \eqref{Sank s1 gen}, $s_0(v) = s_1(v)$. 
\end{enumerate}
This completes the proof of the claim. 
\end{proof}

\begin{claim}
For any non-leaf vertex $v$ with $B(v)=\{i\}$ $(i\in \{0,1\})$, both Fitch's algorithm and Sankoff's algorithm will define $\varphi'(v)=i$. 
\label{equal for 0}
\end{claim}
\begin{proof}
In Fitch's algorithm, this is an immediate consequence of \eqref{Fitch phi}. 

Now consider Sankoff's algorithm. If $B(v)=\{0\}$ then, by Claim~\ref{cl: B0}, $s_0(v)=0$ and $s_1(v)=2$. Observe $s_0(v)< s_1(v)+1$ and $s_1(v)>s_0(v)+1$. Therefore $\varphi'(v)=0$ by \eqref{Sank phi0} and \eqref{Sank phi1}.
\end{proof}

\begin{claim}
Suppose $B(\rho)=\{0,1\}$. For any vertex $v$ with $B(v)=\{0,1\}$, if both algorithms set $\varphi'(\rho):=0$, then 
both Fitch's algorithm and Sankoff's algorithm will set $\varphi'(v)=0$. Likewise, if $\varphi'(\rho)=1$, then both algorithms will set $\varphi'(v)=1$. 
\label{equal for 01}
\end{claim}

\begin{proof}
For any vertex $v$ with $B(v)=\{0,1\}$, there is a path $\rho=u_0, u_1, \ldots,u_{t-1}, u_t=v$ of vertices such that $B(u_i)=\{0,1\}$ for each $i\in[t]$. It suffices to show that, in both algorithms, if $\varphi'(u_i)=0$ for some $0\leq i <t$, then $\varphi'(u_{i+1})=0$. 

In Part 2 of Fitch's algorithm, if $\varphi'(u_i)=0$ and $B(u_{i+1})=\{0,1\}$, then \eqref{Fitch B} implies $\varphi'(u_{i+1})=0$. 

In Sankoff's algorithm, if $\varphi'(u_i)=0$ and $B(u_{i+1})=\{0,1\}$, then by Claim~\ref{equal for 01}, $s_0(u_{i+1})=s_1(u_{i+1})$. Thus $s_0(u_{i+1})<s_1(u_{i+1})+1$. Since $\varphi'(u_i)=0$, \eqref{Sank phi0} implies $\varphi'(u_{i+1})=0$.

A similar argument can be used to show that if $\varphi'(\rho)=1$, then $\varphi'(v)=1$. 
Therefore Fitch's algorithm and Sankoff's algorithm will agreed on $\varphi'(v)$ if they agree on $\varphi'(\rho)$. 
\end{proof}

\begin{proof}[Proof of Lemma~\ref{Lemma: Fitch}]
In each algorithm, once $\varphi(\rho)$ has been set, the algorithm deterministically outputs a most parsimonious labeling of all vertices. Therefore, it suffices to prove that both algorithms have the same choices for labeling the root and both algorithms output the same most parsimonious labeling for the same choice for $\varphi'(\rho)$.  

If $B(\rho)=\{0\}$ or $B(\rho)=\{1\}$, then there is only one choice in Fitch's algorithm for $\varphi'(\rho)$. By Claim~\ref{cl: B0}, Sankoff's algorithm has the same determined value for $\varphi'(\rho)$. Further, all vertices $v\in V(T)$ will have either $B(v)=\{0\}$ or $B(v)=\{1\}$ by condition \eqref{Prop for Fitch} and Claim~\ref{equal for 0} completes the proof. 

If $B(\rho)=\{0,1\}$, then in Fitch's algorithm, there are two choices for $\varphi'(\rho)$. By Claim~\ref{cl: B01}, $s_0(\rho)=s_1(\rho)$ in Sankoff's algorithm, which means there are also two choices for $\varphi'(\rho)$.
Claim~\ref{equal for 01} implies that if we make the same choice for the root, both algorithms give the same most parsimonious labeling $\varphi'$.  

Sankoff's algorithm is guaranteed to find all most parsimonious labelings and the most parsimonious labelings from Fitch's algorithm coincide with those from Sankoff's algorithm, this implies that Fitch's algorithm finds all most parsimonious labelings. 
\end{proof}

As mentioned earlier, these algorithms are designed for a tree $T$ with leaf-labeling $\varphi: L(T) \rightarrow \{0,1\}$. However, given a tree $T$ with leaf-labeling $\phi: L(T) \rightarrow \{0,1\}^n$, we can restrict all strings to a single coordinate and run one of the above algorithms to find a most parsimonious labelings for $V(T)$ in that coordinate. Repeat this for each coordinate. The most parsimonious labelings found for each coordinate can then be combined into a most parsimonious labeling of $V(T)$ that extends $\phi$.  

\section{Complexity result for \#Binary}
\label{sec: bin result}
Here is our main result on binary trees.
\begin{thm}
\#Binary is \#P-complete.
\label{thm: binary} 
\end{thm}

\begin{proof}
A proof similar to that of Lemma~\ref{lemma: sharpP} shows that \#Binary is in \#P. To prove \#Binary is in \#P-hard, we provide a polynomial reduction from \#D3SAT. 

Fix a D3CNF, $\Gamma=\bigwedge_{i\in[k]} c_i$ with $k$ clauses and $n$ variables. Let 
\[\Psi(\Gamma) := \bigwedge_{i\in[k]} c_i \wedge \bigwedge_{i\in [n]} \Phi_i\] with $2n$ variables, $\{v_1, v_2, \ldots v_{n}, w_1, w_2, \ldots, w_n\}$, where each clause $c_i$ has three distinct literals from $\{v_i, \ov{v_i}: i\in [n]\}$, and $\Phi_i$ is the D3CNF in \eqref{def Phi} which guarantees that,  for each $i\in [n]$,  $v_i$ and $w_i$ have different truth values in a satisfying assignment.   
By Lemma~\ref{lemma: xor}, there is a bijection between the satisfying truth assignments for $\Gamma$ and the satisfying truth assignments for $\Psi(\Gamma)$. We will construct a binary tree $\mathcal{B}$ and define a labeling $\varphi$ of its leaves with binary strings in $B$ such that  the number of satisfying truth assignment for $\Psi(\Gamma)$ is directly computable from $Z_{T, \varphi}(B,x!)$, the number of most parsimonious scenarios for binary tree $T$ with leaf labeling $\varphi$. 
 
Each $\Phi_i$ has 4 clauses, so $\Psi(\Gamma)$ has $k+4n$ clauses. Assign the names $c_{k+1}, \ldots, c_{k+4n}$ to the $4n$ clauses of $\bigwedge_{i\in [n]} \Phi_i$.
Then
\[\Psi(\Gamma)=\bigwedge_{i\in[k+4n]} c_i.\] 
 
For each $i\in [k+4n]$, we define a binary tree $\mathcal{B}_i$ which encodes clause $c_i$. The final binary tree $\mathcal{B}$ will join $\mathcal{B}_1, \mathcal{B}_2, \ldots, \mathcal{B}_{k+4n}$ by a comb. For $t=148(16n^2 + 8kn)(k+4n)$, the leaf-labeling $\varphi: L(\mathcal{B})\rightarrow \{0,1\}^{2n+t}$,  will assign a binary string with coordinates $(x_1, y_1, \ldots, x_n, y_n, e_1, \ldots, e_t)$ to each leaf. The $x_i$ coordinates will correspond to the $v_i$ variables and the $y_i$ coordinates will correspond to the $w_i$ variables of $\Psi(\Gamma)$. The $e_i$ coordinates will be for additional ones, used in a manner similar to the additional ones in the previous sections for star trees. 

In this section and the next, we denote the left child of a non-leaf vertex $v$ by $v_\ell$ and the right child by $v_r$. The height of a vertex is its graph distance from the root. 
The construction of $\mathcal{B}_i$ with its leaf labeling $\varphi$ will come in Definition~\ref{def: Bi}, but first we need some preliminary definitions. 

For any clause $c = v_\alpha \vee v_\beta \vee v_\gamma$ which is the disjunction of 3 distinct literals, \textcite{miklos} defined a \textit{unit subtree}, $\mathcal{U}$, with 248 leaves. They also defined a leaf-labeling $\hat{\varphi}: L(\mathcal{U}) \rightarrow \{0,1\}^{151}$ where the binary strings in the range have coordinates $(\hat{x}_\alpha, \hat{x}_\beta, \hat{x}_\gamma, \hat{e}_1, \hat{e}_2, \ldots, \hat{e}_{148})$. The first three coordinates correspond to the variables in $c_i$ and the remaining 148 coordinates are for additional ones. 
This unit subtree has some useful properties which will be discussed after Definition~\ref{def: Bi}.

For each $i\in [k+4n]$, let $\mathcal{U}_i$ be the unit subtree for clause $c_i$.
If $i\leq k$ where $c_i$ relates $v_\alpha, v_\beta, v_\gamma$, then $\mathcal{U}_i$ will have leaf labels with coordinates $\{x_\alpha, x_\beta, x_\gamma\}$ and 148 coordinates for additional ones. If $i> k$ where $c_i$ relates variables $v_\alpha, w_\alpha, v_{\alpha+1}$, then $\mathcal{U}_i$ will have leaf labels with coordinates $\{x_\alpha, y_\alpha, x_{\alpha+1}\}$ and 148 coordinates for additional ones. 

\begin{defn}
The tree $\mathcal{T}_i$ in Step 5 of Definition~\ref{def: T_i} is a comb joining $16n^2 + 8kn$ copies of $\mathcal{U}_i$, as in Figure~\ref{comb}. 
\label{def C_i}
\end{defn}

\begin{figure}
\centering
\begin{tabular}{c}
\begin{tikzpicture}[every node/.style={scale=.8}]

\draw (3,2.2) [fill=black] circle(0.01);
\draw (2.85,2.1) [fill=black] circle(0.01);
\draw (2.7, 2) [fill=black] circle(0.01);

\draw (0,1)--(3.3,3.2);
\draw (.6,1.4)--(1.2,1);
\draw (1.5,2)--(2.1,1.6);
\draw (3.3,3.2)--(3.9,2.8);
\draw (-.42,0)--(0,1)--(.42,0);
\foreach \x in {1,2,3,4} \draw [fill=black]({.06-\x*.12},0) circle(0.03);
\foreach \x in {4} \draw [fill=black]({-.06+\x*.12},0) circle(0.03);
\foreach \x in {1.3,2,2.8} \draw [fill=black]({-.06+\x*.12},0) circle(0.01);
\draw ({0},.4) node {$\mathcal{U}_i$};

\begin{scope}[shift={+(1.2,0)}]
\draw (-.42,0)--(0,1)--(.42,0);
\foreach \x in {1,2,3,4} \draw [fill=black]({.06-\x*.12},0) circle(0.03);
\foreach \x in {4} \draw [fill=black]({-.06+\x*.12},0) circle(0.03);
\foreach \x in {1.3,2,2.8} \draw [fill=black]({-.06+\x*.12},0) circle(0.01);
\draw ({0},.4) node {$\mathcal{U}_i$};
\end{scope}

\begin{scope}[shift={+(2.1,.6)}]
\draw (-.42,0)--(0,1)--(.42,0);
\foreach \x in {1,2,3,4} \draw [fill=black]({.06-\x*.12},0) circle(0.03);
\foreach \x in {4} \draw [fill=black]({-.06+\x*.12},0) circle(0.03);
\foreach \x in {1.3,2,2.8} \draw [fill=black]({-.06+\x*.12},0) circle(0.01);
\draw ({0},.4) node {$\mathcal{U}_i$};
\end{scope}

\begin{scope}[shift={+(3.9,1.8)}]
\draw (-.42,0)--(0,1)--(.42,0);
\foreach \x in {1,2,3,4} \draw [fill=black]({.06-\x*.12},0) circle(0.03);
\foreach \x in {4} \draw [fill=black]({-.06+\x*.12},0) circle(0.03);
\foreach \x in {1.3,2,2.8} \draw [fill=black]({-.06+\x*.12},0) circle(0.01);
\draw ({0},.4) node {$\mathcal{U}_i$};
\end{scope}

\node at (5,0) {};
\node at (-1,0) {};
\end{tikzpicture}
\end{tabular}
\caption[Comb connecting copies of $\mathcal{U}_i$ to create $\mathcal{T}_i$.]{Comb connecting $16n^2 + 8kn$ copies of $\mathcal{U}_i$ to create $\mathcal{T}_i$}
\label{comb}
\end{figure}
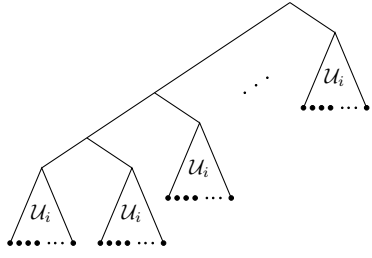

\begin{defn}
For three literals $a,b,c$, we define $S(a,b,c)$ to be the complete binary tree of height 3 with root $\rho$ with the vertices labeled with equations as follows: 
\begin{enumerate}
\item Assign the label ``$a=0$'' to vertex $\rho_\ell$ and ``$a=1$'' to $\rho_r$. 
\item For each vertex $u$ of height 1, assign the label ``$b=0$'' to $u_\ell$ and ``$b=1$'' to $u_r$. 
\item For each vertex $v$ of height 2, assign the label ``$c=0$'' to $v_\ell$ and ``$c=1$'' to $v_r$.
\end{enumerate}
This tree is pictured on the right in Figure~\ref{sorting tree}. We will use the representation on the left in place of $S(a,b,c)$ in future figures. 
\end{defn}

Next we construct $\hat{\mathcal{B}_i}$ which will have the same tree structure as the desired $\mathcal{B}_i$. However, $\hat{\mathcal{B}_i}$ will have all of its vertices labeled with equations while $\mathcal{B}_i$ will only have leaf labels which are binary strings. The leaf labeling of $\mathcal{B}_i$ will be induced by the vertex labels of $\hat{\mathcal{B}_i}$. Each leaf will essentially inherit the labels of its ancestors.

\begin{defn}
Fix $i\in [k+4n]$. Construct $\hat{\mathcal{B}_i}$, a binary tree with vertex labels, as follows.
\begin{enumerate}
\item[A.] If $i\in [k]$, then say clause $c_i$ has variables $v_\alpha,v_\beta,v_\gamma$. The construction of $\hat{\mathcal{B}_i}$ described below is drawn in Figure~\ref{Bin_Reduct_Tree}.
\begin{enumerate}
\item Draw a vertex $\rho^i$ with two children, $\rho^i_\ell$ and $\rho^i_r$. 
\item Label vertex $\rho^i_\ell$ with the equations ``$x_j=y_j=0$'' for each $j\in [n]\setminus\{\alpha, \beta, \gamma\}$. Label $\rho^i_r$ with ``$x_j=y_j=1$'' for all $j\in [n]\setminus\{\alpha, \beta, \gamma\}$.
\item From each of $\rho^i_\ell$ and $\rho^i_r$, hang a copy of $S(y_\alpha, y_\beta, y_\gamma)$. 
\item From each leaf of each copy of $S(y_\alpha, y_\beta, y_\gamma)$, hang a copy of $S(x_\alpha, x_\beta, x_\gamma)$.
\item Delete the left-most copy of $S(x_\alpha, x_\beta, x_\gamma)$, the one which hangs below the vertices with labels ``$y_{\alpha}=0$,'' ``$y_\beta=0$,'' ``$y_{\gamma}=0$,'' and with ancestor $\rho_{\ell}^i$, and replace it with a copy of the comb $\mathcal{T}_i$ from Definition~\ref{def C_i}.  
\end{enumerate}

\item[B.] If $i\in \{k+1, \ldots, k+4n\}$, then clause $c_i$ relates variables $v_\alpha, w_\alpha, v_{\alpha+1}$.
The construction of $\hat{\mathcal{B}_i}$ described below requires only a change of variables from the previous construction. 
\begin{enumerate}
\item Draw a vertex $\rho^i$ with two children, $\rho^i_\ell$ and $\rho^i_r$. 
\item Label $\rho^i_\ell$ with the system of equations ``$x_j=y_j=0$'' for all $j\in [n]\setminus\{\alpha, \alpha+1, \alpha+2\}$. Label $\rho^i_r$ with equations ``$x_j=y_j=1$'' for each $j\in [n]\setminus\{\alpha, \alpha+1, \alpha+2\}$.
\item From each of $\rho^i_\ell$ and $\rho^i_r$, hang a copy of $S(y_{\alpha+1}, x_{\alpha+2}, y_{\alpha+2})$. 
\item Hang a copy of $S(x_\alpha, y_\alpha, x_{\alpha+1})$ from each leaf of each and every copy of $S(y_{\alpha+1}, x_{\alpha+2}, y_{\alpha+2})$.
\item Delete the left-most copy of $S(x_\alpha, y_\alpha, x_{\alpha+1})$ and replace it with a copy of comb $\mathcal{T}_i$ from Definition~\ref{def C_i}.  
\end{enumerate}
\end{enumerate}
\label{def: T_i}
\end{defn}

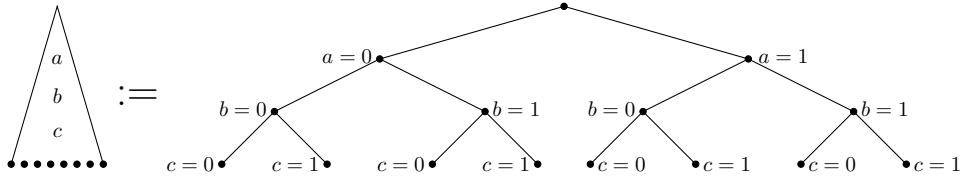
\begin{figure} 
\centering
\begin{tabular}{c}
\begin{tikzpicture}[scale=.7, every node/.style={scale=.8}]
\foreach \x in {0,2,4,6, 7,9,11,13} \draw [fill=black](\x,0) circle(0.06);
\foreach \x in {1,5,8,12} \draw [fill=black](\x,1) circle(0.06);
\foreach \x in {3,10} \draw [fill=black](\x,2) circle(0.06);
\draw [fill=black](6.5,3) circle(0.06);
\foreach \x in {0,4,7,11} \draw (\x,0)--(\x+1,1)--(\x+2,0);
\foreach \x in {1,8} \draw (\x,1)--(\x+2,2)--(\x+4,1);
\draw (3,2)--(6.5,3)--(10,2);
\foreach \x in {0,4} \draw (\x,0) node [left] {$c=0$};
\foreach \x in {2,6} \draw (\x,0) node [left] {$c=1$};
\foreach \x in {7,11} \draw (\x,0) node [right] {$c=0$};
\foreach \x in {9,13} \draw (\x,0) node [right] {$c=1$};
\foreach \x in {1,8} \draw (\x,1.05) node [left] {$b=0$};
\foreach \x in {5,12} \draw (\x,1.05) node [right] {$b=1$};
\draw (3,2.05) node [left] {$a=0$};
\draw (10.05,2.05) node [right] {$a=1$};

\begin{scope}[shift={+(-20,0)}]
\draw (16,0)--({16+1.75/2},3)--(17.75,0);
\foreach \x in {16,16.25,16.5, 16.75, 17, 17.25, 17.5, 17.75} \draw [fill=black](\x,0) circle(0.06);
\draw ({16+1.75/2},2) node {$a$};
\draw ({16+1.75/2},1.3) node {$b$};
\draw ({16+1.75/2},.6) node {$c$};
\draw (18.4,1.3) node {\Huge{$:=$}};
\end{scope}
\end{tikzpicture}
\end{tabular}
\caption{The labeled binary tree on the right is $S(a,b,c)$. The representation on the left will be used in place of $S(a,b,c)$ in future figures.}
\label{sorting tree}
\end{figure}

\begin{figure}
\centering
\begin{tabular}{c}
\begin{tikzpicture}[every node/.style={scale=0.85}]

\begin{scope}[shift={+(0,0)}]
\draw (-.42,0)--(0,1.5)--(.42,0);
\foreach \x in {1,2,3,4} \draw [fill=black]({.06-\x*.12},0) circle(0.03);
\foreach \x in {4} \draw [fill=black]({-.06+\x*.12},0) circle(0.03);
\foreach \x in {1.3,2,2.8} \draw [fill=black]({-.06+\x*.12},0) circle(0.01);
\draw ({0},.5) node {$\mathcal{T}_i$};
\end{scope}

\begin{scope}[shift={+(1.1,0)}]
\draw (-.42,0)--(0,1.5)--(.42,0);
\foreach \x in {1,2,3,4} \draw [fill=black]({.06-\x*.12},0) circle(0.03);
\foreach \x in {1,2,3,4} \draw [fill=black]({-.06+\x*.12},0) circle(0.03);
\draw ({0},.9) node {$x_1$};
\draw ({0},.6) node {$x_2$};
\draw ({0},.3) node {$x_3$};
\end{scope}

\begin{scope}[shift={+(3.5,0)}]
\draw (-.42,0)--(0,1.5)--(.42,0);
\foreach \x in {1,2,3,4} \draw [fill=black]({.06-\x*.12},0) circle(0.03);
\foreach \x in {1,2,3,4} \draw [fill=black]({-.06+\x*.12},0) circle(0.03);
\draw ({0},.9) node {$x_1$};
\draw ({0},.6) node {$x_2$};
\draw ({0},.3) node {$x_3$};
\end{scope}

\draw (2.3,.7) node {$\cdots$};

\draw (0,1.5)--(1.75,3)--(3.5,1.5);
\draw (1.75,1.2+1.5) node {$y_1$};
\draw (1.75,.8+1.5) node {$y_2$};
\draw (1.75,.4+1.5) node {$y_3$};

\foreach \x in {0,1.1,  1.5, 1.9  ,2.3, 2.7, 3.1 ,3.5} \draw [fill=black] (\x,1.5) circle (0.04);

\begin{scope}[shift={+(5,0)}]
\begin{scope}[shift={+(0,0)}]
\draw (-.42,0)--(0,1.5)--(.42,0);
\foreach \x in {1,2,3,4} \draw [fill=black]({.06-\x*.12},0) circle(0.03);
\foreach \x in {1,2,3,4} \draw [fill=black]({-.06+\x*.12},0) circle(0.03);
\draw ({0},.9) node {$x_1$};
\draw ({0},.6) node {$x_2$};
\draw ({0},.3) node {$x_3$};
\end{scope}

\begin{scope}[shift={+(1.1,0)}]
\draw (-.42,0)--(0,1.5)--(.42,0);
\foreach \x in {1,2,3,4} \draw [fill=black]({.06-\x*.12},0) circle(0.03);
\foreach \x in {1,2,3,4} \draw [fill=black]({-.06+\x*.12},0) circle(0.03);
\draw ({0},.9) node {$x_1$};
\draw ({0},.6) node {$x_2$};
\draw ({0},.3) node {$x_3$};
\end{scope}

\begin{scope}[shift={+(3.5,0)}]
\draw (-.42,0)--(0,1.5)--(.42,0);
\foreach \x in {1,2,3,4} \draw [fill=black]({.06-\x*.12},0) circle(0.03);
\foreach \x in {1,2,3,4} \draw [fill=black]({-.06+\x*.12},0) circle(0.03);
\draw ({0},.9) node {$x_1$};
\draw ({0},.6) node {$x_2$};
\draw ({0},.3) node {$x_3$};
\end{scope}

\draw (2.3,.7) node {$\cdots$};

\draw (0,1.5)--(1.75,3)--(3.5,1.5);
\draw (1.75,1.2+1.5) node {$y_1$};
\draw (1.75,.8+1.5) node {$y_2$};
\draw (1.75,.4+1.5) node {$y_3$};
\foreach \x in {0,1.1,  1.5, 1.9  ,2.3, 2.7, 3.1 ,3.5} \draw [fill=black] (\x,1.5) circle (0.04);
\end{scope}

\foreach \x in {1.75,6.75} \draw [fill=black] (\x,3) circle (0.04);
\draw (1.75,3)--(4.25,4)--(6.75,3);
\draw [fill=black] (4.25,4) circle (0.04);
\draw (1.95,2.9) node [above left] {\begin{tabular}{c}$x_j=y_j=0$\\ $\forall j\in [n], j\geq 4$ \end{tabular}};
\draw (6.55,2.9) node [above right] {\begin{tabular}{c}$x_j=y_j=1$\\ $\forall j\in [n], j\geq 4$ \end{tabular}};
\end{tikzpicture}
\end{tabular}
\caption{The binary tree $\hat{\mathcal{B}_i}$, for $i\in [k]$ created for clause $c_i=x_1\vee x_2 \vee x_3$.}
\label{Bin_Reduct_Tree}
\end{figure}
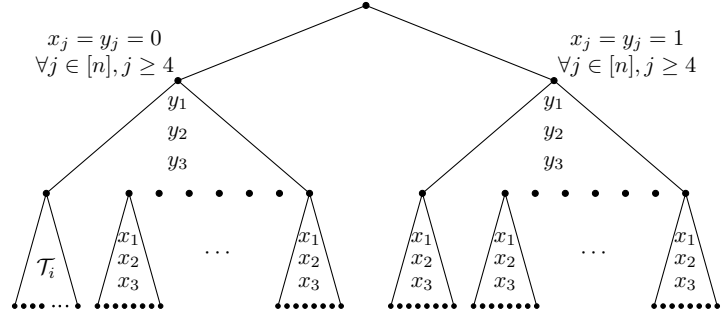

Recall $\mathcal{B}$ is a comb connecting binary trees $\mathcal{B}_i$. The binary tree $\B$ has a leaf labeling $\varphi: L(\mathcal{B}) \rightarrow \{0,1\}^{2n+t}$ where \[t=148(16n^2 + 8kn)(k+4n).\] Each leaf label will have coordinates $(x_1, y_1, \ldots, x_n, y_n, e_1, \ldots, e_t)$. In the next definition, we define $\B_i$ and values of $\varphi$ on the leaves of $\B_i$. 

\begin{defn}
For each $i\in [k+4n]$, the binary tree $\mathcal{B}_i$ will have the same tree structure as $\hat{\mathcal{B}_i}$. We only need to explain the labeling $\varphi: L(\mathcal{B}_i) \rightarrow \{0,1\}^{2n+t}$. 

Partition $[t]$ into classes $E_{ij}$ with $|E_{ij}|=148$ for each $i\in [k+4n]$ and each $j\in [16n^2 + 8kn]$.
Identify the set $E_{ij}$  with the $j^{th}$ copy of $\mathcal{U}_i$ in $\mathcal{T}_i$. Here we define $\varphi(\ell)$ for each leaf $\ell\in L(\mathcal{B}_i)$.

There are two cases: 
\begin{itemize}
\item If leaf $\ell$ is not in subtree $\mathcal{T}_i$, then $\varphi(\ell)[e_s]=0$ for all $s\in [t]$. The value of each $\varphi(\ell)[x_w]$ and $\varphi(\ell)[y_w]$ for $w\in[n]$ is inherited from the labels of the ancestors of $\ell$ as they appeared in $\hat{\mathcal{B}_i}$. 

\item If leaf $\ell$ is in the subtree $\mathcal{T}_i$ within $\hat{\mathcal{B}_i}$, then it is a leaf within the $j^{{th}}$ copy of unit subtree $\mathcal{U}_i$ for some $j\in [16n^2 + 8kn]$. 
Recall $\hat\varphi (\ell) \in \{0,1\}^{151}$.
Identify the coordinates $\{\hat e_1, \hat e_2, \ldots, \hat e_{148}\}$ with the indices in the 148 coordinates of $E_{ij}$ in any order. 
If $\hat e_s$ corresponds to coordinate $e_r$ for $r\in E_{ij}$, then we require $\varphi(\ell)[e_r] = \hat \varphi(\ell)[\hat e_s]$.
If $z$ is one of the three coordinates which correspond to a variable in $c_i$, we also require that $\varphi(\ell)[z]=\hat\varphi(\ell)[\hat z]$. Set $\varphi(\ell)[e_s]= 0$ for $s\not\in E_{ij}$. All other coordinates of $\varphi(\ell)$ will take the value 0 (the value inherited from the labeling of their ancestors in $\hat{\mathcal{B}_i}$). 
\end{itemize}
\label{def: Bi}
\end{defn}

Define $\varphi_{x_j}: L(\mathcal{B}) \rightarrow \{0,1\}$ so that $\varphi_{x_j}(\ell) = \varphi(\ell)[x_j] $. Define $\varphi_{y_j}$ and $\varphi_{e_j}$ similarly. 
We want to examine the Fitch solutions on $\B$ for each $\varphi_{x_j}$, $\varphi_{y_i}$, and $\varphi_{e_j}$. We will first prove that the conditions of Lemma~\ref{Lemma: Fitch} hold and thus Fitch's algorithm find all most parsimonious labelings for $\varphi$ on $\B$. 

We first explore the Fitch solutions for $\varphi_{e_j}$, $j\in [t]$.
\begin{fact}
Fix $j\in[t]$. There is only one $\ell\in L(\B)$ with $\varphi_{e_j}(\ell)=1$.  After running Part 1 of Fitch's algorithm, $B(\ell) =\{1\}$, the parent $v$ of $\ell$ has $B(v)=\{0,1\}$, and $B(u)=\{0\}$ for all other vertices. 
Consequently, Part 2 of Fitch's algorithm will output a most parsimonious labeling $\varphi'_{e_j}$ such that $\varphi'_{e_j}(\ell)=1$ and for all other vertices $u\in V(\mathcal{B})$, $\varphi'_{e_j}(u)=0$.
\label{Fitch ej}
\end{fact}
\begin{proof}
These values of $B$ follow directly from the description of $\varphi(\ell)[e_j]$, for leaf $\ell$, which was given in Definition~\ref{def: Bi}. The conclusion follows from the definition of $\varphi'$ \eqref{Fitch phi}.
\end{proof}

\begin{fact}
For $j\in [t]$, there is only one most parsimonious labeling $\varphi'_{e_j}$ which extends leaf labeling $\varphi_{e_j}$ of $\B$. 
\end{fact}
\begin{proof}
Recall that most parsimonious labelings minimize the sum of Hamming distances between adjacent vertices in the tree. The most parsimonious labeling obtained from Fitch's algorithm has \[\sum_{uv\in E(\B)} H(\varphi'_{e_j}(u), \varphi'_{e_j}(v)) =1.\] Because there is only one leaf $\ell$ with $\varphi_{e_j}(\ell)=1$, the $\varphi'_{e_j}$ obtained from Fitch's algorithm is the only extension of $\varphi_{e_j}$ with the sum of Hamming distances equal 1. 
\end{proof}

Fix $j\in [n]$. Next we consider the most parsimonious labelings for $\varphi_{x_j}$ on $\B$. The same arguments will hold for each $\varphi_{y_j}$. 

Run Part 1 of Fitch's algorithm on $\B$ with leaf labeling $\varphi_{x_j}$. For those clauses $c_i$ which contain variable $v_j$, we have the following result.

\begin{prop}[\cite{miklos}]
Fix a clause $c_i$. Suppose variable $v_j$ is in $c_i$ with coordinate $x_j$ corresponding to variable $v_j$.  Let $r^i$ be the root of unit subtree $\mathcal{U}_i$ for $c_i$.  Run Fitch's algorithm on $\mathcal{U}_i$ with leaf labeling $\varphi_{x_j}$. The following hold:
\begin{enumerate}
\item  $B(r^i)=\{0,1\}$.
\item For $u,v\in V(U_i)$, if $v$ is a child of $u$, then $B(v)=\{0,1\} \Rightarrow B(u)=\{0,1\}$.  
\end{enumerate}
\label{unit Fitch}
\end{prop}

In a single copy of $S(a,b,c)$, all vertices of the same distance from the root either have $B(v)\in \{\{0\}, \{1\}\}$ or all of them have $B(v) = \{0,1\}$. This fact together with Proposition~\ref{unit Fitch} implies that, when Fitch's algorithm is run on $\mathcal{B}$ with leaf-labeling $\varphi(x_i)$, for any $u,v\in V(T)$ with $v$ a child of $u$, \[B(v)=\{0,1\} \Rightarrow B(u)=\{0,1\}.\] 
With this result and the structure of each $S(a,b,c)$, 
by Lemma~\ref{Lemma: Fitch}, we can conclude that Fitch's algorithm finds all most parsimonious labelings of $\B$ that extend $\varphi_{x_j}$.
Further, $B(\rho)=\{0,1\}$ implies there are exactly two such most parsimonious labelings. 

As mentioned earlier, these results also hold for coordinate $y_i$. Fitch's algorithm finds the only two most parsimonious labelings that extend $\varphi_{y_j}$ on $\B$. 

For most parsimonious labeling $\varphi'$ that extends $\varphi$, on each $v\in V(T)$, notate $\varphi'(v)[x_j]$ by $\varphi'_{x_j}$. Likewise, define the notations $\varphi'_{y_j}$ and $\varphi'_{e_s}$.

\begin{lemma}
For leaf labeling $\varphi$ of $\B$, Fitch's algorithm finds all most parsimonious labelings. 
Each is characterized by the string it assigns to the root $\rho$ of $\mathcal{B}$ and there are precisely $2^{2n}$ most parsimonious labelings, one for each root label in $\{0,1\}^{2n}\times \{0\}^t$. 
\end{lemma}
\begin{proof}
Given a most parsimonious labeling $\varphi'$ that extends $\varphi$, each $\varphi'_{x_j}$, $\varphi'_{y_j}$, and $\varphi'_{e_s}$ is a most parsimonious labeling for that coordinate. So it suffices to first find all most parsimonious scenarios for the leaf labelings $\varphi_{x_j}, \varphi_{y_j}, \varphi_{e_s}$ for all $j\in [n]$ and $s\in [t]$ and take combinations of these labelings. 

We have already seen that Fitch's algorithm will find all most parsimonious labelings for $\varphi_{x_j}$ and $\varphi_{y_j}$, and there are exactly 2 of each. Fitch's algorithm will also find the one and only most parsimonious labeling for $\varphi_{e_s}$. Therefore, there are $2^{2n}$ most parsimonious labelings of $\B$ that extend $\varphi$. Part 2 of Fitch's algorithm shows that each most parsimonious labeling is characterized by the string it assigns to $\rho$. 
Since $B(\rho)=\{0,1\}$ for each $\varphi_{x_j}$ and $\varphi_{y_j}$ and $B(\rho)=\{0\}$ for each $\varphi_{e_s}$, the possible strings for $\varphi'(\rho)$ are $\{0,1\}^{2n} \times \{0\}^t$. 
\end{proof}

Set $\mathcal{M}:= \{0,1\}^{2n}\times \{0\}^t$. 

\begin{defn}
There is a bijection between $\mathcal{M}$ and the possible truth assignments for $\Psi(\Gamma)$. In particular, given any $\mu\in \mathcal{M}$, define a truth assignment for variables $\{v_i\}_{i=1}^n \cup \{w_i\}_{i=1}^n$ as follows: 
\begin{itemize}
\item For each $i\in[n]$, let $v_i$ be assigned the value true if $\mu[x_i]=1$ and false otherwise.
\item For each $i\in[n]$, let $w_i$ be assigned the value true if $\mu[y_i]=1$ and false otherwise.
\end{itemize}
Define $\mathcal{M}_{\Psi(\Gamma)}$ to be the set of $\mu\in \mathcal{M}$ which correspond to satisfying truth assignments for $\Psi(\Gamma)$. 
Likewise, for any $\Theta$, a clause or conjunction of clauses from ${\Psi(\Gamma)}$, define $\mathcal{M}_{\Theta}$ to be the set of $\mu\in \mathcal{M}$ which correspond to satisfying truth assignments for $\Theta$. 
\end{defn}

Now we know that each most parsimonious labeling of $\B$ extending $\varphi$ is found using Fitch's algorithm and is characterized by the binary string it assigns to the root. From here, we are interested in the number of scenarios admitted by each of these most parsimonious labelings. Ultimately, we wish to make a distinction between the binary strings in $\mathcal{M}_{\Psi(\Gamma)}$ and those in $\mathcal{M} \setminus \mathcal{M}_{\Psi(\Gamma)}$ by examining the number of scenarios admitted by the corresponding most parsimonious labeling. 

Let $\varphi'$ be a most parsimonious labeling for $\mathcal{B}$. The number of scenarios which are admitted by $\varphi'$ is precisely
\[\mathcal{H}(\varphi'(\rho)):=\prod_{uv\in E(\mathcal{B})} H(\varphi'(u), \varphi'(v))!  .\]
To calculate this, we partition the edges of $\B$ into 4 sets. 

First, consider the edges of the comb which connects $\{\B_i\}_{i=1}^{k+4n}$ to form $\mathcal{B}$. 
Part 2 of Fitch's algorithm will set $\varphi'(\rho) = \varphi'(\rho^{i})$ where $\rho^{i}$ is the root of $\B_i$. So the Hamming distance along each of these edges is 0. 

Next we look within each $\mathcal{B}_i$.

\begin{claim}
Set $\Phi:=\bigwedge \Phi_\iota$ as defined in \eqref{def Phi2}. For $i\in [k+4n]$, let $\rho^i$ be the root of $\B_i$ with children $\rho^i_{\ell}$ and $\rho^i_{r}$. Set $\eta:=\varphi'(\rho^i)$. 
If $\eta\in \M'_{\Phi}$, then 
\[H(\eta,\varphi'(\rho^i_\ell))=H(\eta,\varphi'(\rho^i_r))=n-3.\]
Otherwise 
\[(n-3)! ^2\leq H(\eta,\varphi'(\rho^i_\ell))!\cdot H(\eta,\varphi'(\rho^i_r))!\leq (2n-6)!0!.\]
\end{claim}

\begin{proof}
Suppose $\eta\in \M'_{\Phi}$. Then for each $j\in [n]$ considered in Step 2 of Definition~\ref{def: T_i} (there are $n-3$ such $j$) , if $\eta[x_j]=0$ then we have the following properties: 
\begin{itemize}
\item $\eta[y_j]=1$ because $\eta$ corresponds to a satisfying assignment for ${\Phi}$,
\item $0=\eta[x_j] \neq \varphi'(\rho^i_r)[x_j]=1$,
\item $1=\eta[y_j]\neq \varphi'(\rho^i_\ell)[y_j] =0$. 
\end{itemize}
On the other hand, if $\eta[x_j]=1$ then we have the following properties: 
\begin{itemize}
\item  $\eta[y_j]=0$ because $\eta$ corresponds to a satisfying assignment for ${\Phi}$,
\item $1=\eta[x_j] \neq \varphi'(\rho^i_{\ell})[x_j] =0$,
\item $0=\eta[y_j] \neq \varphi'(\rho^i_r)[y_j] =1$.
\end{itemize}
For each $s\in [t]$, $\eta[e_s] = \varphi'(\rho^i_\ell)[e_s] = \varphi'(\rho^i_r)[e_s] =0$. 
For each $j\in [n]$ which was not considered in Step 2 of Definition~\ref{def: T_i}, $\eta[x_j] = \varphi'(\rho^i_\ell)[x_j] = \varphi'(\rho^i_r)[x_j]$ and 
$\eta[y_j] = \varphi'(\rho^i_\ell)[y_j] = \varphi'(\rho^i_r)[y_j]$ because the $B$ values (from Fitch's algorithm) for these coordinates at these vertices will be $\{0,1\}$.
Thus 
\[H(\eta,\varphi'(\rho^i_\ell))=H(\eta,\varphi'(\rho^i_r))=n-3.\]

Alternatively, if $\eta\not\in \M_{\Phi}$, then $H(\eta,\varphi'(\rho^i_\ell))+H(\eta,\varphi'(\rho^i_r))=2n-6$ because each $x_i$ and each $y_i$ will contribute 1 to one of the Hamming distances. 
Using the convexity of the factorial, this establishes the last line of the claim.
\end{proof}

Based on the construction of $S(a,b,c)$, the Hamming distance $H(\varphi'(u), \varphi'(v))$ for each edge $uv$ in each copy of $S(a,b,c)$ is exactly 1.

The only piece remaining is $\mathcal{T}_i$. We make the following remarks for the clause $c_i = v_1 \vee v_2 \vee v_3$ to make the explanation easier. However, the arguments can be extended for any clause $c_i$ in $\Psi(\Gamma)$.

\begin{fact}
If $t^i$ is the root of $\mathcal{T}_i$ and $r^i$ is the root of one of the copies of $\mathcal{U}_i$ below $\mathcal{T}_i$, then running Fitch's algorithm for each coordinate, we find
\begin{itemize}
\item $B(t^i) = B(r^i) = \{0,1\}$ for each $x_i$, $i\in [3]$, by Proposition~\ref{unit Fitch}. 
\item $B(t^i) = B(r^i)= \{0\}$ for each $x_i$, $i\geq 4$, by the construction of $\mathcal{B}_i$.
\item $B(t^i)  = B(r^i) = \{0\}$ for each $y_i$, $i\in [n]$, by the construction of $\mathcal{B}_i$. 
\item  $B(t^i) = B(r^i) = \{0\}$ for each $e_s$, $s\in [t]$, because there is only one leaf $\ell\in L(\mathcal{B})$ with $\varphi(\ell)[e_s] = 1$. 
\end{itemize}
\end{fact}
Therefore, it is easy to see that, along the edges of the comb which connect the copies of $\mathcal{U}_i$, the Hamming distances will be 0.

Next we turn our attention to a single copy of $\mathcal{U}_i$, say the $j^{th}$ copy. 
\begin{fact}
Fix $\Gamma$ and build binary tree $\mathcal{B}$. Fix a most parsimonious labeling $\varphi'$ which extends leaf labeling $\varphi$. For clause $c_i = v_1 \vee v_2 \vee v_3$, we have the following characteristics for each $v\in \mathcal{U}_i$, 
\begin{itemize}
\item for $s\geq 4$, $\varphi'(v)[x_s] = 0$,
\item for $s\in [n]$,  $\varphi'(v)[y_s] = 0$,
\item for $s\not\in E_{ij}$, $\varphi'(v)[e_s] = 0$.
\end{itemize}
\label{mpl on U}
\end{fact}
Therefore, only the values of $\varphi'(v)$ on the coordinates $x_1, x_2, x_3$ and $e_s$ for $s\in E_{ij}$ will affect the Hamming distances along the edges in $\mathcal{U}_i$. These are precisely the 151 coordinates that appeared in the original labeling $\hat\varphi$ of the leaves of $\mathcal{U}_i$ given by \textcite{miklos}. For each $v\in V(\mathcal{U}_i)$, define $\hat\varphi'(v): V(\mathcal{U}_i) \rightarrow \{0,1\}^{151}$ to be the restriction of $\varphi'(v)$ to these 151 coordinates. In particular, $\hat\varphi'$ is a most parsimonious labeling on $\mathcal{U}_i$ which extends leaf labeling $\hat\varphi$. 

The following fact is a consequence of Fact~\ref{mpl on U}.
\begin{fact}
Let $r^i$ be the root of $\mathcal{U}_i$. If $\varphi'(r^i) = \hat\varphi'(r^i)$, then for each $uv\in E(\mathcal{U}_i)$, 
\[H(\varphi'(u), \varphi'(v)) = H(\hat\varphi'(u), \hat\varphi'(v)).\]
\end{fact}

As a result 
\[\prod_{uv\in \mathcal{U}_i} H(\varphi'(u), \varphi'(v))! = \prod_{uv\in \mathcal{U}_i} H(\hat\varphi'(u), \hat\varphi'(v))!.\]
This is calculated as follows: 
\begin{fact}[\cite{miklos}]
Fix $i\in [k+4n]$, the binary tree $\mathcal{U}_i$ with root $r^i$, and leaf-labeling $\hat\varphi$. Then for any most parsimonious labeling $\hat\varphi'$ which extends $\hat\varphi$:
\begin{enumerate}
\item If $\hat\varphi'(r^i)$ corresponds to a satisfying truth assignment for $c_i$, then 
\[\prod_{uv\in \mathcal{U}_i} H(\hat\varphi'(u), \hat\varphi'(v))! = 2^{156} \times 3^{64}.\] 
\item If $\hat\varphi'(r^i)$ corresponds to a  truth assignment which does not satisfy $c_i$, then \[\prod_{uv\in \mathcal{U}_i} H(\hat\varphi'(u), \hat\varphi'(v))! = 2^{136} \times 3^{76}.\]
\end{enumerate}
\end{fact}

Since $\varphi'(\rho) = \varphi'(r^i) = \hat\varphi'(r^i)$, $\varphi'(\rho)$ corresponds to a satisfying truth assignment for $c_i$ if and only if $\hat\varphi'(r^i)$ also corresponds to a satisfying truth assignment for $c_i$. 

As a result of the above discussion, we have proven the following claim. 
\begin{claim}
Fix $i\in [k+4n]$. If $\varphi'(\rho)$ corresponds to a satisfying truth assignment for clause $c_i$ and $\bigwedge_{\iota\in[n]} \Phi_\iota$, then 
\[\prod_{uv\in E(\B_i)} H(\varphi'(u), \varphi'(v))!
  = (n-3)!^2 \left(2^{156} \times 3^{64}\right)^{16n^2 + 8kn}.\]
If $\varphi'(\rho)$ corresponds to a  truth assignment which does not satisfy $c_i$, then 
\begin{align*}
(n-3)!^2 \left(2^{136} \times 3^{76}\right)^{16n^2 + 8kn} 
&\leq  \prod_{uv\in E(\B_i)} H(\varphi'(u), \varphi'(v))!  \\
 & \leq (2n-6)! \left(2^{136} \times 3^{76}\right)^{16n^2 + 8kn}.
 \end{align*}
If $\varphi'(\rho)$ corresponds to a  truth assignment which satisfies $c_i$ but does not satisfy $\bigwedge_{i\in[n]} \Phi_i$, then
\begin{align*}
(n-3)!^2 \left(2^{156} \times 3^{64}\right)^{16n^2 + 8kn} 
&<    \prod_{uv\in E(\B_i)} H(\varphi'(u), \varphi'(v))! \\
& \leq (2n-6)! \left(2^{156} \times 3^{64}\right)^{16n^2 + 8kn}.
\end{align*}
\end{claim}

Observe,
\begin{align*}
\frac{ (2n-6)!\left[2^{136} \times 3^{76}\right]^{16n^2+8kn}}{(n-3)!^2\left[2^{156}\times 3^{64}\right]^{16n^2+8kn}}
&= \left[\frac{ 3^{12} }{ 2^{20} }\right]^{16n^2+8kn} \binom{2n-6}{n-3} \\
&<  \left[\frac{ 3^{12} }{ 2^{20} }\right]^{16n^2+8kn} 2^{2n} \\
&<  \left[\frac{ 3^{12} }{ 2^{20} }\right]^{16n^2+8kn} 2^{2n+k} \\
&=\left[\frac{ 3^{12} }{ 2^{20-1/(8n) }}\right]^{16n^2+8kn} \\
&< \left[\frac{3^{12}}{2^{19.5}}\right]^{16n^2+8kn} \\
 &< 1.
\end{align*}

Consequently, 
 \begin{eqnarray*} (2n-6)! \left(2^{136} \times 3^{76}\right)^{16n^2+8kn} \,
& < \, (n-3)! ^2  \left(2^{156}\times 3^{64}\right)^{16n^2+8kn}\\
& < (2n-6)! \left(2^{156} \times 3^{64}\right)^{16n^2+8kn}.
  \label{scenario comp}\end{eqnarray*}

\begin{claim}
 If $\varphi'(\rho)$ corresponds to a satisfying truth assignment for $\Psi(\Gamma)$, then
 \[\mathcal{H}(\varphi'(\rho))
  = \left[ (n-3)!^2 \left(2^{136} \times 3^{76}\right)^{16n^2+8kn}\right] ^{k+4n}=: B_{good}.\]
If $\varphi'(\rho)$ corresponds to a  truth assignment which does not satisfy $\Psi(\Gamma)$, then there must be a clause $c_i$ for some $i\in [k+4n]$ which is not satisfied. Therefore, 
\begin{align*}
\mathcal{H}(\varphi'(\rho))
   \leq  &(2n-6)!\left(2^{136} \times 3^{76}\right)^{16n^2+8kn}\\
   &\cdot \left[(2n-6)!\left(2^{156}\times 3^{64}\right)^{16n^2+8kn}\right]^{k+4n-1}\\
   =: & B_{bad}.
   \end{align*}
\end{claim}

Define 
\[B_{total}: = \sum_{\varphi'}  \mathcal{H}(\varphi'(\rho))\]
which is the total number of most parsimonious scenarios for $\mathcal{B}$ which extend leaf labeling $\varphi$, as in Definition~\ref{def: SPSCJ}. 

 Given only $B_{total}$, we would like to determine the number of satisfying truth assignments, $|S|$, for $\Psi(\Gamma)$. 
 \begin{align*}
 B_{total} &= \sum_{\eta\in \M'_{\Psi(\Gamma)}} \mathcal{H}(\eta) + \sum_{\eta'\in \M\setminus \M'_{\Psi(\Gamma)}} \mathcal{H}(\eta') \\
 &= |S|B_{good} +  \sum_{\eta'\in \M\setminus \M'_{\Psi(\Gamma)}} \mathcal{H}(\eta').\end{align*}
 As long as $ \sum_{\eta'\in \M\setminus \M'_{\Psi(\Gamma)}} \mathcal{H}(\eta') < B_{good}$, we can conclude that the number of satisfying truth assignments for $\Psi(\Gamma)$ (and for $\Gamma$) is precisely 
 \[\left\lfloor \frac{B_{total}}{B_{good}} \right\rfloor .\]
Observe, for $n\geq 2$,
\begin{align*}
\frac{2^{2n}B_{bad}}{B_{good}}
&=2^{2n} \left[\frac{ 3^{12} }{ 2^{20} }\right]^{16n^2+8kn} \binom{2n-6}{n-3}^{k+4n} \\
&< 2^{2n} \left[\frac{ 3^{12} }{ 2^{20} }\right]^{16n^2+8kn} 2^{2n(k+4n)}\\
&< 2^{8n^2+2kn+2n} \left[\frac{ 3^{12} }{ 2^{20} }\right]^{16n^2+8kn}\\
&< 2^{8n^2+4kn} \left[\frac{ 3^{12} }{ 2^{20} }\right]^{16n^2+8kn}\\
&=\left[\frac{ 3^{12} }{ 2^{20-1/2 }}\right]^{16n^2+8kn}\\
& < 1.
\end{align*}
  Because there are only $2^{2n}$ truth assignments and $2^{2n}$ most parsimonious labelings, we obtain our desired result: 
  \begin{align*}
   \sum_{\eta'\in U} \mathcal{H}(\eta') 
   \leq \sum_{\eta'\in U} B_{bad}
   \leq 2^{2n} B_{bad}
   < B_{good}.
   \end{align*}

  Therefore, if we could determine the total number of most parsimonious scenarios for this binary tree in polynomial time, then we could obtain the total number of satisfying assignments for $\Psi(\Gamma)$ and for $\Gamma$ in polynomial time. This completes the proof. 
\end{proof}

\section{Conclusion}
We proved that it is \#P-complete to calculate the partition function $Z(B,x!)$. However, the existence of an FPAUS for this quantity has not yet been established. Following a number of results relating to calculating $Z(B,f(x))$ exactly for various functions $f(x)$, we where able to prove that, when $\log f(x)$ is strictly decreasing, under mild conditions an FPAUS exists for $Z(B,f(x))$ only if RP=NP. The question of approximating $Z(B,f(x))$ when $\log f(x)$ is strictly increasing remains unsettled. We concluded with a \#P-complete result for the extension of the partition function to binary trees, a natural extension to the bioinformatics interpretation of $Z(B,x!)$.

\section{Acknowledgements}
Heather Smith acknowledges support from DARPA and AFOSR under contract \#FA9550-12-1-0405, NSF DMS  contracts 1300547 and 1344199, and a SPARC Graduate Research Grant from the Office of the Vice President for Research at the University of South Carolina.

\printbibliography

\end{document}